\let\ManuscriptKernelLabel\label
\documentclass[aps,prx,reprint,onecolumn,notitlepage,11pt,
 superscriptaddress,longbibliography,nofootinbib]{revtex4-2}
\let\label\ManuscriptKernelLabel
\usepackage[T1]{fontenc}
\usepackage[utf8]{inputenc}
\usepackage{lmodern}

\usepackage{amsmath,amssymb,amsthm,mathtools,mathrsfs}
\usepackage{microtype,booktabs}
\usepackage{array}[=2016-10-06] 
\usepackage{tabularx}
\usepackage{placeins}
\usepackage{xcolor,graphicx}
\usepackage[colorlinks=true,linkcolor=blue!45!black,citecolor=blue!45!black,
 urlcolor=blue!45!black]{hyperref}
\hypersetup{pdftitle={Optimal limits on weak integrability breaking and protected thermal memory near qutrit exchange},
 pdfauthor={Hanbing Liang and Fujun Liu},pdflang={en}}
\newtheorem{theorem}{Theorem}
\newtheorem{corollary}[theorem]{Corollary}
\newtheorem{proposition}[theorem]{Proposition}
\newtheorem{lemma}[theorem]{Lemma}
\theoremstyle{remark}

\newcommand{\I}{\mathbb I}
\newcommand{\C}{\mathbb C}
\newcommand{\R}{\mathbb R}
\newcommand{\tr}{\operatorname{tr}}
\newcommand{\Tr}{\operatorname{Tr}}
\newcommand{\End}{\operatorname{End}}
\newcommand{\Span}{\operatorname{span}}
\newcommand{\Sym}{\operatorname{Sym}}
\newcommand{\Stab}{\operatorname{stab}}
\newcommand{\Tel}{\operatorname{Tel}}
\newcommand{\Hel}{\operatorname{Hel}}
\newcommand{\ad}{\operatorname{ad}}
\newcommand{\spec}{\operatorname{spec}}
\newcommand{\slthree}{\mathfrak{sl}_3}
\newcommand{\norm}[1]{\left\lVert#1\right\rVert}
\newcommand{\ket}[1]{|#1\rangle}
\newcommand{\bra}[1]{\langle#1|}
\newcommand{\HS}{\mathrm{HS}}
\begin{document}
\title{Optimal limits on weak integrability breaking and protected thermal memory near qutrit exchange}
\author{Hanbing Liang}
\affiliation{Nanophotonics and Biophotonics Key Laboratory of Jilin Province,
School of Physics, Changchun University of Science and Technology,
Changchun 130022, China}
\author{Fujun Liu}
\email[Contact author: ]{fjliu@cust.edu.cn}
\affiliation{Nanophotonics and Biophotonics Key Laboratory of Jilin Province,
School of Physics, Changchun University of Science and Technology,
Changchun 130022, China}
\date[]{}
\begin{abstract}
Although integrability does not universally require a continuous one-site symmetry, we rigorously prove that every jointly analytic, regular Yang-Baxter deformation of the qutrit exchange interaction necessarily retains a nontrivial, analytically varying one-site charge. Breaking this local symmetry imposes a fundamental physical constraint on approximate conservation, governed by the optimal uniform bound $\delta^3 \le C\varepsilon$ that explicitly relates the minimal one-site symmetry defect $\delta$ to the local current-conservation residual $\varepsilon$. While breaking all one-site charges strictly forbids an exact integrable completion, an optimally compensated nearest-neighbor interaction saturates this cubic limit and anomalously extends the guaranteed infinite-temperature energy-current correlation window to order $|\lambda|^{-3}$ in the perturbation strength $\lambda$. Furthermore, we reveal a fundamental resonance obstruction for intrinsic conversion perturbations that strictly prevents any exact first-order repair of a broken one-site charge on any finite ring. Nevertheless, we demonstrate that the complete eight-dimensional charge memory matrix remains thermodynamically protected and approaches the identity for timescales $t=o(|\lambda|^{-3/2})$, a robust feature of the full infinite-temperature dynamics when the thermodynamic limit is taken before weak coupling.
\end{abstract}

\maketitle
\section{Introduction}

While integrability can sometimes survive the total loss of continuous one-site symmetry, as seen in the anisotropic spin-$1/2$ XYZ chain~\cite{baxter}, such freedom is severely restricted in higher-spin systems. We prove that this unrestricted symmetry breaking is fundamentally absent for Hermitian qutrit densities that admit a jointly analytic, regular additive Yang-Baxter deformation through exchange. In this setting, integrability itself forcefully selects a continuous symmetry, ensuring that at least one nontrivial one-site charge survives with a constant spectrum. This necessity complements previous classifications that rely on conserved charges and vertex patterns as initial inputs~\cite{zf,ik,idzumi,cfr,boost,vieira,pt19,mansson,ffr}. Because exact local criteria relate current conservation directly to the commutation of the underlying integrable hierarchy~\cite{grabowski,hokkyo,ssi}, understanding how these symmetry constraints change under general perturbations is critical for determining the stability of integrable dynamics and macroscopic transport~\cite{ilievski,surace}.

The rigidity of this exact symmetry requirement provides a rigorous foundation for understanding weak integrability breaking. Constructive approaches to weak breaking typically utilize local, boost, and bilocal deformations to engineer higher-order approximate conservation laws~\cite{surace,vanovac,schouten}. However, a complementary physical question arises regarding what compensation can achieve when the interaction range remains strictly fixed. We establish a necessary restriction for the full space of nearest-neighbor interactions near the qutrit exchange point, proving that the least one-site symmetry defect $\delta$ and the local integrability residual $\varepsilon$ satisfy the optimal uniform bound $\delta^3\le C\varepsilon$. This cubic constraint reveals that suppressing the current-conservation residual necessarily forces at least one generator closer to exact conservation. The relation holds uniformly over all perturbation directions, dictating that if every one-site generator remains broken at first order, no higher-order two-site counterterms can reduce the local residual beyond cubic order.

This quantitative local bound translates directly into stringent limits on many-body dynamical observables. An exact periodic norm identity equates the infinite-temperature variance of the current force to a weighted norm of the local residual, which converts the cubic constraint into an optimal sixth-power lower bound on the initial current curvature. The necessity of this bound is invisible to linearized equations, where a cyclic pair-conversion tangent breaks all one-site symmetries and passes the first-order integrability test before failing a cubic obstruction. Nevertheless, this same tangent supports an optimally compensated nearest-neighbor interaction that successfully attains the cubic limit. Comparing this saturated model with the uncompensated conversion shows that while both break all one-site charges at the same leading order, the optimal compensation reduces the current curvature from fourth to sixth order in the perturbation strength $\lambda$. Spectral positivity subsequently extends the guaranteed current-memory window from inverse-square to inverse-cube scaling, thereby controlling the energy-spreading second moment within that extended timeframe.

Beyond local transport, the failure of exact charge repair on finite lattices does not strictly dictate the ultimate fate of thermal memory. Throughout the Hermitian intrinsic conversion sector, resonances fundamentally prevent every global first-order repair of a broken one-site charge on any finite ring. Despite this exact finite-ring obstruction, a truncated bilocal correction successfully controls the full perturbed dynamics to provide robust thermodynamic protection. For these intrinsic perturbations, this mechanism protects the complete eight-dimensional memory matrix of uniform one-site charges for timescales $t=o(|\lambda|^{-3/2})$ when the thermodynamic limit is taken before the weak-coupling limit. By combining the exact charge continuation and uniform stability bounds with these thermodynamic memory protections, we establish a comprehensive framework for weakly broken qutrit integrability. The remainder of this work details the exchange neighborhood and its conversion processes, proves the optimal stability bound, and ultimately connects the repair obstruction to thermal protection and complete charge memory in Section~\ref{main:sec:memory} and Section~\ref{main:sec:proof}.

\section{Exact Charge Conservation in Analytic Deformations}
\label{main:sec:exact}

\subsection{Exchange interactions, cyclic conversion, and the linear integrability test}
\label{main:sec:setting}

To establish the physical premise of our analysis, we define the local vector space $V=\mathbb{C}^3$ with orthonormal states $\ket0,\ket1,\ket2$, and introduce the exchange operator $P\ket{ab}=\ket{ba}$. Exchange moves these internal states between neighboring sites without changing their total populations. This operation trivially satisfies $[P,q\otimes\I+\I\otimes q]=0$ for every one-site operator $q$, meaning the exchange chain conserves every uniform charge $Q_L(q)=\sum_jq_j$. To answer whether an integrable deformation can remove all of these continuous one-site symmetries, we examine a concrete process that alters populations known as the cyclic pair conversion $T=K_{\rm F}=\sum_{\rm cyclic}\bigl[\ket{aa}(\bra{bc}-\bra{cb})+(\ket{bc}-\ket{cb})\bra{aa}\bigr]$, where the sum runs over the three cyclic permutations of $(0,1,2)$. Each term converts two equal states into the antisymmetric combination of the other two species, and reverses that process. A diagonal one-site charge would assign weights $q_0,q_1,q_2$ to the species, and conservation across all three channels strictly requires $2q_0=q_1+q_2$, $2q_1=q_2+q_0$, and $2q_2=q_0+q_1$. Only equal weights satisfy these constraints, resulting in the trivial identity charge. The full commutator equation additionally excludes every off-diagonal generator, demonstrating that this cyclic conversion breaks every nontrivial one-site charge. These specific vertices already occur in the Leigh-Strassler Hamiltonian~\cite{mansson}. 

Having identified a process that shatters all one-site symmetries, we must determine if such a mechanism can survive the rigorous constraints of quantum integrability. We therefore allow arbitrary higher-order Hermitian corrections to this cyclic process and investigate whether it can be completed into an exact integrable family. More generally, we consider an analytic Yang-Baxter density arc $h(\tau)=P+\tau K+O(\tau^2)$ with $h(\tau)^\dagger=h(\tau)$ for real $\tau$. This arc is associated with a matrix $\check R(\tau,u)$ that is jointly analytic near $(0,0)$ and satisfies $\check R(\tau,0)=\I$ alongside $h(\tau)=\partial_u\check R(\tau,0)$. Integrability demands the spectral equation $\check R_{12}(\tau,u)\check R_{23}(\tau,u+v)\check R_{12}(\tau,v)=\check R_{23}(\tau,v)\check R_{12}(\tau,u+v)\check R_{23}(\tau,u)$, where the model parameter $\tau$, the spectral parameters $u,v$, and the physical time are distinct. The exchange density possesses the elementary realization $\check R_0(u)=\I+uP$. We impose no predefined charge, matrix-entry pattern, or restriction on higher model derivatives. Differentiating the spectral equation yields the necessary Reshetikhin condition $F(h):=[h_{12}+h_{23},[h_{12},h_{23}]]=\Delta X$, where $\Delta X=X_{23}-X_{12}$~\cite{kulish,grabowski}. The linearized condition readily accepts the conversion tangent via $DF_P(T)=2\Delta T$, meaning that complete one-site symmetry breaking successfully passes the first-order integrability test. However, a severe mathematical obstruction universally appears at third order and survives every choice of higher derivatives, including corrections that leave the original vertex ansatz intact. This obstruction strictly excludes an exact analytic Yang-Baxter completion of $T$, yet this same direction admits a quadratic correction with only a cubic residual, directly connecting the exact algebraic obstruction to the physical realm of approximate conservation.

\subsection{Survival of an exact charge along every analytic deformation}

The failure of the cyclic conversion to form an integrable family hints at a deeper structural rigidity governing the system. Specifically, a continuous one-site symmetry supplies a fundamental local selection rule. In a basis that diagonalizes its generator, every allowed two-site process preserves the sum of the two charge weights. The cyclic conversion introduced above violates every nontrivial choice of weights, even after a common change of basis. The following theorem definitively proves that any valid analytic integrable deformation must retain at least one such rule. Because its generator can dynamically rotate as the coupling changes, the formal statement accurately follows the charge along the prescribed family.

\begin{theorem}[Exact charge along every analytic deformation]
\label{fs:thm:main}
Let \(h(\tau)\) satisfy the specified hypotheses. There are \(\epsilon>0\), a nonzero traceless Hermitian matrix \(q_0\), and a real-analytic common one-site unitary \(U(\tau)\), with \(U(0)=\I\), such that
\begin{equation}
q(\tau)=U(\tau)q_0U(\tau)^\dagger,\qquad
[h(\tau),q(\tau)\otimes\I+\I\otimes q(\tau)]=0 \qquad (|\tau|<\epsilon).
\label{fs:eq:exactcharge}
\end{equation}
The charge is nonzero at the base point and has constant spectrum. No nonvanishing first-derivative or simple-spectrum assumption is required.
\end{theorem}

The fixed eigenvalues generated by this theorem guarantee that the charge weights stay exactly the same while their one-site basis can change, and the identical unitary acts on every site. This powerful conclusion includes arbitrary higher coefficients and vanishing first derivatives without selecting one universal generator shared by all possible arcs. In particular, a common basis rotation beginning at the identity leaves $P$ unchanged but can dynamically rotate the surviving charge at higher orders. Two distinct conversion patterns emerge for arcs possessing a nonzero intrinsic leading interaction. In an analytic common basis, species exchange preserves two independent populations, whereas neutral-pair conversion preserves one equally spaced charge proportional to $\operatorname{diag}(1,-1,0)$. These distinct charges allow fifteen and nineteen matrix entries, respectively, with integrability further constraining their exact amplitudes. The nonlinear equations definitively select these specific patterns from the unrestricted tangent space, proving that integrability forcefully dictates the underlying symmetry.

\subsection{Macroscopic bounds on initial charge drift}

While the theorem establishes a dynamically rotating conserved quantity, a physical observer may prefer to measure the initial exchange charge rather than follow its rotating generator. Proximity to the exact conserved charge then rigorously bounds the macroscopic drift of that fixed observable at every physical time. We normalize the charge in Theorem~\ref{fs:thm:main} by setting $\norm{q_0}_{\rm op}=1$. For a periodic chain, we define the macroscopic quantities $H_L(\tau)=\sum_jh_{j,j+1}(\tau)$, $Q_L(\tau)=\sum_jq_j(\tau)$, and $Q_L^0=\sum_j(q_0)_j$. There exist finite arc-dependent constants $M_q,M_r$ such that
\begin{equation}
\frac1L\norm{Q_L^0(t)-Q_L^0}_{\rm op}
\le\min\{4M_r\tau^2|t|,\,2M_q|\tau|\}.
\label{main:eq:drift}
\end{equation}
This many-body bound is entirely independent of chain length and initial state. Its time-independent term follows directly from proximity to the exact analytic charge, providing a robust, volume-independent constraint on initial charge drift across the entire macroscopic system.

\section{Sharp Stability Bounds and Saturating Pair Conversions}
\label{main:sec:stability_merged}

\subsection{A uniform cubic bound for arbitrary nearby densities}
\label{main:sec:stability}

Moving away from exact integrability, all one-site generators can be systematically broken. The core quantitative question is how small the current-conservation residual can become while the interaction remains strictly nearest-neighbor. We compare two specific defects of a specified density, namely its least breaking of a normalized one-site generator and the local current residual after removing boundary differences. For any Hermitian two-site density, we define the symmetry defect $\delta(h)=\min \norm{[h,q\otimes\I+\I\otimes q]}_\HS$ operating over $q=q^\dagger$ with $\tr q=0$ and $\norm q_\HS=1$, alongside the local residual $r(h)=\pi_\Delta F(h)$ and its magnitude $\varepsilon(h)=\norm{r(h)}_\HS$. Local Hilbert-Schmidt norms utilize the unnormalized trace, and $\pi_\Delta$ represents the orthogonal projection away from the boundary differences $\Delta X=X_{23}-X_{12}$. The minimum formulation selects the least broken one-site generator without choosing its basis in advance. The residual measures the explicit failure of the local current-conservation condition rather than serving as a mere distance to an assumed manifold of integrable models.

This defect is an intrinsic property of the specified two-site density. For a general representative, $\delta>0$ need not exclude a periodic uniform charge because a one-site difference $A\otimes\I-\I\otimes A$ inherently cancels in $H_L$. If the density is balanced such that $\Tr_1h=\Tr_2h$, the periodic conservation of $Q_L(q)$ is equivalent to $[h,q_1+q_2]=0$ for $L\ge3$. The kernel of two-site periodization consists of these exact differences. Because their two partial traces are opposite, whereas those of $[h,q_1+q_2]$ are equal for balanced $h$, a vanishing periodic sum forces the local commutator to vanish completely. To separate measurement units from proximity to exchange, we write the orthogonal projection of $h$ onto $\Span\{\I,P\}$ as $a\I+bP$, and set $h_n=(h-a\I)/b$ for $b\ne0$. This scaling yields $\delta(h)=|b|\delta(h_n)$ and $\varepsilon(h)=|b|^3\varepsilon(h_n)$.

\begin{theorem}[Sharp uniform local exponent]\label{st:thm:sharp}
There are $r_0,C>0$ such that every Hermitian qutrit density with $b\ne0$ and $\norm{h_n-P}_\HS<r_0$ satisfies the optimal bound $\delta(h)^3\le C\varepsilon(h)$.
The exponent $1/3$ in $\delta\le C^{1/3}\varepsilon^{1/3}$ cannot be improved. The density need not be integrable or belong to a prescribed analytic arc.
\end{theorem}

This bound fundamentally limits what physical compensation can achieve. If a nearby family breaks every normalized one-site generator at least linearly in its perturbation strength, ensuring $\delta(h_\lambda)\ge d_0|\lambda|$ for a fixed $d_0>0$, then the residual must satisfy $\varepsilon(h_\lambda)\ge (d_0^3/C)|\lambda|^3$. No higher-order two-site counterterms can make the residual $o(|\lambda|^3)$ while preserving this amount of symmetry breaking. Every individual density with zero residual in this neighborhood must possess a nontrivial continuous one-site symmetry. While this fixed-density conclusion requires no jointly analytic model path, a spectral realization still demands the regularity and difference-form requirements specified previously.

\subsection{The unavoidable cost in energy-current curvature}

The local residual has a direct and profound meaning for infinite-temperature current fluctuations. Their autocorrelation measures the extent to which an initial energy-current fluctuation overlaps with the macroscopic current at a later time. Its initial curvature is precisely the variance per site of the current force $i[H_L,J_L]$. Fixing the density representative and using $\hbar=1$ on a periodic chain, we define $H_L=\sum_jh_{j,j+1}$, $J_L=i\sum_j[h_{j,j+1},h_{j+1,j+2}]$, and the correlation $C_J(t)=\Tr[J_L(t)J_L]/(L3^L)$. We write $\chi_J=C_J(0)$, $\kappa_J=-C_J''(0)$, and $\Gamma_J=\kappa_J/\chi_J$. Complete cancellation of four-site commutators yields $i[H_L,J_L]=-\sum_jr_{j,j+1,j+2}$~\cite{surace}. An orthogonal operator-word identity then establishes $\kappa_J=(3\norm{r_1}_\HS^2+2\norm{r_2}_\HS^2+\norm{r_3}_\HS^2)/27$, leading to the rigorous bounds $\varepsilon^2/27\le\kappa_J\le\varepsilon^2/9$ for every $L\ge5$. Here $r_s$ groups residual terms by their minimum support span in the identity/traceless tensor-word basis, and the weights account for overlapping copies when these local terms are summed around the ring. Thus, a nonzero local residual cannot disappear from the current-force variance in a long chain. Theorem~\ref{st:thm:sharp} therefore implies an optimal sixth-power bound:
\begin{equation}
\kappa_J(h)\ge\frac{\delta(h)^6}{27C^2}.
\label{main:eq:sixth}
\end{equation}
The susceptibility equals $16/9$ at $P$ and stays bounded above and away from zero in a smaller normalized neighborhood. At a fixed exchange scale, $\Gamma_J\ge c\delta^6$, meaning that at a fixed local symmetry defect, the current correlation possesses an unavoidable initial curvature. For densities with scalar partial traces where $\Tr_1h=\Tr_2h=c_0\I$, the one-site response is equally direct. We find $\Gamma_{\rm one}:=\min_q[-C_q''(0)/C_q(0)]=\delta(h)^2/3\le C^{2/3}\kappa_J^{1/3}$, utilizing the same traceless normalized $q$ defining the defect and the macroscopic charge correlation $C_q(t)$. The scalar-partial-trace assumption ensures that the translated charge forces are orthogonal, and the representative of $h$ remains fixed throughout to maintain consistency between the local defect and current curvature.

\subsection{Optimal compensation in the Fermat interaction}
\label{main:sec:model}

The forbidden exact tangent $T$ becomes the optimal near-integrable example after applying a highly specific structural compensation. We define the compensated Fermat interaction as $h_\lambda=P+\lambda T+\lambda^2N$, where the correction operator is $N=\I+2P-5D-C_+$, with $D=\sum_a\ket{aa}\bra{aa}$ and $C_+=\sum_{\rm cyc}\bigl[\ket{aa}(\bra{bc}+\bra{cb})+(\ket{bc}+\ket{cb})\bra{aa}\bigr]$. Its added term carefully adjusts exchange, diagonal energies, and symmetric pair-conversion amplitudes at second order. These contributions perfectly cancel the quadratic current force, while the first-order breaking of one-site charges rigidly remains. The residual then initiates at the unavoidable third order, demonstrating an additional property beyond the previous occurrence of the conversion vertices~\cite{mansson}.

Removing only the quadratic compensation yields the uncompensated raw path $h_\lambda^{\rm raw}=P+\lambda T$. Both paths break every one-site charge at the exact same leading order, but their current forces differ drastically. The raw force starts at second order, whereas the compensated force starts at third order. The raw path produces $\delta_{\rm raw}=\sqrt{12}|\lambda|$ alongside a raw curvature $\Gamma_J^{\rm raw}=22\lambda^4+O(\lambda^6)$. In stark contrast, the compensated path yields $\Gamma_J=486\lambda^6+O(\lambda^8)$. Spectral positivity translates this dramatic reduction in force into a significantly longer guaranteed memory window. The half-retention bounds subsequently obey $t_{\rm cert}^{\rm raw}\sim|\lambda|^{-2}/\sqrt{22}$ versus the extended $t_{\rm cert}\sim|\lambda|^{-3}/\sqrt{486}$ for the compensated interaction. The compensation thus improves the weak-coupling scaling without restoring any one-site charge. Figure~\ref{main:fig:sharp}(b) illustrates this exact curvature comparison, highlighting an asymptotic improvement as $\lambda\to0$.

\begin{figure}[t]
\centering
\includegraphics[width=\linewidth]{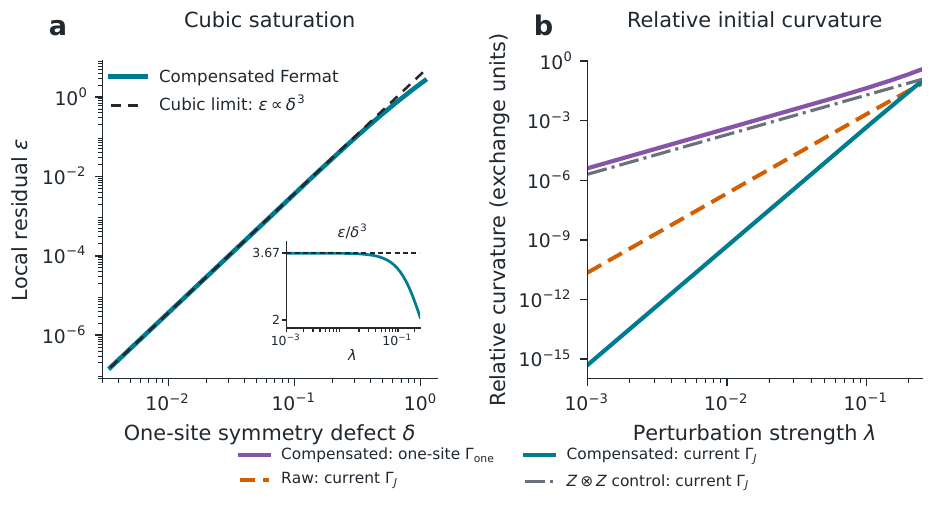}
\caption{Sharp local stability and separated initial responses.
(a) Exact defects of the compensated Fermat interaction for
$10^{-3}\le\lambda\le0.25$. The dashed curve is the asymptote
$\varepsilon=(3\sqrt6/2)\delta^3$; its coefficient is specific to this
family, not the unspecified universal constant $C$.
The inset shows the ratio \(\varepsilon/\delta^3\) approaching that
coefficient as \(\lambda\to0\).
(b) Relative initial curvatures. The compensated Fermat model breaks every one-site
charge with $\Gamma_{\rm one}=4\lambda^2+36\lambda^4$, while its
current curvature starts at $486\lambda^6$. Removing the quadratic
compensation, while keeping the same tangent, gives the raw current
curvature starting at $22\lambda^4$. The control
$P+\lambda Z\otimes Z$, $Z=\operatorname{diag}(1,-1,0)$, retains
two diagonal charges but has current curvature starting at $2\lambda^2$.
All curves evaluate exact local expressions and hold for every $L\ge5$;
they have no statistical error bars. AI-assisted plotting code and source
data are available in Ref.~\cite{code}; see
\hyperref[main:sec:computational]{Computational methods and AI assistance}.}
\label{main:fig:sharp}
\end{figure}

The exact residual and susceptibility polynomials hold for every $L\ge5$ and generate the analytical curves presented in Fig.~\ref{main:fig:sharp}. The residual is demonstrably nonzero at every $\lambda\ne0$, completely excluding a regular additive spectral realization for that density. Every nontrivial one-site charge is concurrently broken at every $\lambda\ne0$. For small coupling $\lambda^2<1/3$, we explicitly find $\delta=|\lambda|\sqrt{12+108\lambda^2}$, $\Gamma_{\rm one}=4\lambda^2+36\lambda^4$, and $\Gamma_J=486\lambda^6+O(\lambda^8)$. Therefore, the ratios $\varepsilon/\delta^3\to3\sqrt6/2$ and $\kappa_J/\delta^6\to1/2$ emerge naturally. Both the cube-root symmetry bound and the sixth-power current bound are perfectly saturated in exponent, cementing the optimally compensated nature of the Fermat model.

\subsection{Current retention and energy spreading}
\label{main:sec:retention}

Small force fluctuations fundamentally control a finite interval of the full macroscopic dynamics. At infinite temperature, the current autocorrelation operates as a positive weighted sum of cosines of energy differences. Utilizing the elementary bound $0\le1-\cos x\le x^2/2$, we derive $0\le1-C_J(t)/\chi_J\le\Gamma_J t^2/2$. For the optimally compensated explicit model, at least half of the initial correlation rigidly remains throughout the guaranteed timeframe $|t|\le t_{\rm cert}:=\Gamma_J^{-1/2}\sim|\lambda|^{-3}/\sqrt{486}$. This guaranteed window is strictly uniform in volume, deriving its cubic scale directly from the explicit upper bound on the current force alongside spectral positivity. A general lower bound alone would not mathematically imply such a precise persistence window.

Energy conservation intimately connects this current memory to the spatial spread of an associated energy perturbation. The ensuing bound controls the deviation of its second moment from quadratic growth throughout the guaranteed current-memory window. On the infinite chain, we define the centered density $e=h_\lambda-\I/3$, the susceptibility $\chi_E=\frac29(4+12\lambda^2+27\lambda^4)$, and the normalized second moment $M_2(t)=\chi_E^{-1}\sum_xx^2\langle e_x(t)e_0\rangle_0$. With $j_x=i[h_{x,x+1},h_{x+1,x+2}]$, the sum $C_J^\infty(t)=\sum_x\langle j_x(t)j_0\rangle_0$ serves as the fixed-time thermodynamic limit. Finite-range locality bounds guarantee that this sum and the energy second moment are rigorously convergent at every fixed time~\cite{bravyi}. The scalar marginals give $M_2(0)=0$, and tracial spectral symmetry confirms $M_2'(0)=0$. The continuity equation enforces $M_2''(t)=2C_J^\infty(t)/\chi_E$, allowing us to integrate the bounds and find:
\begin{equation}
\frac{\chi_J}{\chi_E}t^2-\frac{\kappa_J}{12\chi_E}t^4 \le M_2(t)\le\frac{\chi_J}{\chi_E}t^2.
\label{main:eq:spreading}
\end{equation}
Throughout the guaranteed window, the lower bound consistently remains at least $11/12$ of the ballistic expression $(\chi_J/\chi_E)t^2$. For analytical comparison, the symmetry-preserving control $P+\lambda Z\otimes Z$ with $Z=\operatorname{diag}(1,-1,0)$ preserves two diagonal charges but produces a significantly faster curvature $\Gamma_J^{(Z)}=2\lambda^2+O(\lambda^4)$. The same spectral argument guarantees only a window of order $|\lambda|^{-1}$ for this control. Consequently, a retained simple symmetry does not inherently entail the powerful cubic current-retention window exhibited by the optimally compensated symmetry-breaking construction.

\section{Resonant Obstructions and Thermodynamic Thermal Protection}
\label{main:sec:resonance_memory_merged}

The current in the optimally compensated Fermat model exhibits an $O(\lambda^3)$ force, while every simple one-site charge retains an $O(\lambda)$ force. Attempting to repair a broken charge by adding a many-body correction fails intrinsically at first order due to a fundamental physical obstruction. This failure does not depend on the nonlinear classification but originates directly from the degeneracy of the exchange energies. A many-body correction $Q_1$ contributes matrix elements $[H_0,Q_1]_{mn}=(E_m-E_n)(Q_1)_{mn}$ in an energy eigenbasis, meaning this contribution strictly vanishes within any degenerate energy block. Consequently, a charge force acting inside such a block cannot be canceled by simply redefining the charge. We demonstrate that the intrinsic conversion sector universally possesses this resonant component whenever it breaks the one-site generator, connecting the failure of exact local repair to the broader thermodynamic protection of thermal memory.

\subsection{The failure of exact first-order repair}
\label{main:sec:resonance}

Extending the cyclic example to the broader intrinsic conversion sector, we find that these first-order interactions couple symmetric and antisymmetric pair states. This sector is parametrized by a symmetric complex tensor $f^{abc}=f^{(abc)}$, defining the specific interaction $(T_f)^{ab}{}_{ij}=\epsilon_{ijm}f^{abm}+\epsilon^{abm}\bar f_{ijm}$ with $\epsilon_{012}=\epsilon^{012}=1$. The tangent $T$ features $f^{000}=f^{111}=f^{222}=1$ with all other entries strictly zero. For an arbitrary tensor $f$, we define $H_{0,L}=\sum_jP_{j,j+1}$ and $V_{f,L}=\sum_j(T_f)_{j,j+1}$. A formal repair $Q_L(q)+\lambda Q_1$ would subsequently require $[H_{0,L},Q_1]+[V_{f,L},Q_L(q)]=0$, where $q_1+q_2$ denotes $q\otimes\I+\I\otimes q$ acting on two adjacent sites.

\begin{theorem}[Resonant obstruction]\label{st:thm:resonance}
For every $L\ge3$, traceless Hermitian $q$, and allowed momentum $k=2\pi m/L$,
\begin{equation}
\inf_{Q_1}\norm{[H_{0,L},Q_1]+[V_{f,L},Q_L(q)]}_{\rm op}
\ge\frac{|\sin k|}{\sqrt{15}}
\norm{[T_f,q_1+q_2]}_\HS.\label{st:eq:resonance}
\end{equation}
The infimum includes arbitrary operators on the whole chain. Thus the first-order repair equation has a solution if and only if $[T_f,q_1+q_2]=0$.
\end{theorem}

No restriction of locality, range, or translation invariance is imposed on this attempted correction. For the specific Fermat tensor, every traceless Hermitian $q$ with $\norm q_\HS=1$ reliably satisfies the bound $\inf_{Q_1}\norm{[H_{0,L},Q_1]+[V_{f,L},Q_L(q)]}_{\rm op}\ge\sqrt{3/5}$ for any length $L\ge3$. Thus, no second-order Hamiltonian term, including the compensation $N$, can repair any of these charges through first order. A fixed finite-range bulk repair would inherently periodize to all sufficiently long rings and is likewise unequivocally excluded. For an unperturbed Hamiltonian $H_0=\sum_EEP_E$, the optimal first-order repair leaves exactly the degenerate-block part of a force $A$, mathematically expressed as $\inf_X\norm{[H_0,X]+A}_{\rm op}=\norm{\mathcal P(A)}_{\rm op}$ with the pinching operator $\mathcal P(A)=\sum_EP_EAP_E$. Off-diagonal energy blocks can be systematically canceled by division by their energy differences, whereas the pinching operation is contractive and robustly annihilates every commutator with $H_0$. By taking a uniform reference state and changing one site to either of the other two internal states, momentum superpositions yield two one-magnon flavors with the identical exact exchange energy $L-2+2\cos k$. This degeneracy securely persists after any common $SU(3)$ rotation, and matrix elements between these degenerate states reliably detect the resonant force.

\subsection{Local first-order repair of the invariant hierarchy}

While one-site charges rigidly resist exact repair, the exchange hierarchy also contains conserved densities that are invariant under a common $SU(3)$ rotation. Their physical response differs substantially because this additional invariance enables a bilocal deformation to act locally on each density. Consequently, the presence of broken one-site charges does not strictly determine whether the broader invariant hierarchy can be repaired. Setting $r_f=-iT_fP/2$ for an interval $I$, the bilocal operator $X_I=\sum_{i<j\in I}(r_f)_{ij}$ possesses a local commutator with every $SU(3)$-invariant density supported in $I$.

\begin{proposition}[First-order separation of charge repair]\label{st:prop:scalarrepair}
Let $C_\alpha^{(0)}=\sum_j a_{\alpha,j}$ be mutually commuting translation-invariant finite-range charges of the exchange chain with $SU(3)$-invariant local densities, including $C_2^{(0)}=H_0$. For a density $A_I$ supported in an integer interval $I$, we define $\mathcal D_f A_I=i[X_I,A_I]$ and $C_\alpha^{(1)}=\sum_j\mathcal D_f a_{\alpha,j}$.
The resulting corrections remain within the original density support intervals and satisfy the bulk identities
\begin{equation}
[H_0,C_\alpha^{(1)}]+[V_f,C_\alpha^{(0)}]=0,\qquad
[C_\alpha^{(0)},C_\beta^{(1)}]+[C_\alpha^{(1)},C_\beta^{(0)}]=0.
\label{st:eq:mutualrepair}
\end{equation}
For any fixed finite collection of these charges, the identities hold on all sufficiently long periodic chains. Hence they admit simultaneous local first-order repairs for every intrinsic $T_f$, whereas for the Fermat direction no nontrivial uniform one-site charge admits any periodic first-order repair.
\end{proposition}

Terms joining an external site to the interval combine cleanly into the action of a common $SU(3)$ generator on that specific interval. They therefore commute with an invariant density, ensuring that only pairs directly inside its support contribute. The Jacobi identity subsequently preserves the mutual commutation of the repaired charges throughout first order. A finite collection with maximum density range $R$ successfully periodizes for $L>2R$, and the bilocal generator itself need not periodize. This reveals that the resonant perturbation fundamentally commutes with the invariant charges inside exchange energy blocks but need not commute with the one-site $SU(3)$ generators. Because the resulting correction $\mathcal D_fA_I$ generally ceases to be invariant, the locality argument cannot be simply iterated, and second-order repair of the entire hierarchy is not globally guaranteed.

\subsection{Thermodynamic protection of the charge memory matrix}
\label{main:sec:memory}

An exact repair must perfectly cancel the charge force as an operator on the whole finite ring, whereas thermal memory pragmatically asks how much of a small initial charge bias remains after time evolution. These dynamic tests can yield different answers because the thermal estimate allows a nonzero force remainder as long as its mean square per site is small enough. We make this distinction quantitative for the same intrinsic perturbations in the infinite-temperature ensemble, ensuring the thermodynamic limit is taken strictly before weak coupling. We bias the infinite-temperature ensemble infinitesimally toward one uniform charge and evolve under the full perturbed Hamiltonian. For a real orthonormal basis of traceless Hermitian one-site operators obeying $\tr(q_aq_b)=\delta_{ab}$, we define the memory matrix $M_{ab,L}(t)=3\Tr[Q_L(q_a,t)Q_L(q_b)]/(L3^L)$ initialized at $M_L(0)=I_8$. A reduction in this operator norm cannot be artificially hidden by rotating the charge basis, making closeness to $I_8$ a rigorous guarantee of retention in every original direction. 

For a broken one-site charge, our protection mechanism utilizes a correction built from pairs of sites separated by at most a finite range $R$. Increasing $R$ beneficially reduces the variance per site of the uncancelled first-order force proportionally to $R^{-1}$, while the variance of the correction inevitably grows as $R$. The perturbation acting on this correction adds a further force cost containing terms of order $|\lambda|/\sqrt R$ and $\lambda^2\sqrt R$. Choosing an optimal range $R$ of order $|\lambda|^{-1}$ balances these competing terms while keeping the corrected charge close in thermal norm. Because a first-order force remainder survives at every finite range, this construction successfully controls thermal memory without artificially solving the forbidden exact repair equation.

\begin{theorem}[Thermodynamic thermal protection]
\label{main:thm:thermal}
Let $h_\lambda=P+\lambda T_f+\lambda^2N$, where $T_f$ is Hermitian intrinsic and $N$ is a fixed Hermitian two-site operator with scalar partial traces. Taking the thermodynamic limit first, as $\lambda\to0$ from either side,
\begin{equation}
\norm{M_\infty(t_\lambda)-I_8}_{\rm op}\longrightarrow0
\quad\text{if}\quad |\lambda|^{3/2}|t_\lambda|\longrightarrow0.
\label{main:eq:thermalwindow}
\end{equation}
The statement applies to both the raw and compensated Fermat paths. For these paths, uniformly over normalized one-site charges, the memory deficit at $t=s_0/|\lambda|$, with fixed $s_0$, is strictly $O(|\lambda|)$.
\end{theorem}

The proof constructs an explicit corrected extensive charge using the triangular cutoff $X_{R,L}=\sum_j\sum_{d=1}^R(1-\frac{d}{R+1})r_{j,j+d}$ to define the protected observable $\widetilde Q_L=Q_L(q)+\lambda X_{R,L}$, where $S=i[T_f,q_1+q_2]$ and $r=-iSP/2$. For sufficiently large rings $L\ge2R+3$, the zero partial traces of $r$ make distinct pair supports perfectly orthogonal in the normalized tracial norm $\norm A_0^2=3^{-L}\Tr(A^\dagger A)$. The commutators of $X_{R,L}$ with the perturbing densities possess a variance per site of $O(R)$. Time integration under the full Hamiltonian $H_{\lambda,L}$ yields an explicit bound $E_R(\lambda,t)$ ensuring $\norm{Q_L(q,t)-Q_L(q)}_0/\sqrt{L\chi_q}\le E_R(\lambda,t)$ and $1-c_{q,L}(t)\le E_R(\lambda,t)^2/2$, where $\chi_q=\tr(q^2)/3$ and $c_{q,L}(t)=\Tr[Q_L(q,t)Q_L(q)]/(L3^L\chi_q)$. Taking $R=\lceil|\lambda|^{-1}\rceil$ provides the critical scaling bound $E_R=O(|\lambda|^{1/2})+O(|t|\,|\lambda|^{3/2})$. The exact repair obstruction remains mathematically intact because the necessary range diverges as $\lambda\to0$ and a first-order remainder perpetually survives at every finite range. When $S\ne0$, the same construction bounds the normalized thermodynamic spectral measure of the intrinsic charge force under the unperturbed exchange Liouvillian, yielding \begin{equation}
\nu_{S,\infty}([-\epsilon,\epsilon])\le(13/6+2/\sqrt3)\epsilon,\qquad 0<\epsilon\le1.
\label{main:eq:lowfrequency}
\end{equation} Therefore, finite-ring resonances do not generate a zero-frequency atom of finite thermal weight for this specific force.

\subsection{Simultaneous loss of every one-site memory at fixed strength}

Thermodynamic protection mitigates decay but does not imply perfect, infinite-time retention. We next interrogate whether every one-site charge direction exhibits a nonzero memory loss at a precisely common time while the macroscopic current remains more strongly protected. Resolving the full multi-dimensional matrix is paramount here because a one-dimensional population measurement alone would miss orthogonal one-site directions. For the Fermat paths, the entire memory matrix robustly assumes the exact block-diagonal form \begin{equation}
M_L(t)=\operatorname{diag}(c_ZI_2,c_XI_2,c_{XZ}I_4).
\label{main:eq:memoryblocks}
\end{equation} Here, the operators obey $X\ket a=\ket{a+1\bmod3}$ and $Z\ket a=\omega^a\ket a$ with $\omega=e^{2\pi i/3}$. Here $Z$ denotes the Weyl matrix, distinct from the diagonal control observable used in Fig.~\ref{main:fig:sharp}. Global Weyl symmetry eliminates cross-channel entries, while reflection combined with color inversion ensures the blocks remain real scalars.

For a normalized positive spectral measure, the second and fourth moments dictate the rigorous all-time inequalities \begin{equation}
m_2t^2/2-m_4t^4/24\le1-c(t)\le m_2t^2/2.
\label{main:eq:momentbounds}
\end{equation} The exact moment polynomials corresponding to these bounds produce one distinct common time at which all three operational blocks simultaneously contract.

\begin{proposition}[Complete memory contraction at fixed strength]
\label{main:prop:memoryloss}
For the compensated Fermat path parameterized by $u=\lambda^2$, we define $B(u)=8+98u+309u^2+783u^3$, generating the squared time $t_*^2=6(1+9u)/B(u)$ and the bound $d_*(u)=6u(1+9u)^2/B(u)$. For coupling $0<|\lambda|\le1/5$, uniformly in $L\ge5$ and in the thermodynamic limit,
\begin{equation}
\norm{M_L(t_*)}_{\rm op}\le1-d_*(u),\qquad
\frac{C_J(t_*)}{\chi_J}\ge1-\frac{\Gamma_Jt_*^2}{2}.
\label{main:eq:jointmemory}
\end{equation}
\end{proposition}

Figure~\ref{main:fig:time} explicitly displays these two dynamic bounds. For example, at $\lambda=0.1$ and $t_*\simeq0.851895$, all uniform one-site charge directions conclusively lose at least approximately $0.791\%$ of their thermal memory, whereas the energy current loss is tightly constrained to at most approximately $0.0171\%$. These constitute rigorous mathematical bounds rather than numerically sampled time-evolution curves, and at any fixed coupling, the guaranteed loss securely persists as the volume approaches infinity. Consequently, at a time tending to $\sqrt3/2$, every one-site direction universally loses at least $\tfrac34\lambda^2+O(\lambda^4)$ of its memory, while macroscopic current loss remains bounded at $O(\lambda^6)$. This represents a controlled comparison within a single ensemble, establishing a definitive hierarchy between macroscopic current retention and uniform extensive-charge memory loss.

\begin{figure}[t]
\centering
\includegraphics[width=0.66\linewidth]{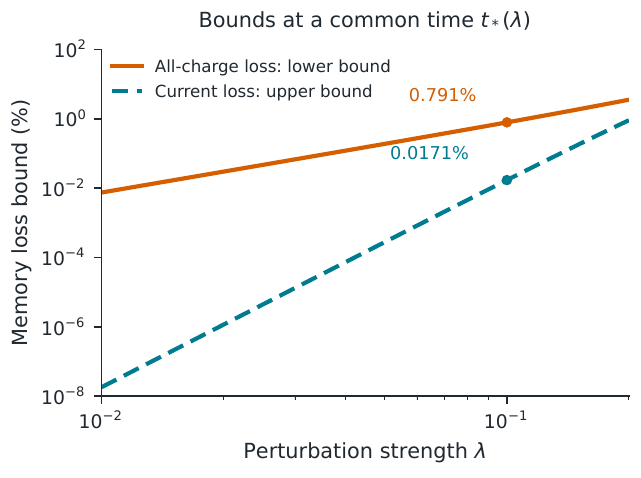}
\caption{Bounds on complete charge-memory loss and current loss for the
compensated Fermat model. At each $0.01\le\lambda\le0.2$, both curves
refer to the common time $t_*(\lambda)$ defined previously. The orange curve is the lower bound
$d_*(\lambda^2)$ on $1-\norm{M(t_*)}_{\rm op}$; the dashed teal
curve is the upper bound $\Gamma_Jt_*^2/2$ on current loss.
The marked values use $\lambda=0.1$, with losses shown as percentages.
These bounds hold uniformly for $L\ge5$ and in the thermodynamic
limit; they are not measured decay curves.
AI-assisted plotting code and source data are available in Ref.~\cite{code}; see
\hyperref[main:sec:computational]{Computational methods and AI assistance}.}
\label{main:fig:time}
\end{figure}

\FloatBarrier
\section{Nonlinear Proof Strategy and Parameter Uniformity}
\label{main:sec:proof}

Two complex features of the underlying physics necessitate a deeply nonlinear mathematical proof. The cyclic conversion process paradoxically passes the linear integrability test despite globally breaking every one-site charge, and an optimal compensation bound must comprehensively cover all nearby interactions, including elusive paths that approach singular directions as the coupling decreases. We elucidate how leveraging the same nonlinear constraints efficiently resolves both outstanding issues.

\subsection{Nonlinear selection and the two intrinsic types}

After meticulously removing scalar and exchange-scale transformations, the linear condition broadly allows arbitrary couplings between symmetric and antisymmetric pair states. The cyclic conversion embodies one such specific coupling. The proof intentionally initiates from this expansive full space without preemptively restricting the allowed vertex structures. Higher orders critically test whether these varied physical processes can simultaneously coexist in an exact integrable interaction. Hermiticity naturally provides stringent positive quadratic forms, and the ensuing cubic constraints reduce the possibilities to exactly two nonzero intrinsic mechanisms up to a common basis change. One mechanism flawlessly preserves two independent species populations, while the alternative allows neutral-pair conversion while concurrently preserving a single equally spaced charge. Because the cyclic conversion interaction breaks every possible variation of such charges, it admits no exact analytic Yang--Baxter completion through exchange.

A classification limited to first derivatives does not automatically follow a dynamic charge along a prescribed analytic family. For a nonzero intrinsic leading direction, exact spectral families densely fill the local solution set of the full governing equations, systematically controlling higher-order changes without forcing an artificial mathematical ansatz. Directions led by one-site terms require a fundamentally separate argument where the analytic reduction scrupulously respects the candidate symmetry. Variables that threaten to break this symmetry are subsequently eliminated utilizing strict invertibility or positivity arguments, firmly establishing robust charge continuation without requiring the classification of every neutral higher coefficient.

\subsection{Ensuring uniform estimates at singular directions}

Uniformity across the parameter space is paramount because a dynamic compensation scheme could inadvertently approach a singular direction as the primary coupling strength strictly decreases. The proof initially eliminates directions whose physical displacement is directly controlled by the residual. It then accurately bounds the symmetry defect within the remaining, highly reduced family manifold. Letting $\sigma$ denote its tangent amplitude and $\rho$ its residual norm, the foundational elementary estimate reads $\delta\le C_1\sigma$. Operating near a nonzero intrinsic direction or a direction led by an exchange-dependent one-site term, the full nonlinear equations yield the refined relation $\delta\le C\rho/\sigma^2$. Conversely, near a pure one-site difference, spectral positivity instead enforces $\delta^2\le C\rho/\sigma$ within the critical small-residual regime. In every evaluated case, the elementary estimate seamlessly closes the exact same overarching cubic bound, while models with large residuals already satisfy the limit directly. These comprehensive estimates govern actual analytic functions across the phase space. The topological compactness of the unit direction sphere subsequently supplies a single, universal bound that completely covers the entire interaction neighborhood.

\section{Discussion and Outlook}
\label{main:sec:discussion}

The cyclic conversion process provides a concrete physical lens through which to view our primary theoretical results. Because this mechanism changes all three species populations, it explicitly breaks every nontrivial one-site charge. Exact analytic Yang-Baxter integrability strictly excludes this direction, mathematically demanding that at least one weighted population must remain conserved in a suitable common basis. Approximate current conservation can nevertheless be robustly optimized without ever restoring that broken charge. By introducing a quadratic compensation, the specific model successfully reaches the cubic residual bound and dramatically changes the guaranteed current-memory window from inverse-square to inverse-cube scaling. The resulting neighborhood-wide inequality rigorously establishes that no nearest-neighbor compensation can improve the residual power while simultaneously retaining complete first-order local symmetry breaking. This fundamental limitation adds a necessary physical restriction to the constructive tools typically employed in weak integrability breaking~\cite{surace,schouten}.

Beyond local mathematical constraints, our analysis cleanly separates microscopic repair obstructions from macroscopic thermal memory. The exact repair obstruction questions whether the charge force can vanish identically as an operator, particularly inside each degenerate exchange eigenspace. Conversely, the thermal theorem investigates whether an infinitesimal uniform charge bias meaningfully persists after averaging over the infinite-temperature ensemble. For the intrinsic perturbations and second-order terms specified previously, systematically increasing the correction range successfully controls that thermodynamic average even though exact repair of a broken charge fundamentally fails on every finite ring. Taking the thermodynamic limit first, the full memory matrix approaches its initial value for the timescale $t=o(|\lambda|^{-3/2})$. This thermodynamic protection follows directly from explicit force-variance bounds and does not require any ad hoc assumption that the resonant states are rare. The same theoretical construction rigorously bounds the intrinsic force's low-frequency mass under $[H_0,\,\cdot\,]$ and definitively excludes an isolated atom at zero frequency.

To visualize these macroscopic consequences, the exact moment comparison in Fig.~\ref{main:fig:time} resolves all eight one-site directions within the exact same ensemble as the current for the compensated model. At a common time of order one, the collective memory loss of these charges is at least $\tfrac34\lambda^2+O(\lambda^4)$, while the corresponding current loss is tightly constrained to at most $O(\lambda^6)$. Together with the current half-retention and energy-spreading bounds, this stark dynamic contrast yields a direct finite-time physical consequence of the underlying algebraic conservation constraints.

While these results establish a robust framework, the foundational hypotheses remain essential for defining the strict boundaries of the theory. Stability is definitively uniform for nearby Hermitian two-site densities, whereas thermodynamic thermal protection additionally requires an intrinsic tangent combined with a scalar-marginal second-order term. Exact charge continuation inherently assumes a jointly analytic regular additive Yang-Baxter realization. Although our classification fully exhausts first derivatives and arcs with nonzero intrinsic leading terms, it leaves the neutral higher coefficients of pure one-site leading arcs open for future exploration. Uniform fields and alternative spectral charts are formally distinguished in Appendixes~\ref{fs:app:uniformfield}, \ref{fs:sec:33vertex}, and \ref{fs:app:elliptic}. Because the fixed-strength memory-loss bound naturally vanishes at weak coupling, these results do not currently identify an exact asymptotic relaxation rate or an infinite-time Drude weight. The force spectrum under $H_0$ similarly supplies no corresponding long-time statement for the compensated current force. Consequently, determining the fixed fractional loss of complete charge memory before the current guarantee ends remains a compelling open dynamical question, featuring a possible power-law scale restricted to $|\lambda|^{-p}$ with $3/2\le p<3$. Identifying effective uniform constants and extending these optimal constraints to larger local dimensions constitute the next critical steps for understanding thermal protection mechanisms when exact finite-size repair tests fail.

\begin{acknowledgments}
This work was supported by the Science and Technology Development Project
of Jilin Province (Grant No.\ 20250102032JC).
GPT Astra (OpenAI; model identifier gpt-6-astra) assisted with manuscript
organization, drafting, and LaTeX preparation. Its outputs were revised
under author direction; research and figure assistance is described in
\hyperref[main:sec:computational]{Computational methods and AI assistance}.
\end{acknowledgments}

\section*{Author contributions}
\textbf{Hanbing Liang:} Conceptualization, Project administration, Validation,
Writing -- original draft.
\textbf{Fujun Liu:} Validation, Writing -- review \& editing.

\section*{Competing interests}
The authors declare that they have no competing interests.

\section*{Data and code availability}
All definitions, tensor identities, analytic constructions, and proofs
supporting the results are provided in this article and its appendixes.
Auxiliary verification code, figure source data, plotting scripts, and
execution instructions are publicly available in version 0.3.0 of
\emph{Qutrit Forced Symmetry}~\cite{code}, which accompanies this article.
The response curves
evaluate analytic expressions and do not use fitted or simulated data.

\section*{Computational methods and AI assistance}
\label{main:sec:computational}
The classification and analytic continuation proofs in
Appendixes~\ref{fs:app:linearproof}--\ref{fs:app:continuation} supply
the input to the uniform error estimate in Appendix~\ref{main:app:stability}.
The response bounds use the periodic operator-word identity in
Appendix~\ref{main:app:norm} and the finite contractions in
Appendixes~\ref{main:app:fermat} and \ref{main:app:controls}.
Separate arbitrary-length proofs establish the resonance obstruction
(Appendix~\ref{main:app:resonance}) and the thermal support estimates
(Appendix~\ref{main:app:thermal}). Discrete symmetry, the displayed
coefficient Gram matrices, and spectral positivity give complete-memory
contraction in Appendix~\ref{main:app:memory}. All finite coefficients
follow from the explicit tensor and coefficient-extraction rules.

Symbolic calculations assisted the development of these formulas.
The repository programs independently check the displayed tensor
identities, complete spectral Yang--Baxter equations, density maps,
removable limits, residual polynomials, and Gram spectra. Further finite
checks test intrinsic charge-action conventions and periodic magnon
closing-bond factors. The programs are cross-checks of the written
derivations; their output is not an additional premise of a theorem.
Scientific plots use NumPy and Matplotlib to evaluate the exact
expressions. Figure~\ref{main:fig:time} compares rigorous memory-loss
bounds, with no dynamics fitting or finite-size extrapolation.

GPT Astra (OpenAI; model identifier gpt-6-astra) assisted with literature
searches, symbolic calculations and verification code, proof development,
and scientific plotting code. The same model was used in separate sessions for additional
automated critical checks. The authors directed the research questions,
scope, and revision tasks. Outputs were assessed through explicit
derivations, exact algebraic checks, and comparison with primary sources;
the arguments and the scope of executable checks are documented in the
article and repository documentation. Responsibility for scientific
judgments, mathematical claims, and the accuracy of the final manuscript
rests with the authors.
\FloatBarrier
\appendix
\section{Exact charge theorem and complete first-order classification}
\label{fs:app:linear}
\label{fs:sec:theorem}

Let \(\mathcal J_{\mathrm H}(P)\) denote the first derivatives of the
jointly analytic Hermitian Yang--Baxter families specified in
Sec.~\ref{main:sec:setting}, after removing the scalar and scale directions
\(\I,P\). This is a set of velocities of actual analytic families.
A scalar spectral factor and an analytic spectral rescaling restore
\(a\I+bP\); both operations preserve every total on-site charge.
No further entry or symmetry restriction is imposed on the density.

\subsection{Charge theorem and tensor coordinates}

Appendix~\ref{fs:app:continuation} proves Theorem~\ref{fs:thm:main} for the given density
arc, including every scalar and scale term. The unitary acts identically
on all sites; it may depend on the model parameter. The assertion is local
in \(\tau\) and does not prescribe a charge common to different arcs.

\begin{corollary}[First-order charges and fixed-charge realizations]
\label{fs:cor:firstcharge}
For every \(K\in\mathcal J_{\mathrm H}(P)\), there is a nonzero
traceless Hermitian on-site operator \(q\) such that
\[
[K,q\otimes\I+\I\otimes q]=0.
\]
For every \(K\in\mathcal J_{\mathrm H}(P)\), one may choose a realizing
jointly analytic regular additive Yang--Baxter family with Hermitian
density for which a nonzero traceless Hermitian \(q\) is fixed:
\[
[\check R(\tau,u),q\otimes\I+\I\otimes q]=0
\quad\text{for all }(\tau,u)\text{ near }(0,0).
\]
\end{corollary}

The full linearized solution space of the Reshetikhin condition stated in
Sec.~\ref{main:sec:setting}, after the two normalizations, is
\begin{equation}
\mathcal T=\{K:PKP=-K\}
\cong\slthree\oplus\slthree\oplus S^3V\oplus S^3V^*.
\label{fs:eq:chart}
\end{equation}
It has complex dimension \(36\), and its Hermitian real form has
real dimension \(36\). Subsection~\ref{fs:app:linearproof} proves this statement
directly in the full 81-dimensional two-site matrix space.

Write the coordinates as \((A,B;f,g)\), where \(A,B\) are traceless
one-site matrices and \(f,g\) are completely symmetric cubic tensors.
With \(\varepsilon_{012}=\varepsilon^{012}=1\), set
\begin{align}
K&=\Tel(A)+\Hel(B)+K_{f,g},\label{fs:eq:decomposition}\\
\Tel(A)&=A\otimes\I-\I\otimes A,\nonumber\\
\Hel(B)&=P(B\otimes\I)-(B\otimes\I)P,\nonumber\\
(K_{f,g})^{ab}{}_{ij}
&=f^{abc}\varepsilon_{cij}+\varepsilon^{abc}g_{cij}.
\label{fs:eq:tensor}
\end{align}
Repeated indices are summed from \(0\) to \(2\). We use the polynomial
convention \(f(x)=f^{ijk}x_ix_jx_k\), so that the coefficient of
\(x_0x_1x_2\) is \(6f^{012}\). Hermiticity is precisely
\begin{equation}
A^\dagger=A,\qquad B^\dagger=-B,\qquad g=\bar f.
\label{fs:eq:realform}
\end{equation}
The bar means componentwise conjugation in an orthonormal basis.
The tensor formula is equivariant under common \(SU(3)\) basis changes;
the scalar phase of a \(U(3)\) basis change acts trivially on the density.

The Lie algebra acts on each upper index of \(f\), and with the opposite
dual action on each lower index of \(g\). For example,
\[
(A\cdot f)^{abc}
=A^a{}_r f^{rbc}+A^b{}_r f^{arc}+A^c{}_r f^{abr}.
\]
Define the simultaneous stabilizer
\[
\Stab(f,g)=
\{Z\in\slthree:Z\cdot f=0,\ Z\cdot g=0\}.
\]
The cubic covariant \(E_{70}=(E_+,E_-)\), defined explicitly in
Appendix~\ref{fs:app:covariant}, has 35 components in each of two dual
irreducible representations. In the Hermitian real form the two equations
are complex conjugates.

\subsection{Complete criterion and intrinsic normal forms}

\begin{theorem}[Complete Hermitian first-order classification]
\label{fs:thm:classification}
In the coordinates \eqref{fs:eq:chart}--\eqref{fs:eq:realform},
\begin{equation}
\mathcal J_{\mathrm H}(P)=
\left\{(A,B;f,\bar f):
\begin{array}{l}
E_{70}(f,\bar f)=0,\\
A,B\in\Stab(f,\bar f),\quad[A,B]=0
\end{array}\right\}.
\label{fs:eq:classification}
\end{equation}
Every element of the right-hand side has a jointly analytic realization
by a regular additive Yang--Baxter matrix with Hermitian density.

For \(f\ne0\), a common on-site unitary transformation puts \(f\) into
exactly one of the following two types:
\begin{align}
\text{triangle:}\qquad
&f=\alpha x_0x_1x_2,\qquad \alpha\in\C^\times;
\label{fs:eq:triangle}\\
\text{secant:}\qquad
&f=3c\,x_2(2x_0x_1+x_2^2),\qquad c\in\R^\times.
\label{fs:eq:secant}
\end{align}
Their connected compact stabilizers are respectively the diagonal
Cartan torus \(U(1)^2\) and
\begin{equation}
\{\operatorname{diag}(e^{i\theta},e^{-i\theta},1):
\theta\in\R\}\cong U(1).
\label{fs:eq:arithmetic}
\end{equation}
The spaces of traceless Hermitian on-site generators commuting with
\(K\) have dimensions two and one, respectively.
The two nonzero types are disjoint. Their realizations can be chosen
to preserve these groups. For \(f=0\), the only condition is
\([A,B]=0\), and the realization preserves a common diagonal torus
of \(A\) and \(B\).
\end{theorem}

Discrete basis transformations can identify different parameter values
within either normal-form type.
The names refer to the factorization of the cubic \(f\): three
distinct linear factors form a triangle, while the secant form
is a line meeting a nondegenerate conic at two distinct points.
For \(f\ne0\), the stabilizers in the theorem are abelian, so their
one-site dressing data commute automatically. The commutator condition
is essential on the pure dressing locus \(f=0\).

For completeness, the last necessary conditions obtained from the
third-order projections, after the second-order stabilizer relations,
are
\begin{equation}
E_{70}(f,\bar f)=0,\qquad [B,[A,B]]=0.
\label{fs:eq:necessarymain}
\end{equation}
On the Hermitian real form, $C=[A,B]=C^\dagger$, so the second condition
gives the positive trace identity
\begin{equation}
0=\tr(A[B,C])=\tr(C^2)=\norm C_\HS^2.
\label{fs:eq:positivitymain}
\end{equation}
Thus $[A,B]=0$. The proof of these projections keeps all higher model
coefficients free.

\begin{proof}[Proof of Theorem~\ref{fs:thm:classification}]
Subsection~\ref{fs:app:linearproof} gives the unrestricted linearized space.
Lemma~\ref{fs:lem:necessary} in Appendix~\ref{fs:app:obstructions} gives
the necessary nonlinear conditions, and Appendix~\ref{fs:app:intrinsic}
classifies their Hermitian intrinsic solutions. The trace identity
\eqref{fs:eq:positivitymain} completes the necessary commutator condition.
Appendixes~\ref{fs:app:lifts} and \ref{fs:app:ikmatrix} realize every
point of \eqref{fs:eq:classification}, and the distinct spectra
\eqref{fs:eq:fingerprint} separate the two nonzero types.
\end{proof}

\begin{proof}[Proof of Corollary~\ref{fs:cor:firstcharge}]
The decomposition \eqref{fs:eq:decomposition} is equivariant and direct.
For \(f\ne0\), the intrinsic stabilizer is one of the abelian algebras
in Theorem~\ref{fs:thm:classification}, and therefore commutes with \(A\) and \(B\).
Conversely a symmetry of \(K\) must fix its intrinsic component,
so the number of generators is exactly the one stated.
For \(f=0\), the commuting Hermitian matrices \(A\) and \(iB\) are
simultaneously diagonalizable and hence commute with a Cartan torus.
The realization assertion follows from the group-preserving constructions
in Appendix~\ref{fs:app:lifts}.
\end{proof}

\subsection{Dimensions, the contraction diagnostic, and one-site terms}

The nonlinear restriction can be quantified by the generic real
dimensions of the three loci in Theorem~\ref{fs:thm:classification}:
\begin{equation}
\begin{aligned}
\dim_{\R}\mathcal J_{\rm triangle}&=(8-2)+2+4=12,\\
\dim_{\R}\mathcal J_{\rm secant}&=(8-1)+1+2=10,\\
\dim_{\R}\mathcal J_{f=0}&=(8-2)+4=10.
\end{aligned}
\label{fs:eq:stratumdimensions}
\end{equation}
For the triangle, these terms count the common \(SU(3)\) orbit,
the complex coefficient \(\alpha\), and the two Cartan-valued dressing
matrices. The secant has a seven-dimensional orbit, one real intrinsic
amplitude, and two real dressing parameters. For \(f=0\), a generic
commuting pair \((A,iB)\) has a six-dimensional common orbit and four
diagonal parameters. Discrete normal-form identifications do not change
these dimensions. Thus the largest realizable stratum has dimension 12
inside the 36-dimensional Hermitian linearized space, before quotienting
by common basis changes.

There is also a basis-independent distinction between the two types.
The positive contraction
\begin{equation}
M^a{}_i=f^{amn}\overline{f^{imn}},\qquad m=\tr M>0
\label{fs:eq:contractionM}
\end{equation}
has normalized spectrum
\begin{equation}
\spec(M/m)=
\begin{cases}
(1/3,1/3,1/3),&\text{triangle},\\
(2/15,2/15,11/15),&\text{secant}.
\end{cases}
\label{fs:eq:fingerprint}
\end{equation}
For an allowed first derivative, this contraction spectrum identifies
the intrinsic symmetry type without constructing its normalizing basis.
The full membership criterion remains Eq.~\eqref{fs:eq:classification},
which also constrains the one-site terms.
For example, the Fermat tensor \(f_{\rm F}=x_0^3+x_1^3+x_2^3\)
has \(M=\I\), and hence the uniform normalized spectrum, yet
\(E_+(f_{\rm F},f_{\rm F})^{0000}{}_0=-2\). The obstruction
described in Sec.~\ref{main:sec:setting} therefore excludes this tensor
despite its contraction spectrum.

The one-site terms in \eqref{fs:eq:decomposition} have distinct roles.
For a periodic chain, \(\sum_j\Tel(A)_{j,j+1}=0\);
on an open chain it is \(A_1-A_L\).
It must nevertheless be retained in the local Yang--Baxter problem,
where different density representatives can have different local
compatibility conditions.
The term \(\Hel(B)\) is the derivative of a unitary spatial twist.
The theorem says that both contributions must use the symmetry
algebra selected by the intrinsic interaction. On the pure dressing
locus, Hermiticity forces their generators to commute.

\subsection{The cyclic conversion direction has no on-site charge}
\label{fs:app:fermat}

The cyclic conversion \(T=K_{\rm F}\) defined in
Sec.~\ref{main:sec:setting} has \(A=B=0\) and
\(f_{\rm F}=x_0^3+x_1^3+x_2^3\). The tensor decomposition is direct
and equivariant. Thus a traceless on-site generator \(Z\) commuting
with \(K_{\rm F}\) must satisfy \(Z\cdot f_{\rm F}=0\).
The coefficients of \(x_a^3\) and \(x_a^2x_b\) in this equation
force every entry of \(Z\in\slthree\) to vanish. This excludes a
nontrivial charge in an arbitrary basis, not only a diagonal charge in
the basis used to define \(T\) in Sec.~\ref{main:sec:setting}.

Substitution in the full cubic obstruction \eqref{fs:eq:covariant} gives
\begin{equation}
E_+(f_{\rm F},f_{\rm F})^{0000}{}_0=-2.
\label{fs:eq:fermatobstruction}
\end{equation}
Lemma~\ref{fs:lem:necessary} requires every such component to vanish for
an actual arc, independently of its higher derivatives. Hence
\(K_{\rm F}\) passes the linearized condition but has no analytic
Hermitian Yang--Baxter realization.

\subsection{The boundary with 33-vertex models}

\label{fs:sec:33vertex}

The degenerate charge \(q_{33}=\operatorname{diag}(0,1,1)\) has
two-site sectors of dimensions \(1,4,4\), allowing
\(1^2+4^2+4^2=33\) entries. This is the symmetry assumed in
Ref.~\cite{cfrv33}. Its Section A.4 notes that the generic listed
Hamiltonians fail the Reshetikhin condition stated in
Sec.~\ref{main:sec:setting}; additional parameter constraints would be
required for a difference-form spectral matrix.

Even for an exceptional specialization satisfying our hypotheses, the
degenerate charge has a sharp first-order consequence. The decomposition
\eqref{fs:eq:decomposition} is equivariant under common basis changes, so
a charge of \(K\) also stabilizes its intrinsic tensor. For the nonzero
single-charge type, every nontrivial Hermitian stabilizer has three
distinct, equally spaced eigenvalues. Hence a realizable \(K\)
preserving \(q_{33}\) has either Cartan-type intrinsic part or no
intrinsic part. In both cases its allowed one-site terms commute, and
\(K\) preserves a Cartan torus. The 33-entry support therefore supplies
no third first-order type at \(P\). This conclusion concerns the first
derivative; it does not classify the higher derivatives or exclude
33-vertex models at finite distance from \(P\).

\subsection{Proof of the full linearized space}
\label{fs:app:linearproof}

Write \(s=P_{12}\), \(t=P_{23}\), \(J=DF_P\), and
\(S=\{P,K\}\). Expanding
\[
F(h)=\{h_{12}^2,h_{23}\}-\{h_{23}^2,h_{12}\}
       +2(h_{23}h_{12}h_{23}-h_{12}h_{23}h_{12})
\]
gives
\begin{equation}
J(K)=2\Delta K+\mathcal L(S),\qquad
\mathcal L(S)=\{S_{12},t\}-\{S_{23},s\}.
\label{fs:eq:linearL}
\end{equation}
The derivative of the last two cubic terms cancels. Indeed, for
\(a=K_{12}\), \(b=K_{23}\), the identities
\(bst=sta\), \(tat=sbs\), and \(tsb=ats\) cancel its six terms in pairs.

Decompose \(K\) according to \(PKP=\pm K\). The odd part has \(S=0\)
and automatically solves the linearized equation. On the even part,
\(S=2PK\) is an invertible linear transformation. It remains to show
\begin{equation}
PSP=S,\qquad\mathcal L(S)\in\operatorname{im}\Delta
\quad\Longrightarrow\quad S\in\Span\{\I,P\}.
\label{fs:eq:evenkernel}
\end{equation}

The two eigenspaces of \(P\) are \(S^2V\) and \(\Lambda^2V\), of
dimensions six and three. Accordingly, the even operator space is
\[
\End(S^2V)\oplus\End(\Lambda^2V)
\cong(\mathbf1\oplus\mathbf8\oplus\mathbf{27})
 \oplus(\mathbf1\oplus\mathbf8).
\]
Here \(\mathbf8=\slthree\).
Contraction of one upper and one lower index in
\(\End(S^2V)\) gives the scalar and adjoint summands; its traceless
kernel is the irreducible highest-weight \((2,2)\) module of
dimension 27. The identification
\(\Lambda^2V\cong V^*\) gives the second scalar and adjoint summands.
The map \(\mathcal L\bmod\operatorname{im}\Delta\) is equivariant,
so its kernel is a submodule.

The scalar summands are spanned by \(\I,P\), and
\(\mathcal L(\I)=2\Delta P\), \(\mathcal L(P)=0\).
For the other summands put \(X=E_{02}\). Highest-weight vectors
in the two adjoint copies are
\[
S_1=X\otimes\I+\I\otimes X,\qquad S_2=PS_1,
\]
and a highest-weight vector of \(\mathbf{27}\) is \(S_{27}=X\otimes X\).
Index rows and columns by three-letter words. Direct use of the
permutation action gives
\[
\begin{array}{c|cc}
 &\mathcal L(S_1)&\mathcal L(S_2)\\ \hline
(010,201)&2&0\\
(100,201)&0&1
\end{array}
\qquad
\mathcal L(S_{27})_{001,212}=1.
\]
All three functionals vanish on \(\Delta Y\), since both the first
and last site indices change. Thus the map is injective on the
two-dimensional highest-weight multiplicity space of the adjoint
summands and nonzero on \(\mathbf{27}\). Complete reducibility
proves \eqref{fs:eq:evenkernel}.

The odd operator space is
\[
\operatorname{Hom}(\Lambda^2V,S^2V)
\oplus\operatorname{Hom}(S^2V,\Lambda^2V),
\]
of dimension \(2\cdot6\cdot3=36\).
Using
\(S^2V\otimes V\cong S^3V\oplus\slthree\)
and its dual gives \eqref{fs:eq:chart} and the tensor coordinates
\eqref{fs:eq:decomposition}--\eqref{fs:eq:tensor}.
Including \(\I,P\), the full kernel has dimension 38 and the
linear map has rank \(81-38=43\).
A Hermitian odd operator is specified by an arbitrary complex
\(6\times3\) block and its adjoint, giving 36 real dimensions.
The two one-site coordinates can be read off directly:
\[
A=\tfrac13\Tr_2K,\qquad B=-\tfrac13\Tr_2(KP).
\]
The intrinsic tensor has zero partial traces, both before and
after multiplication by \(P\), which verifies these formulas
from \eqref{fs:eq:tensor}. They also give \(A^\dagger=A\) and
\(B^\dagger=-B\) when \(K^\dagger=K\) and \(PKP=-K\).
This also proves that no first-order direction has been lost
by a support restriction.

\subsection{Uniform fields and the local spectral-derivative class}
\label{fs:app:uniformfield}

For a nonzero traceless Hermitian \(q\), the uniform-field density
derivative is
\[
K_{\rm field}=q\otimes\I+\I\otimes q.
\]
Its sum on a periodic chain is \(2Q_L\). The common on-site
symmetry of the permutation chain makes \(Q_L\) commute with its
transfer-matrix hierarchy. Adding this field therefore preserves that
commuting hierarchy.

However, \(PK_{\rm field}P=K_{\rm field}\), and
\(K_{\rm field}\notin\Span\{\I,P\}\): its second partial trace
is \(3q\), whereas both \(\Tr_2\I\) and \(\Tr_2P\) are scalar.
The full linearized-space proof above consequently excludes this
direction from the regular additive spectral-derivative problem of
Sec.~\ref{main:sec:setting}. The charge theorem concerns that specified
class of local realizations.

\section{Explicit low-order necessary conditions}
\label{fs:app:obstructions}

\subsection{The intrinsic cubic covariant}
\label{fs:app:covariant}

For \(S\in S^4V\otimes V^*\), let
\((\tr S)^{abc}=S^{abcr}{}_r\), and define
\begin{equation}
(\Pi_{4,1}S)^{abcd}{}_e
=S^{abcd}{}_e-\frac16\left(
\delta^a_e(\tr S)^{bcd}+\delta^b_e(\tr S)^{acd}
+\delta^c_e(\tr S)^{abd}+\delta^d_e(\tr S)^{abc}\right).
\label{fs:eq:traceprojection}
\end{equation}
Taking the trace of the four correction terms gives six copies of
\(\tr S\), so the image is traceless. It is the 35-dimensional
irreducible module of highest weight \((4,1)\).
Complete symmetrization below is normalized as the average over
the 24 permutations. Set
\begin{equation}
\begin{split}
E_+(f,g)^{abcd}{}_e
=\Pi_{4,1}\Sym_{abcd}\bigl(
&2f^{amn}f^{bcd}g_{mne}
-6f^{abm}f^{cdn}g_{mne}\\
&-3\varepsilon^{aij}\varepsilon^{bmn}
 f^{cdr}g_{imr}g_{jne}\bigr).
\end{split}
\label{fs:eq:covariant}
\end{equation}
Interchanging \(f,g\) and upper and lower indices defines the
dual covariant \(E_-\). The equation \(E_{70}=0\) means
\(E_+=E_-=0\). If \(g=\bar f\), the two equations are conjugate.

For the component calculations below, the tensor formula has a useful
polynomial form. Put
\[
\phi(x)=f^{abc}x_ax_bx_c,\quad
G_e=(g_{ije})_{i,j=0}^2,\quad
\Omega(x)_{ij}=\varepsilon^{aij}x_a .
\]
Define three quartic polynomials
\begin{equation}
\begin{split}
\mathcal P_e(x)={}&\tfrac13\phi\,\tr(G_e\nabla^2\phi)
-\tfrac23(\nabla\phi)^{\mathsf T}G_e\nabla\phi\\
&+\sum_r(\partial_r\phi)\tr(\Omega G_r\Omega G_e),\\
\mathcal E_e(x)={}&\mathcal P_e(x)
-\frac{x_e}{6}\sum_r\partial_r\mathcal P_r(x).
\end{split}
\label{fs:eq:polynomialcovariant}
\end{equation}
Then
\(\mathcal E_e(x)=E_+^{abcd}{}_e x_ax_bx_cx_d\).
Indeed, \(\partial_m\partial_n\phi=6f^{amn}x_a\)
and \(\partial_m\phi=3f^{abm}x_ax_b\) give the first two
contractions in \eqref{fs:eq:covariant}. For the last one,
\[
\sum_{i,j,m,n}\Omega_{ij}\Omega_{mn}g_{imr}g_{jne}
=\tr(\Omega^{\mathsf T}G_r\Omega G_e)
=-\tr(\Omega G_r\Omega G_e).
\]
Finally, divergence of the symmetrized raw quartic is four times its
tensor trace. Its trace subtraction is consequently
\(-x_e\operatorname{div}\mathcal P/6\), as displayed.
Euler's identity for the cubic \(\operatorname{div}\mathcal P\)
also gives \(\operatorname{div}\mathcal E=0\).
Thus a component with upper-index multiplicities
\((n_0,n_1,n_2)\) is obtained by taking the coefficient of
\(x_0^{n_0}x_1^{n_1}x_2^{n_2}\) in \(\mathcal E_e\)
and dividing by \(4!/(n_0!n_1!n_2!)\).
This replaces symmetrization over four indices by three explicit
polynomial expressions.

\subsection{Expansion with free higher derivatives}

Define
\[
\mathscr Q(K)=[\tau^2]F(P+\tau K),\qquad
\mathscr B(K,N)=\mathscr Q(K+N)-\mathscr Q(K)-\mathscr Q(N).
\]
For
\(h=P+\tau K+\tau^2H_2+\tau^3H_3+\cdots\), the next two equations are
\begin{equation}
\begin{aligned}
J(H_2)+\mathscr Q(K)&\in\operatorname{im}\Delta,\\
J(H_3)+\mathscr B(K,H_2)+F(K)&\in\operatorname{im}\Delta.
\end{aligned}
\label{fs:eq:jet23}
\end{equation}
If \(s=P_{12}\), \(t=P_{23}\), \(L=K_{12}\), \(R=K_{23}\), the
quadratic expression is explicitly
\begin{equation}
\mathscr Q(K)
=[s+t,[L,R]]+[L+R,[s,R]+[L,t]].
\label{fs:eq:Qexplicit}
\end{equation}
Together with \eqref{fs:eq:linearL}, this formula fixes all
normalizations in the projections that follow.
Rows and columns of three-site matrices are indexed by three-letter words.

\begin{lemma}[Necessary relations for a Hermitian arc]
\label{fs:lem:necessary}
Let \(h=P+\tau K+O(\tau^2)\) be a formal Hermitian density arc
satisfying the Reshetikhin condition stated in Sec.~\ref{main:sec:setting},
and remove the scalar and scale parts of \(K\). In the full coordinates
\eqref{fs:eq:decomposition},
\[
A\cdot f=A\cdot\bar f=B\cdot f=B\cdot\bar f=0,
\qquad [A,B]=0,\qquad E_{70}(f,\bar f)=0.
\]
\end{lemma}

We prove the lemma in the following four steps.

\subsection{Second order forces the common stabilizer}

Consider the functionals
\[
\lambda_0Y=Y_{000,102},\qquad
\lambda_1Y=-Y_{000,021}+2Y_{000,102}
              -Y_{001,112}-Y_{002,212}.
\]
Both vanish on \(\Delta X\) and on \(J(X)\) for every two-site \(X\).
Using \eqref{fs:eq:Qexplicit} and \eqref{fs:eq:tensor} gives
\begin{equation}
\lambda_0\mathscr Q(K)=2((A+B)\cdot f)^{000},\qquad
\lambda_1\mathscr Q(K)=6(B\cdot f)^{000}.
\label{fs:eq:twoactionrows}
\end{equation}
For example, the right-hand sides are respectively
\(6\sum_j(A^0{}_j+B^0{}_j)f^{00j}\) and
\(18\sum_j B^0{}_jf^{00j}\).
These are entrywise identities for arbitrary traceless \(A,B\) and
symmetric \(f,g\). They follow by matrix multiplication with
\[
J(X)=2\Delta X+
\{\{P,X\}_{12},P_{23}\}-\{\{P,X\}_{23},P_{12}\}.
\]
The Kronecker deltas in \(P\) eliminate the corresponding intermediate
sums in these products.

The common \(SL_3\) orbit of the \(000\) component spans the dual
of \(S^3V\). Thus the second-order equation in \eqref{fs:eq:jet23}
implies \(A\cdot f=B\cdot f=0\).
Ordinary transposition, followed by interchange of upper and
lower indices, gives \(A\cdot g=B\cdot g=0\).

\subsection{An explicit quadratic correction}

Write \(T=K_{f,g}\), and introduce the contractions
\begin{align}
M^a{}_i&=f^{amn}g_{imn},&
m&=\tr M,&
S&=M\otimes\I+\I\otimes M,\nonumber\\
\mathsf U^{ab}{}_{ij}&=f^{abm}g_{ijm},&&&
\label{fs:eq:normaldefinitions}\\
\mathsf F^{ab}{}_{ij}
&=\varepsilon_{imn}\varepsilon_{jpq}f^{amp}f^{bnq},&
\mathsf G^{ab}{}_{ij}
&=\varepsilon^{amn}\varepsilon^{bpq}g_{imp}g_{jnq}.&&\nonumber
\end{align}
Then
\begin{equation}
N_0=-\mathsf F-\mathsf G-5\mathsf U+\tfrac12S+PS,
\qquad PN_0P=N_0,
\label{fs:eq:Nzero}
\end{equation}
satisfies the identity
\begin{equation}
J(N_0)+\mathscr Q(T)
=\Delta(-T^2P+3mP).
\label{fs:eq:intrinsicnormal}
\end{equation}
The basic contractions used to verify it are
\begin{equation}
T^2P=2\mathsf U-(m\I-S)(\I-P),\qquad
\Tr_1N_0=\Tr_2N_0=\tfrac12(m\I-3M).
\label{fs:eq:normalcontractions}
\end{equation}
They follow from
\(\varepsilon_{ijr}\varepsilon^{klr}
=\delta_i^k\delta_j^l-\delta_i^l\delta_j^k\)
and the complete symmetry of \(f,g\).
To display the cancellation by tensor type, split
\(T=T_++T_-\), where
\((T_+)^{ab}{}_{ij}=\varepsilon_{ijm}f^{abm}\) and
\((T_-)^{ab}{}_{ij}=\varepsilon^{abm}g_{ijm}\).
They obey
\(PT_+=T_+\), \(T_+P=-T_+\),
\(PT_-=-T_-\), \(T_-P=T_-\), and \(T_\pm^2=0\).
Grouping \eqref{fs:eq:Qexplicit} into its \(ff\), \(gg\), and
\(fg\) contractions gives
\begin{equation}
\begin{aligned}
\mathscr Q(T_+)&=J(\mathsf F),\qquad
\mathscr Q(T_-)=J(\mathsf G),\\
\mathscr B(T_+,T_-)
&=J(5\mathsf U-\tfrac12S-PS)\\
&\quad+\Delta\big[-2\mathsf U+(m\I-S)(\I-P)+3mP\big].
\end{aligned}
\label{fs:eq:quadraticsectors}
\end{equation}
The pure contractions use complete symmetry of \(f\) or \(g\)
and the three-dimensional identity
\[
\delta_i^a\varepsilon_{jkm}-\delta_j^a\varepsilon_{ikm}
+\delta_k^a\varepsilon_{ijm}-\delta_m^a\varepsilon_{ijk}=0.
\]
For the mixed contraction, contracting an upper and a lower epsilon
leaves \(\mathsf U\), its one-site contraction \(M\), and
\(m=\tr M\); placing the surviving permutation on the two-site
factor gives the displayed \(PS\) and \((\I-P)\) terms.
Thus \eqref{fs:eq:quadraticsectors} uses only the two epsilon rules
and the four permutation identities above. Adding its three lines
and using \eqref{fs:eq:normalcontractions} proves
\eqref{fs:eq:intrinsicnormal}.

For the pure dressing part, set
\(K_p=\Tel(A)+\Hel(B)\), \(C=[A,B]\), and \(T_B=\Tel(B)\).
Without requiring \(C=0\), one has
\begin{equation}
N_p=\tfrac12T_B^2P-\tfrac12(C\otimes\I+\I\otimes C),\qquad
J(N_p)+\mathscr Q(K_p)=-\Delta(K_p^2P).
\label{fs:eq:purenormal}
\end{equation}
This identity follows by moving one-site factors through \(P\)
with \(PX_1=X_2P\).

Under the stabilizer conditions already proved, a full particular
second-order correction is
\begin{equation}
N_*=N_0+N_p+[B_2,T],\qquad PN_*P=N_*.
\label{fs:eq:fullnormal}
\end{equation}
Two exact density identities establish its cross terms.
If \([k,X_1+X_2]=0\), then
\begin{equation}
F(k+\Tel X)-F(k)
=-\Delta[k-X_2,[k,X_2]].
\label{fs:eq:exacttel}
\end{equation}
To check the identity, put \(a=k_{12}\), \(b=k_{23}\),
\(y=X_2\). The inner commutator changes by \([a+b,y]\).
Jacobi identities cancel the mixed terms; the remainder is
\([a-y,[a,y]]+[b+y,[b,y]]\), which is the stated boundary difference.

If \([k,B_1+B_2]=0\), let \(D=e^{\tau B}\). Then
\begin{equation}
k^D=(\I\otimes D)k(\I\otimes D^{-1}),\qquad
F(k^D)=G F(k)G^{-1},\quad
G=D^{-2}\otimes D^{-1}\otimes\I.
\label{fs:eq:exacthelix}
\end{equation}
Indeed \(G\) conjugates both adjacent copies of \(k\) into the
corresponding copies of \(k^D\).
If a witness has the same symmetry, its boundary difference is
also conjugated to a boundary difference.
Apply these identities to \(P+\tau T+\tau^2N_0\).
The \(A\)--\(T\) cross term needs no correction, whereas the
\(B\)--\(T\) cross term contributes \([B_2,T]\).
Combining this fact with \eqref{fs:eq:intrinsicnormal} and
\eqref{fs:eq:purenormal} proves \eqref{fs:eq:fullnormal}.

The linear kernel from Appendix~\ref{fs:app:linear} now gives the
complete freedom in the actual second coefficient:
\begin{equation}
H_2=N_*+W+c_0\I+d_0P,\qquad PWP=-W.
\label{fs:eq:freeHtwo}
\end{equation}

\subsection{Third-order positivity forces commuting dressing data}

Let \(\mu_r\) send
\(X_1\otimes\cdots\otimes X_r\) to the ordered product
\(X_1\cdots X_r\). Expansion of the six terms of the double
commutator gives
\begin{equation}
\mu_3\Delta X=0,\qquad
\mu_3F(h)=[\mu_2(h^2),\mu_2(h)],\qquad
\mu_2(U)=\Tr_2(UP).
\label{fs:eq:multiplication}
\end{equation}
For instance, writing \(h=\sum_\alpha a_\alpha\otimes b_\alpha\),
the images of \(h_{23}h_{12}^2\) and \(h_{23}^2h_{12}\)
coincide after relabeling the three summation indices.
The images of \(h_{23}h_{12}h_{23}\) and
\(h_{12}h_{23}h_{12}\) likewise coincide. The two remaining terms
are the two ordered products on the right of \eqref{fs:eq:multiplication}.

The zero partial traces of \(T\), the stabilizer conditions,
and \eqref{fs:eq:normalcontractions} give
\begin{gather}
\mu_2K=-3B,\qquad
\mu_2(T^2)=3M-m\I,\qquad
\mu_2\{P,N_0\}=m\I-3M,\nonumber\\
\mu_2(K_p^2)=\mu_2\{P,N_p\}=-3C,\nonumber\\
\mu_2\{\Tel A,T\}
=\mu_2\{\Hel B,T\}
=\mu_2\{P,[B_2,T]\}=0.
\label{fs:eq:multiplicationcontractions}
\end{gather}
For odd \(W\), its two partial traces sum to zero and hence
\(\mu_2\{P,W\}=0\).
Consequently, up to a scalar matrix,
\(\mu_2(K^2+\{P,H_2\})=-6C\).
Since \(\mu_2P=3\I\), the third-order coefficient of
\eqref{fs:eq:multiplication} yields
\begin{equation}
0=[-6C,-3B]=-18[B,[A,B]].
\label{fs:eq:phithree}
\end{equation}
On the Hermitian real form, \(C^\dagger=C\).
Taking the trace pairing with \(A\) gives
\(0=\tr(A[B,C])=\tr(C^2)\), and positivity forces \(C=0\).

\subsection{The intrinsic third-order projection}

Let \(\ell Y=Y_{000,122}\). This functional vanishes on every
\(J(X)\) and \(\Delta X\). Moreover,
\begin{equation}
\ell\mathscr Q(U)=0\qquad\text{whenever }PUP=-U.
\label{fs:eq:oddQ}
\end{equation}
This is obtained from \eqref{fs:eq:Qexplicit} by pairing terms with
\(U^{ab}{}_{ij}=-U^{ba}{}_{ji}\).
Polarization gives \(\ell\mathscr B(K,W)=0\) for the free odd
term in \eqref{fs:eq:freeHtwo}. The free scalar and \(P\) terms
also vanish under this projection.

For the intrinsic \(T\), put
\[
\mathsf a=f^{000}f^{0mn}g_{mn2},\qquad
\mathsf b=f^{00m}f^{00n}g_{mn2},\qquad
\mathsf c=\varepsilon^{0ij}\varepsilon^{0mn}
 f^{00r}g_{imr}g_{jn2}.
\]
The entries in the next table admit short scalar contraction rules.
For an odd two-site \(K\), set \(N^+=(N+PNP)/2\). Expansion of
\eqref{fs:eq:Qexplicit} gives
\begin{equation}
\begin{split}
\tfrac12\ell\mathscr B(K,N)=\sum_{m=0}^2\big[
&K^{00}{}_{1m}(N^+)^{m0}{}_{22}
+(N^+)^{00}{}_{1m}K^{m0}{}_{22}\\
&+K^{00}{}_{m2}(N^+)^{m0}{}_{12}
-(N^+)^{00}{}_{m2}
 (2K^{m0}{}_{12}+K^{0m}{}_{12})\big].
\end{split}
\label{fs:eq:scalarbilinear}
\end{equation}
To derive this formula, replace one factor in each cubic product in
\(F\) by \(P\); its Kronecker deltas remove the intermediate pair
of indices. Pair the remaining terms using
\(K^{ab}{}_{ij}=-K^{ba}{}_{ji}\). The odd part of \(N\)
drops out by \eqref{fs:eq:oddQ}, leaving the four products above.
The cubic entry itself, with \((K^2)\) denoting an ordinary two-site
product, is
\begin{equation}
\begin{split}
\ell F(K)={}&\sum_m\big[
(K^2)^{00}{}_{1m}K^{m0}{}_{22}
-K^{00}{}_{1m}(K^2)^{m0}{}_{22}\\
&\hspace{17mm}+K^{00}{}_{m2}(K^2)^{0m}{}_{12}
-(K^2)^{00}{}_{m2}K^{0m}{}_{12}\big]\\
&+2\sum_{m,n,p}K^{00}{}_{mp}
\big[K^{0m}{}_{1n}K^{np}{}_{22}
-K^{p0}{}_{n2}K^{mn}{}_{12}\big].
\end{split}
\label{fs:eq:scalarcubic}
\end{equation}
This is just the six terms of the double commutator with their external
indices fixed. Equations~\eqref{fs:eq:scalarbilinear} and
\eqref{fs:eq:scalarcubic}, together with the two-site tensors
\eqref{fs:eq:normaldefinitions}, give every row of
\eqref{fs:eq:thirdtable} using sums over \(0,1,2\).

For example, the \(\mathsf U\) row splits into the following two
types of monomial:
\begin{align*}
\tfrac12\ell\mathscr B(T,\mathsf U)
={}&f^{000}(f^{01r}g_{12r}+f^{02r}g_{22r})
-f^{00r}(f^{001}g_{12r}+f^{002}g_{22r})
-3\mathsf b\\
&+f^{00r}(2g_{12r}g_{122}-g_{11r}g_{222}-g_{22r}g_{112})\\
={}&\mathsf a-4\mathsf b-\mathsf c.
\end{align*}
Repeated \(r\) is summed. The first two terms combine to
\(\mathsf a-\mathsf b\). The last parenthesis is
\[
-\varepsilon^{0ij}\varepsilon^{0mn}g_{imr}g_{jn2}.
\]
For \(S\), the surviving terms in \eqref{fs:eq:scalarbilinear} give
\((2+2-6)f^{000}M^0{}_2=-2\mathsf a\).
For \(PS\), the middle contribution is absent, giving
\((2-6)f^{000}M^0{}_2=-4\mathsf a\).
The \(\mathsf F,\mathsf G\) rows use the same four products and
\(\varepsilon_{ijr}\varepsilon^{klr}
=\delta_i^k\delta_j^l-\delta_i^l\delta_j^k\).

The required contractions are
\begin{equation}
\begin{array}{c|c}
\text{expression}&\ell\text{ applied to the expression}\\ \hline
F(T)&-\mathsf a\\
\mathscr B(T,\mathsf F)&-4\mathsf a+4\mathsf b\\
\mathscr B(T,\mathsf G)&-8\mathsf c\\
\mathscr B(T,\mathsf U)&2\mathsf a-8\mathsf b-2\mathsf c\\
\mathscr B(T,S)&-2\mathsf a\\
\mathscr B(T,PS)&-4\mathsf a
\end{array}
\label{fs:eq:thirdtable}
\end{equation}
Each row uses only \eqref{fs:eq:Qexplicit},
\eqref{fs:eq:tensor}, and \eqref{fs:eq:normaldefinitions}.
Inserting the five coefficients of \(N_0\) gives
\begin{equation}
\ell\bigl(F(T)+\mathscr B(T,N_0)\bigr)
=-6(2\mathsf a-6\mathsf b-3\mathsf c)
=-6E_+(f,g)^{0000}{}_2.
\label{fs:eq:thirdhighest}
\end{equation}
There is no trace correction in this component, since its upper
indices are all \(0\) and its lower index is \(2\).

It remains to show that the same projection applies in the
presence of \(A,B\). We have already proved \([A,B]=0\).
Apply \eqref{fs:eq:exacthelix} to
\(k=P+\tau T+\tau^2N_0\), and then add \(\tau\Tel A\).
The first two coefficients of the resulting arc are \(K,N_*\).
The equivariant contraction \(N_0\) inherits both symmetries.
Witnesses through second order inherit them as well: after fixing
their scalar ambiguity, this follows by applying the symmetry
commutator to \(\Delta X\), whose kernel consists of scalars.
Equations \eqref{fs:eq:exacttel}--\eqref{fs:eq:exacthelix} therefore
preserve the third-order residual modulo \(\operatorname{im}\Delta\),
because \(G(0)=\I\).

An actual arc has the additional freedom \eqref{fs:eq:freeHtwo},
which \(\ell\) annihilates at third order.
Equations \eqref{fs:eq:jet23} and \eqref{fs:eq:thirdhighest} thus imply
\(E_+^{0000}{}_2=0\).
The same holds after every common \(SU(3)\) transformation.
This extremal component generates the dual of the irreducible
35-dimensional module, so \(E_+=0\).
On the Hermitian real form \(E_-=0\) follows by conjugation.
This completes the proof of Lemma~\ref{fs:lem:necessary}.

\section{Exhaustive Hermitian intrinsic normal forms}
\label{fs:app:intrinsic}

\begin{lemma}
\label{fs:lem:intrinsic}
A nonzero \(f\in S^3\C^3\) satisfies \(E_+(f,\bar f)=0\)
if and only if it has one of the two \(SU(3)\) normal-form types
\eqref{fs:eq:triangle} and \eqref{fs:eq:secant}.
\end{lemma}

The proof proceeds through a chart obtained from a global maximum of
\(|f|\), so every nonzero tensor is covered. Its transverse Takagi
singular values obey \(0\le s\le r\le1/2\). We first exclude
\(r=0\), then rule out \(r>s\), including \(s=0\) and the endpoint
\(r=1/2\). For \(r=s>0\), the remaining equations split into the
cases \(e=0\) and \(e\ne0\), which yield the two forms. We then
check their converse, disjointness, and stabilizers.

\subsection{A compact maximum and its transverse singular values}

Let \(a>0\) be the maximum of \(|f|\) on the complex unit sphere.
Change the phase of a maximizing vector, move it to \(e_0\) by
\(SU(3)\), and divide \(f\) by the positive number \(a\).
The equation remains valid by real homogeneity.
We then have \(f^{000}=1\). First variation in the two orthogonal
complex directions gives \(f^{001}=f^{002}=0\).
Takagi factorization on that two-plane, with the determinant phase
separated from the \(SU(2)\) change of basis, gives
\begin{equation}
\begin{split}
f={}&x^3+3x(by^2+dz^2)+ey^3+3ty^2z+3uyz^2+vz^3,\\
b={}&r\eta,\qquad d=s\eta,\qquad
|\eta|=1,\qquad 0\le s\le r\le\tfrac12.
\end{split}
\label{fs:eq:maximumchart}
\end{equation}
The common phase \(\eta\) is retained because the residual change
of basis is in \(SU(2)\).
To see the upper bound, evaluate \(f\) on
\(\sqrt{1-\epsilon^2}e_0+\epsilon e^{i\phi}e_j\).
The maximal second-order coefficient of its modulus is
\(-3/2+3|f^{0jj}|\), which must be nonpositive.

Here is the coefficient extraction behind the chart equations.
At \(x=e_0\), the gradient and Hessian of the cubic are
\[
\nabla\phi=(3,0,0)^{\mathsf T},\qquad
\nabla^2\phi=6\operatorname{diag}(1,b,d),\qquad
G_0=\operatorname{diag}(1,\bar b,\bar d).
\]
The last two diagonal entries of
\(\Omega(e_0)G_0\Omega(e_0)\) are \(-\bar d,-\bar b\).
For \(e=1,2\), the trace-projection term
\(x_e\operatorname{div}\mathcal P/6\) vanishes at \(e_0\).
Equation~\eqref{fs:eq:polynomialcovariant} therefore gives, for example,
\[
\mathcal E_1(e_0)
=2(b\bar e+d\bar u)-3(\bar d\bar e+\bar b\bar u).
\]
The other equation follows with \((\bar e,\bar u)\) replaced by
\((\bar t,\bar v)\). For the next equation, extraction of \(x^3y\)
in \(\mathcal E_2\), divided by two, gives
\[
2E_+^{0001}{}_2
=e\bar t+2t\bar u+u\bar v
 +3b(\bar t\bar u-\bar e\bar v).
\]
After \(u=\rho e\), \(v=\rho t\), the last parenthesis vanishes,
leaving \eqref{fs:eq:harmonic}.

The components
\(E_+^{0000}{}_1=E_+^{0000}{}_2=0\) give
\begin{equation}
\begin{aligned}
(2b-3\bar d)\bar e+(2d-3\bar b)\bar u&=0,\\
(2b-3\bar d)\bar t+(2d-3\bar b)\bar v&=0.
\end{aligned}
\label{fs:eq:chartlinear}
\end{equation}
If \(r=s=0\), another component is
\[
E_+^{0000}{}_0
=-\frac{4+|e|^2+3|t|^2+3|u|^2+|v|^2}{3}<0,
\]
which is impossible. Thus \(r>0\). Define
\begin{equation}
\zeta=\eta^2,\qquad
\rho=\frac{3s\zeta-2r}{2s-3r\zeta}.
\label{fs:eq:rho}
\end{equation}
The denominator is nonzero because \(3r>2s\).
Taking conjugates in \eqref{fs:eq:chartlinear} gives
\(u=\rho e,\ v=\rho t\).
The equation \(2E_+^{0001}{}_2=0\) now reduces to
\begin{equation}
(1+|\rho|^2)e\bar t+2\bar\rho\,t\bar e=0.
\label{fs:eq:harmonic}
\end{equation}

\subsection{Unequal singular values are impossible}

Suppose \(r>s\).
The difference between the squared moduli of the denominator and
numerator in \eqref{fs:eq:rho} is \(5(r^2-s^2)>0\).
Hence \(|\rho|<1\). Taking absolute values in \eqref{fs:eq:harmonic}
shows that \(e\bar t=0\), since
\(1+|\rho|^2>2|\rho|\).
There are therefore only two cases:
\(t=v=0\) or \(e=u=0\).

Let \(D_0=r^2-s^2\), \(R_0=|\rho|^2\), and set
\[
c_1=-4(b+\bar d)-2D_0(3b-\bar d),\qquad
c_2=-4(\bar b+d)+2D_0(3d-\bar b).
\]
The two coefficient extractions
\(6E_+^{0011}{}_0=[x^2y^2]\mathcal E_0\) and
\(6E_+^{0022}{}_0=[x^2z^2]\mathcal E_0\) give the master equations
\begin{align}
6E_+^{0011}{}_0={}&c_1+(8b-3\bar d)|e|^2+(12b+\bar d)|t|^2
 +(3b+4\bar d)|u|^2-b|v|^2\nonumber\\
&+(3d-7\bar b)(e\bar u+t\bar v),\nonumber\\
6E_+^{0022}{}_0={}&c_2-d|e|^2+(3d+4\bar b)|t|^2
 +(12d+\bar b)|u|^2+(8d-3\bar b)|v|^2\nonumber\\
&+(3b-7\bar d)(u\bar e+v\bar t).
\label{fs:eq:chartmaster}
\end{align}
Substitution of \(u=\rho e\), \(v=\rho t\), and the two possible
zero pairs in \eqref{fs:eq:chartmaster} yields the table below.
The equations \(6E_+^{0011}{}_0=6E_+^{0022}{}_0=0\) take the
form \(c_1+a_1w=c_2+a_2w=0\).
Their coefficients are as follows:
\begin{equation}
\begin{array}{c|c|l}
\text{case}&w&(a_1,a_2)\\ \hline
t=v=0&|e|^2&
\begin{aligned}
a_1&=b(8+3R_0)+(3d-7\bar b)\bar\rho+\bar d(4R_0-3),\\
a_2&=(3b-7\bar d)\rho+\bar bR_0+d(12R_0-1);
\end{aligned}\\[4pt]
e=u=0&|t|^2&
\begin{aligned}
a_1&=b(12-R_0)+(3d-7\bar b)\bar\rho+\bar d,\\
a_2&=\bar b(4-3R_0)+d(3+8R_0)+(3b-7\bar d)\rho .
\end{aligned}
\end{array}
\label{fs:eq:chartcoefficients}
\end{equation}
Eliminating the single real unknown \(w\) requires
\(c_1a_2-c_2a_1=0\).
Using \eqref{fs:eq:rho}, the two determinants factor respectively as
\begin{equation}
\frac{-8D_0H(r,s,\zeta)}
{\zeta(3r\zeta-2s)(2s\zeta-3r)},
\qquad
\frac{-8D_0H(s,r,\zeta)}
{\zeta(3r\zeta-2s)(2s\zeta-3r)},
\label{fs:eq:chartdeterminants}
\end{equation}
where
\begin{align}
H(r,s,\zeta)={}&72r^2s^2\zeta^4
 -(39r^3s+69rs^3)\zeta^3\nonumber\\
&+(-6r^4-34r^2s^2-16s^4+10r^2+10s^2)\zeta^2\nonumber\\
&+(45r^3s+55rs^3+20rs)\zeta-24r^2s^2.
\label{fs:eq:Hpolynomial}
\end{align}
For completeness, the determinant factorization can be checked without
division or a resultant. Introduce the four linear polynomials
\[
U=2s-3r\zeta,\quad V=2s\zeta-3r,\quad
M=3s\zeta-2r,\quad N=3s-2r\zeta.
\]
Then \(\rho=M/U\), \(\bar\rho=N/V\), and \(R_0=MN/(UV)\).
Multiplying the two equations by \(\eta UV\) gives
\(C_1UV+A_1w=C_2UV+A_2w=0\) in the first case and
\(C_1UV+B_1w=C_2UV+B_2w=0\) in the second, where
\begin{align}
C_1&=-4(r\zeta+s)-2D_0(3r\zeta-s),\nonumber\\
C_2&=-4(r+s\zeta)+2D_0(3s\zeta-r),\nonumber\\
A_1&=r\zeta(8UV+3MN)+(3s\zeta-7r)NU+s(4MN-3UV),\nonumber\\
A_2&=(3r\zeta-7s)MV+rMN+s\zeta(12MN-UV),\nonumber\\
B_1&=r\zeta(12UV-MN)+(3s\zeta-7r)NU+sUV,\nonumber\\
B_2&=r(4UV-3MN)+s\zeta(3UV+8MN)+(3r\zeta-7s)MV.
\label{fs:eq:clearedchartrows}
\end{align}
Here \(C_j=\eta c_j\), while \(A_j\) or \(B_j\) equals
\(\eta UVa_j\). Distributivity gives the polynomial identities
\begin{align}
C_1A_2-C_2A_1&=8D_0H(r,s,\zeta),\nonumber\\
C_1B_2-C_2B_1&=8D_0H(s,r,\zeta).
\label{fs:eq:clearedchartdeterminants}
\end{align}
Thus \eqref{fs:eq:chartdeterminants} follows on dividing by
\(\eta^2UV=\zeta UV\). Alternatively, multiply the two cleared
equations by \(A_2,-A_1\), or by \(B_2,-B_1\), and add.
This eliminates \(w\) directly. The identities hold over every
commutative ring; the nonvanishing argument below is the step that
uses the Hermitian maximum chart.

All denominators in \eqref{fs:eq:chartdeterminants} are nonzero.
Indeed \(|U|,|V|\ge3r-2s>0\), and \(|\zeta|=1\).
If \(s=0\), the two \(H\) factors reduce to
\(\zeta^2r^2(10-6r^2)\) and
\(\zeta^2r^2(10-16r^2)\). Both are nonzero for \(0<r\le1/2\).

For \(rs>0\), a zero of \(H\) on the unit circle would also obey
\begin{align}
0&=H(r,s,\zeta)-\zeta^4H(r,s,\zeta^{-1})\nonumber\\
&=4rs(\zeta^2-1)
\left[24rs(\zeta^2+1)-(21r^2+31s^2+5)\zeta\right].
\label{fs:eq:Hantisym}
\end{align}
Apart from \(\zeta=\pm1\), this would require
\[
48rs\,\operatorname{Re}\zeta=21r^2+31s^2+5.
\]
It is impossible by the explicit identity
\[
21r^2+31s^2+5-48rs
=5+\frac{(21r-24s)^2+75s^2}{21}>0.
\]
At the two remaining phases,
\begin{align}
H(r,s,1)&=-2(r+s)^2(3r^2-9rs+8s^2-5),\nonumber\\
H(r,s,-1)&=-2(r-s)^2(3r^2+9rs+8s^2-5).
\label{fs:eq:Hendpoints}
\end{align}
The first is nonzero because \(3r^2+8s^2\le11/4<5\).
For the second, \(3r^2+9rs+8s^2\le5\), with equality only at
\(r=s=1/2\), already excluded by \(r>s\).
The same estimates hold with \(r,s\) interchanged.
Thus neither determinant can vanish, a contradiction.
We have proved
\begin{equation}
r=s>0.
\label{fs:eq:equalsingular}
\end{equation}

\subsection{The two remaining forms}

Now \(b=d\), \(u=-e\), \(v=-t\), and \eqref{fs:eq:harmonic} gives
\(e\bar t=t\bar e\).
The coefficients \(e,t\) have a common phase.
A real \(SO(2)\) rotation acts with three times the angle on the
binary harmonic cubic, so it can set \(t=v=0\).
The form is reduced to
\[
f=x^3+3bx(y^2+z^2)+e(y^3-3yz^2).
\]
Four components of \eqref{fs:eq:covariant} suffice:
\begin{align}
E_+^{1222}{}_2&=-\tfrac14e(2|e|^2-1),\nonumber\\
E_+^{0011}{}_0&=\tfrac23(b+\bar b)(2|e|^2-1),\nonumber\\
E_+^{0000}{}_0
&=\tfrac43(6|b|^2-3b\bar e^{\,2}-3\bar b^{\,2}-|e|^2-1),\nonumber\\
E_+^{1111}{}_0
&=-3(2b^2+2|b|^2+2\bar b e^2-|e|^2).
\label{fs:eq:lastfour}
\end{align}
If \(e=0\), these equations give \(b+\bar b=0\) and
\(9|b|^2=1\), so \(b=\pm i/3\).
The planar \(SU(2)\) matrix
\[
\frac1{\sqrt2}\begin{pmatrix}1&\mp i\\ \mp i&1\end{pmatrix}
\]
transforms the cubic into \(x(x^2+2yz)\).
Restoring the real scale and cyclically relabeling coordinates
gives \eqref{fs:eq:secant}.

If \(e\ne0\), the first equation in \eqref{fs:eq:lastfour}
gives \(|e|^2=1/2\).
The conjugate of the last equation says
\[
3b\bar e^{\,2}
=\tfrac34-3\bar b^{\,2}-3|b|^2.
\]
Substitution into the third gives \(9|b|^2=9/4\).
The last equation then yields \(e^2=-4b^3\).
Since \(|b|^2=1/4\), set \(\lambda=-e/b\). The relation
\(e^2=-4b^3\) gives directly
\(\lambda^2=-4b\) and \(\lambda^3=4e\), while
\(|\lambda|^2=|e|^2/|b|^2=2\).
We obtain the factorization
\begin{equation}
f=(x+\lambda y)
\left(x-\frac{\lambda y}{2}+\frac{\sqrt3\lambda z}{2}\right)
\left(x-\frac{\lambda y}{2}-\frac{\sqrt3\lambda z}{2}\right).
\label{fs:eq:orthogonalfactors}
\end{equation}
Since \(|\lambda|^2=2\), the three coefficient vectors are pairwise
orthogonal in the Hermitian inner product.
Normalize them and compensate the determinant of their unitary
frame. This gives an \(SU(3)\) transformation to
\(\alpha x_0x_1x_2\), with \(\alpha\ne0\).

\subsection{Converse, disjointness, and stabilizers}

Both normal forms lie in the slice
\(f=6u x_0x_1x_2+v x_2^3\).
Define
\[
a_0=\frac{|v|^2}{3}(u-\bar u),\qquad
b_0=-4u\bar v(u-\bar u),\qquad
c_0=\frac{4v}{3}(6|u|^2+3\bar u^{\,2}-|v|^2).
\]
Using \eqref{fs:eq:polynomialcovariant}, the entire covariant on
this slice factors as
\begin{equation}
\begin{aligned}
\mathcal E_0&=2x_0x_2(6a_0x_0x_1-c_0x_2^2),\\
\mathcal E_1&=2x_1x_2(6a_0x_0x_1-c_0x_2^2),\\
\mathcal E_2&=6b_0x_0^2x_1^2-24a_0x_0x_1x_2^2+c_0x_2^4.
\end{aligned}
\label{fs:eq:conversecomponents}
\end{equation}
For the triangle, \(v=0\). For the secant,
\(u=c\in\R\), \(v=3c\).
In both cases all three quartic polynomials vanish.

The contraction \(M\) of \eqref{fs:eq:contractionM} is
\(2|u|^2\I\) on the triangle and
\(c^2\operatorname{diag}(2,2,11)\) on the secant.
This proves \eqref{fs:eq:fingerprint} and disjointness of the two types.

Unique factorization of the triangle makes its connected stabilizer
preserve each of the three orthogonal factor lines.
In \(SU(3)\), the resulting group is exactly the diagonal Cartan torus.
For the secant, the distinct eigenspaces of \(M\) preserve the
\(x_2\) line and its orthogonal two-plane.
Preserving \(x_2(2x_0x_1+x_2^2)\) makes the phase on the \(x_2\)
line a cube root of unity, hence equal to one in the connected group.
On the two-plane, preservation of \(x_0x_1\) gives opposite diagonal
phases. This is exactly \eqref{fs:eq:arithmetic}.
The proof of Lemma~\ref{fs:lem:intrinsic} is complete.

\subsection{A standard elliptic degeneration and the Hermitian boundary}
\label{fs:app:elliptic}

The \(Z_3\)-Belavin model provides an integrable three-state family
with discrete on-site symmetries~\cite{hao}. To compare it with the
present local theorem, both the density normalization and its real
form must be fixed. We examine the crossing-parameter degeneration
of the weights in Ref.~\cite[Section~B, Eqs.~(B.2) and (B.3)]{mansson}.
The elliptic modulus \(\Omega\) is initially fixed with
\(\operatorname{Im}\Omega>0\); \(\gamma\) denotes the crossing parameter.

Put \(\omega=e^{2\pi i/3}\), let
\(X\ket j=\ket{j-1}\) and \(Z\ket j=\omega^{j-1}\ket j\),
with indices modulo three, and define
\[
S_{nm}=\omega^{nm}X^nZ^m,\qquad
T_{nm}=S_{nm}\otimes S_{-n,-m},\qquad
\alpha_{nm}=\frac{n+m\Omega}{3}.
\]
These conventions give \(\sum_{n,m=0}^2T_{nm}=3P\).
Removing a common scalar factor from the standard weights and
rescaling their spectral argument to \(3\gamma u\) gives the
regular braided solution
\begin{equation}
\check R_{\gamma,\Omega}(u)
=\frac{P}{3}\sum_{n,m=0}^2
e^{2\pi i m\gamma u}
\frac{\theta_1(\gamma(1+3u)+\alpha_{nm}\mid\Omega)}
{\theta_1(\gamma+\alpha_{nm}\mid\Omega)}T_{nm}.
\label{fs:eq:belavin-scaled}
\end{equation}
Canceling the simple zero of \(\theta_1\) at the origin gives a
joint analytic extension of the zero-torsion ratio, with value
\(1+3u\) at \(\gamma=0\). Every other denominator is nonzero
there, and its ratio tends to one. Consequently
\(\check R_{0,\Omega}(u)=\I+uP\).
Thus the base point and joint analyticity alone do not exclude
this degeneration.

Write \(E(z)=\theta_1'(z\mid\Omega)/\theta_1(z\mid\Omega)\).
Its density and first variation are
\begin{equation}
\begin{split}
h_{\gamma,\Omega}
&=\gamma P\sum_{n,m=0}^2
\left(\frac{2\pi i m}{3}+E(\gamma+\alpha_{nm})\right)T_{nm}
=P+\gamma K_\Omega+O(\gamma^2),\\
K_\Omega
&=P\sum_{(n,m)\ne(0,0)}
\left(\frac{2\pi i m}{3}+E(\alpha_{nm})\right)T_{nm}.
\end{split}
\label{fs:eq:belavin-density}
\end{equation}
The zero-torsion term has no linear contribution because
\(\gamma E(\gamma)=1+O(\gamma^2)\).

Here the first variation can be identified exactly in the tensor
coordinates of \eqref{fs:eq:tensor}. Set
\[
Q=e^{2\pi i\Omega},\qquad \xi=e^{2\pi i\Omega/3},\qquad
\mathfrak a(q)=\sum_{r,l\in\mathbb Z}q^{r^2-rl+l^2},\qquad
\mathcal U=\mathfrak a(Q),\quad
\mathcal V=\frac{\mathfrak a(\xi)-\mathcal U}{2}.
\]
The needed three-torsion identities are
\begin{equation}
E(1/3)=\frac{\pi}{\sqrt3}\mathcal U,\qquad
E((\Omega+j)/3)+\frac{2\pi i}{3}
=-\frac{i\pi}{3}(\mathcal U+2\omega^j\mathcal V),\quad j=0,1,2.
\label{fs:eq:belavin-torsion}
\end{equation}
For their derivation, logarithmic differentiation of the convergent
theta product, with \(x=e^{2\pi iz}\), gives
\[
\frac{E(z)}{2\pi i}
=-\frac12-\frac{x}{1-x}
-\sum_{n\ge1}\frac{Q^nx}{1-Q^nx}
+\sum_{n\ge1}\frac{Q^n/x}{1-Q^n/x}.
\]
Evaluate this at \(z=1/3,(\Omega+j)/3\) and use the Lambert
expansion~\cite[Eq.~(2.11)]{cooper}
\[
\mathfrak a(q)=1+6\sum_{k\ge0}
\left(\frac{q^{3k+1}}{1-q^{3k+1}}
-\frac{q^{3k+2}}{1-q^{3k+2}}\right).
\]
This yields \(E(1/3)=\pi\mathfrak a(Q)/\sqrt3\) and
\(E((\Omega+j)/3)+2\pi i/3=-i\pi\mathfrak a(\omega^j\xi)/3\).
Finally, \(r^2-rl+l^2\equiv(r+l)^2\pmod3\).
The residue-zero sublattice is parametrized by
\((r,l)=(a+b,-a+2b)\) and contributes \(\mathfrak a(Q)\);
the other terms all have exponent one modulo three. Hence
\(\mathfrak a(\omega^j\xi)=\mathcal U+2\omega^j\mathcal V\), proving
\eqref{fs:eq:belavin-torsion}.

Using oddness and quasiperiodicity of \(E\), substitute these
identities into the eight nonzero terms of
\eqref{fs:eq:belavin-density}. The resulting coordinates are
\begin{equation}
\begin{gathered}
A=B=0,\qquad g=0,\qquad
f_\Omega=a_\Omega(x_0^3+x_1^3+x_2^3)+6b_\Omega x_0x_1x_2,\\
a_\Omega=\frac{2\pi}{\sqrt3}(\mathcal U-\mathcal V),\qquad
b_\Omega=-\frac{\pi}{\sqrt3}(\mathcal U+2\mathcal V),\qquad
K_\Omega=K_{f_\Omega,0}.
\end{gathered}
\label{fs:eq:belavin-firstjet}
\end{equation}
The finite identification can be seen on a single three-dimensional
block. If \(k_{nm}=2\pi im/3+E(\alpha_{nm})\), then
\[
PT_{nm}|jk\rangle
=\omega^{m(2n+j-k)}|k+n,j-n\rangle .
\]
In the ordered basis \((00,12,21)\), summing these phases gives
\begin{equation}
K_\Omega\big|_{(00,12,21)}
=\begin{pmatrix}0&a_\Omega&-a_\Omega\\
0&b_\Omega&-b_\Omega\\0&b_\Omega&-b_\Omega\end{pmatrix}.
\label{fs:eq:belavinblock}
\end{equation}
For example, oddness and quasiperiodicity give
\(k_{n2}=-k_{-n,1}\), with indices taken modulo three, and
\(k_{10}=-k_{20}=\pi\mathcal U/\sqrt3\). Consequently
\[
\sum_{m=0}^2 k_{1m}\omega^m
=\frac{2\pi}{\sqrt3}(\mathcal U-\mathcal V),\qquad
\sum_{m=0}^2 k_{2m}
=-\frac{\pi}{\sqrt3}(\mathcal U+2\mathcal V).
\]
The zero column and the negative second column use the same sums and
\(1+\omega+\omega^2=0\). Common cyclic shift supplies the other two
blocks. Equation~\eqref{fs:eq:belavinblock} is precisely the block of
\(K_{f_\Omega,0}\), proving \eqref{fs:eq:belavin-firstjet}.
The Hermitian condition \eqref{fs:eq:realform} would require
\(0=g=\bar f_\Omega\). Equivalently,
\(K_{f_\Omega,0}\) maps \(\Lambda^2V\) into \(\Sym^2V\)
and vanishes on \(\Sym^2V\), so its square is zero.
Every nonzero such matrix is non-Hermitian.

The same conclusion holds for an analytic path
\((\gamma(\tau),\Omega(\tau))\) with \(\gamma(0)=0\)
and finite \(\Omega(0)\) in the upper half-plane: since
\(h_{0,\Omega}=P\), its first derivative is
\(\gamma'(0)K_{\Omega(0)}\). It is either zero or a nonzero
square-zero matrix. This identifies the Hermitian obstruction within
this regular elliptic chart; other singular joint degenerations are
outside the scope of this comparison.

\section{Exact analytic realizations of every direction}
\label{fs:app:lifts}

\subsection{The Cartan arc and its closed spectral matrix}

Write \(\alpha=6(a+ib)\), with \(a,b\in\R\), and define
\[
\kappa=6a\tau,\qquad
\ell=\sqrt{1+\kappa^2},\qquad
p=e^{2ib\tau},\qquad q=e^{-2ib\tau}.
\]
The branch of the square root has \(\ell(0)=1\).
Define \(h_T(\tau)\) to have value \(\ell\) on each \(\ket{aa}\).
In each of the cyclically ordered blocks
\((01,10)\), \((12,21)\), and \((20,02)\), let it be
\begin{equation}
\begin{pmatrix}\kappa/3&q\\ p&-\kappa/3\end{pmatrix}.
\label{fs:eq:triangle-density}
\end{equation}
It is Hermitian for real \(\tau\), equals \(P\) at zero, and has
\[
h_T'(0)=K_{f,\bar f},\qquad f=\alpha x_0x_1x_2.
\]
The factor of six here follows from the tensor convention
\(f^{012}=\alpha/6\).

This density is a specialization of the first fifteen-vertex family
of Vieira~\cite{vieira}. In the notation of Eq.~(10) of that reference,
the parameter map is
\begin{equation}
\begin{gathered}
\alpha_{11}=\alpha_{55}=\alpha_{99}=\ell,\qquad
\alpha_{24}=\alpha_{68}=\kappa/3,\qquad
\alpha_{42}=\alpha_{37}=-\kappa/3,\\
\alpha_{22}=\alpha_{66}=q,\qquad \alpha_{33}=p.
\end{gathered}
\label{fs:eq:triangle-vieira}
\end{equation}
The derived parameters satisfy \(\omega=\kappa\) and
\(\omega\Omega=\ell^2-\kappa^2=pq=1\), so the Hamiltonian matrices
agree entry by entry, including the cyclic ordering of the \((20,02)\)
block. The case \(\kappa=0\) is obtained by analytic continuation.
The spectral matrix below uses the braided Yang--Baxter convention of
Sec.~\ref{main:sec:setting}.

A closed regular braided matrix for this density has entry
\begin{equation}
A(u)=\cosh(\kappa u)+\ell\frac{\sinh(\kappa u)}{\kappa}
\label{fs:eq:triangleA}
\end{equation}
on each equal-state sector and blocks
\begin{equation}
\begin{pmatrix}
e^{\kappa u/3}&q\,\sinh(\kappa u)/\kappa\\
p\,\sinh(\kappa u)/\kappa&e^{-\kappa u/3}
\end{pmatrix}
\label{fs:eq:triangle-R}
\end{equation}
in the three cyclic sectors above.
The quotient \(\sinh(\kappa u)/\kappa\) is continued analytically
to \(u\) at \(\kappa=0\).
Thus every displayed entry is jointly analytic in \((\tau,u)\).
Regularity and the prescribed spectral derivative follow directly.

For completeness, the scalar identities underlying the full
Yang--Baxter equation can be displayed without expanding a
\(27\times27\) residual. Put \(S_u=\sinh(\kappa u)/\kappa\).
Since \(pq=1=\ell^2-\kappa^2\), hyperbolic addition gives
\begin{align}
A(u)A(v)-pqS_uS_v&=A(u+v),\nonumber\\
S_{u+v}A(u)-S_uA(u+v)&=S_v,\nonumber\\
S_{u+v}&=e^{\kappa v}S_u+e^{-\kappa u}S_v.
\label{fs:eq:triangleidentities}
\end{align}
Also \(e^{\pm\kappa(u+v)/3}
=e^{\pm\kappa u/3}e^{\pm\kappa v/3}\).
Separate three-site words into those with one, two, or three distinct
labels. One-label words give identical scalar products. On the
three-dimensional spaces with two labels, matrix multiplication
reduces to the first two identities of \eqref{fs:eq:triangleidentities}
and their \(u,v\) interchanges. On the six permutations of three
distinct labels, it reduces to the last identity and its interchanges;
the cyclic diagonal factors occur in cubes,
\((e^{\pm\kappa u/3})^3=e^{\pm\kappa u}\).
The constant exchange factors match on the two sides.
This verifies the full Yang--Baxter equation stated in
Sec.~\ref{main:sec:setting}.
The spectral matrix preserves the whole Cartan torus.

\subsection{A rational Hermitian secant arc}

Let \(D=s^4-s^2+1\). In this subsection the lowercase letters
below denote scalar matrix entries:
\begin{align}
b&=-\frac{s^4+s^2r-2s^2+1}{s^2},&
f&=\frac{s^2}{D},&
g&=-\frac{s^3(s^2-1)}{D},\nonumber\\
e&=-\frac{2s^8+2s^6r-4s^6-2s^4r+5s^4+2s^2r-4s^2+1}
{s^2D},\nonumber\\
k&=\frac{s^6+2s^4r-2s^4-2s^2r+3s^2+2r-2}{D},&
i&=\frac{s^2-1}{sD},\nonumber\\
j&=-\frac{s^8-3s^6+3s^4-3s^2+1}{s^2D}.
\label{fs:eq:secantentries}
\end{align}
In the lexicographic order \(00,01,02,10,11,12,20,21,22\), define
\begin{equation}
H_{\rm sec}(s,r)=
\begin{pmatrix}
1&0&0&0&0&0&0&0&0\\
0&e&0&f&0&0&0&0&g\\
0&0&b&0&0&0&1&0&0\\
0&f&0&k&0&0&0&0&i\\
0&0&0&0&1&0&0&0&0\\
0&0&0&0&0&r&0&1&0\\
0&0&1&0&0&0&r&0&0\\
0&0&0&0&0&1&0&b&0\\
0&g&0&i&0&0&0&0&j
\end{pmatrix}.
\label{fs:eq:secantmatrix}
\end{equation}
It is real symmetric for real parameters and satisfies
\(H_{\rm sec}(1,0)=P\). Every entry preserves the sum of
the one-site weights \(1,-1,0\).

The density belongs to the generalized Izergin--Korepin family in
Ref.~\cite{cfr}. Denote the matrix printed in its Eq.~(5.7),
Sec.~5.1.2, by \(\widetilde H_{\rm IK}(\zeta,\xi)\), where
\(\zeta=k\) and \(\xi=\tau'_p\) in that reference.
Define the common basis permutation and the one-site matrix
\[
W\ket0=\ket1,\qquad W\ket1=\ket2,\qquad W\ket2=\ket0,
\qquad q_*=\operatorname{diag}(-1,1,0).
\]
Direct substitution gives
\begin{equation}
\begin{split}
H_{\rm sec}(s,r)={}&\I-\frac{s^4-1}{s^2}
(W\otimes W)\widetilde H_{\rm IK}(s^2,1)
(W^\dagger\otimes W^\dagger)\\
&+(r-1+s^2)\Tel(q_*).
\end{split}
\label{fs:eq:secant-ik}
\end{equation}
The identity initially holds for \(s\ne1\) near 1 and has a removable
limit there. The rescaled matrix without the last term is
\(H_{\rm sec}(s,1-s^2)\); changing \(r\) adds exactly the displayed
telescope. This identifies the known model in the density normalization
used here. Appendix~\ref{fs:app:ikmatrix} lifts this parameter map to the
explicit spectral matrix and resolves its apparent singularity at
\(s=1\).

Choose
\begin{equation}
r(s)=\frac{(1-s^2)(s^4-s^2+2)}{3D},\qquad
s(\tau)=1-\frac{3c}{2}\tau,\qquad
h_{\rm sec}(\tau)=H_{\rm sec}(s(\tau),r(s(\tau))).
\label{fs:eq:secantarc}
\end{equation}
Differentiation gives
\[
\left.\frac{d}{ds}H_{\rm sec}(s,r(s))\right|_{s=1}
=K_{f_*,f_*},\qquad
f_*=-2x_2(2x_0x_1+x_2^2).
\]
Consequently \(h_{\rm sec}'(0)=K_{f,f}\) with
\(f=3c\,x_2(2x_0x_1+x_2^2)\), exactly as in
\eqref{fs:eq:secant}.

The density identity follows from the spectral matrix constructed in
Appendix~\ref{fs:app:ikmatrix}. The following general derivative argument
also gives its boundary witness. Write a regular additive braided
solution as \(R(u)=\I+uh+u^2a+u^3b+O(u^4)\). The coefficient of \(uv\)
in its Yang--Baxter equation gives
\(2\Delta a=\Delta(h^2)\); hence
\(a=h^2/2+c\I\), since \(\ker\Delta=\C\I\).
The coefficient of \(u^2v\), with \(X=h_{12}\), \(Z=h_{23}\), is
\[
0=3\Delta b-\tfrac12[X+Z,[X,Z]]
  -\tfrac12\Delta(h^3)-3c\Delta h.
\]
Thus an explicit two-site witness is
\begin{equation}
Y=R'''(0)-3hR''(0)+2h^3,\qquad
F(h)=Y_{23}-Y_{12}.
\label{fs:eq:secantwitness}
\end{equation}
For \(R=\check R_{\rm sec}(s,r;u)\) in
\eqref{fs:eq:iksecantR}, its first derivative is exactly
\(H_{\rm sec}(s,r)\), so
\begin{equation}
F(H_{\rm sec}(s,r))=Y_{23}-Y_{12}.
\label{fs:eq:secantresh}
\end{equation}
The derivatives in \eqref{fs:eq:secantwitness} involve only the displayed
two-site rational weights and exponential factors. A scalar multiple
of \(\I\) may be subtracted from \(Y\) to impose \(Y_{00,00}=0\).
The construction below uses the density map and the standard BMW
solution, independently of \eqref{fs:eq:secantresh}.

All denominators are nonzero near \(s=1\).
Thus \eqref{fs:eq:secantarc} is an exact analytic Hermitian
Reshetikhin arc, with the desired fixed \(U(1)\) symmetry.
Appendix~\ref{fs:app:ikmatrix} gives a jointly analytic regular
Yang--Baxter matrix with exactly this density and fixed symmetry.

\subsection{Adding all allowed one-site terms}

Let \(k_\tau\) be either intrinsic arc, with spectral matrix
\(\check R_{\rm int}(\tau,u)\).
Suppose
\[
A^\dagger=A,\quad B^\dagger=-B,\quad
A,B\in\Stab(f,\bar f),\quad[A,B]=0.
\]
The chosen intrinsic family preserves the common symmetry of \(A,B\).
Put \(D_\tau=e^{-\tau B}\) and define
\begin{equation}
h(\tau)=(D_\tau\otimes\I)k_\tau
 (D_\tau^{-1}\otimes\I)+\tau\Tel(A).
\label{fs:eq:combined-density}
\end{equation}
Its exact spectral realization is
\begin{equation}
\begin{split}
\check R(\tau,u)={}&
(\I\otimes e^{-\tau uA})(D_\tau\otimes\I)
\check R_{\rm int}(\tau,u)\\
&\hspace{7mm}\cdot(D_\tau^{-1}\otimes\I)
(e^{\tau uA}\otimes\I).
\end{split}
\label{fs:eq:combined-R}
\end{equation}

The spatial conjugation preserves the Yang--Baxter equation stated in
Sec.~\ref{main:sec:setting}: on three sites,
\(D_\tau^2\otimes D_\tau\otimes\I\) conjugates both adjacent
intrinsic matrices to their twisted versions, using their common
\(B\) symmetry. To verify the spectral factors, pass to the
unbraided matrix \(R=P\check R\) and apply the rapidity-dependent
one-site basis change \(g(z)=e^{-\tau zA}\).
The standard cancellation of the \(g(z)\) factors on each
rapidity line preserves the Yang--Baxter equation.
The common total-\(A\) symmetry makes the transformed matrix depend
only on the rapidity difference \(u\). Returning to the braided
convention gives the first and last factors in \eqref{fs:eq:combined-R}.
The relation \([A,B]=0\) ensures that the spatial twist retains
this total-\(A\) symmetry.

Regularity is immediate. The spectral derivative is
\eqref{fs:eq:combined-density}, which is Hermitian for real \(\tau\)
because \(D_\tau\) is unitary.
Taking its model derivative at zero gives
\[
h'(0)=K_{f,\bar f}+\Hel(B)+\Tel(A).
\]
All operations are jointly analytic.
On the pure dressing locus \(f=0\), take
\(k_\tau=P\) and \(\check R_{\rm int}=\I+uP\) in the same
formulas. This covers every point of \eqref{fs:eq:classification},
including the intrinsic origin.

\section{The explicit Izergin--Korepin matrix and its regular limit}
\label{fs:app:ikmatrix}

We construct the spectral matrix for \eqref{fs:eq:secantarc} directly
from the Izergin--Korepin solution~\cite{ik,jimbo}. The rational
normalization below follows Ref.~\cite{ffr}, with its anisotropy
\(k=s^2\) and the analytic choice \(\sqrt{k}=s\) near \(s=1\).

\subsection{The spectral matrix and the density map}

In lexicographic order, define the unbraided matrix
\begin{equation}
R_{\rm IK}(s,x)=
\begin{pmatrix}
1&0&0&0&0&0&0&0&0\\
0&b&0&c_-&0&0&0&0&0\\
0&0&f&0&d_-&0&h_-&0&0\\
0&c_+&0&b&0&0&0&0&0\\
0&0&d_+&0&g&0&d_-&0&0\\
0&0&0&0&0&b&0&c_-&0\\
0&0&h_+&0&d_+&0&f&0&0\\
0&0&0&0&0&c_+&0&b&0\\
0&0&0&0&0&0&0&0&1
\end{pmatrix},
\label{fs:eq:ikmatrix}
\end{equation}
where \(L=(x+s^6)(x-s^4)\) and
\begin{align}
b&=\frac{s^2(x-1)}{x-s^4},&
c_-&=\frac{1-s^4}{x-s^4},&c_+&=xc_-,\nonumber\\
d&=\frac{s(1-s^4)(x-1)}{L},&
d_-&=s^4d,&d_+&=-xd,\nonumber\\
f&=\frac{s^4(x+s^2)(x-1)}{L},&&&\nonumber\\
g&=\frac{s^2(x+s^6)(x-1)+x(s^6+1)(1-s^4)}{L},&&&\nonumber\\
h_-&=\frac{[x+s^6+s^4(x-1)](1-s^4)}{L},&&&\nonumber\\
h_+&=\frac{x[x+s^6-s^2(x-1)](1-s^4)}{L}.&&&
\label{fs:eq:ikweights}
\end{align}
The overall factor \(x\) in \(h_+\) restores an omission in
Eq.~(4.11) of the arXiv v1 of Ref.~\cite{ffr}.
With this factor, the displayed weights obey the full Yang--Baxter
identity and have the density map below.

For \(s\ne1\) near 1, \(R_{\rm IK}(s,1)=P\).
The braided matrix \(\check R_{\rm IK}=PR_{\rm IK}\) satisfies
\begin{equation}
\check R_{{\rm IK},12}(s,x)
\check R_{{\rm IK},23}(s,xy)
\check R_{{\rm IK},12}(s,y)
=
\check R_{{\rm IK},23}(s,y)
\check R_{{\rm IK},12}(s,xy)
\check R_{{\rm IK},23}(s,x).
\label{fs:eq:ikyb}
\end{equation}
Here is an analytic identification that proves this identity for independent
\(s,x,y\). Set \(q=s^{-2}\), \(\omega_q=q-q^{-1}\), and let
\(\sigma\) have the following blocks in the indicated two-site bases:
\begin{equation}
\begin{array}{c|c}
00\text{ or }22&q\\[2pt]
(01,10)\text{ or }(12,21)&
\begin{pmatrix}\omega_q&1\\1&0\end{pmatrix}\\[6pt]
(02,11,20)&
\begin{pmatrix}
(1-s^2)\omega_q&-s\omega_q&s^2\\
-s\omega_q&1&0\\s^2&0&0
\end{pmatrix}.
\end{array}
\label{fs:eq:ikbraidblocks}
\end{equation}
These blocks identify a standard spin-one BMW representation explicitly.
In the conventions of Ref.~\cite{bmw}, put \(\eta=-\log s\),
\(G=\operatorname{diag}(1,i,1)\), and \(G_2=G\otimes G\).
Then
\[
\sigma=G_2P\check R_{\rm const}(\eta)PG_2^{-1},
\]
where \(\check R_{\rm const}\) is Eq.~(II.22) of that reference.
The \(P\) conjugation reverses neighboring sites; reversing all three
sites preserves the braid relation. Common \(G\) conjugation also
preserves it. The BMW parameter is \(\nu=q^{-2}=s^4\).
For example, its rank-one element is
\[
E=vv^{\mathsf T},\qquad
v=s|02\rangle+|11\rangle+s^{-1}|20\rangle,\qquad
\sigma E=s^4E,\quad
\sigma-\sigma^{-1}=\omega_q(\I-E).
\]
These equalities can also be read directly from the blocks above.

Put \(B=s^2\sigma\). The weights in \eqref{fs:eq:ikweights} combine to
the single polynomial identity
\begin{equation}
L(s,x)\check R_{\rm IK}(s,x)
=x(x-1)B+s^{10}(x-1)B^{-1}
 +x(1-s^4)(1+s^6)\I.
\label{fs:eq:ikbaxterization}
\end{equation}
To obtain it, the coefficients of \(x^2\), \(x^0\), and the value at
\(x=1\) of the left-hand side are read from the two-site blocks:
they are \(B\), \(-s^{10}B^{-1}\), and
\((1-s^4)(1+s^6)\I\), respectively. These three data determine
the quadratic polynomial.

The standard BMW spectral combination is
\[
\mathcal S_q(z)=\tfrac12q^{3/2}(1-e^{-z})\sigma
-\tfrac12(q^{3/2}+q^{-3/2})(q-q^{-1})\I
-\tfrac12q^{-3/2}(1-e^z)\sigma^{-1}.
\]
It obeys additive braided Yang--Baxter, with the normalization of
Ref.~\cite{bmw}, Eq.~(III.9). Substituting \(z=-\log x\) and multiplying
by \(-2s^5x\) gives exactly the right-hand side of
\eqref{fs:eq:ikbaxterization}. Additivity becomes multiplication of
\(x,y\), and scalar normalizations cancel. This proves
\eqref{fs:eq:ikyb} locally and hence as a rational identity.
The three-site coefficient check in the accompanying code is an
independent verification of this analytic identification.

Use \(W\) and \(q_*\) from Appendix~\ref{fs:app:lifts} and put
\begin{equation}
\lambda(s)=\frac{1-s^4}{s^2},\qquad
x(s,u)=e^{\lambda(s)u},\qquad
\beta(s,r)=r-1+s^2.
\label{fs:eq:ikparameters}
\end{equation}
Define
\begin{equation}
\begin{split}
\check R_{\rm sec}(s,r;u)={}&(1+u)
(\I\otimes e^{-\beta(s,r)u q_*})
(W\otimes W)\check R_{\rm IK}(s,x(s,u))\\
&\hspace{8mm}\cdot(W^\dagger\otimes W^\dagger)
(e^{\beta(s,r)u q_*}\otimes\I).
\end{split}
\label{fs:eq:iksecantR}
\end{equation}
Since
\(W\operatorname{diag}(0,1,2)W^\dagger=\I-q_*\),
the exponential factors are the total-charge-preserving spectral
transformation proved in \eqref{fs:eq:combined-R}. The substitution
\(x=e^{\lambda u}\) turns \eqref{fs:eq:ikyb} into the additive equation;
the scalar factor \(1+u\) and common basis change also preserve it.
Differentiation at \(u=0\) gives the exact identity
\begin{equation}
\begin{split}
\partial_u\check R_{\rm sec}(s,r;0)
={}&\I+\lambda(s)(W\otimes W)
\left.P\partial_xR_{\rm IK}(s,x)\right|_{x=1}
(W^\dagger\otimes W^\dagger)\\
&+\beta(s,r)\Tel(q_*)
=H_{\rm sec}(s,r).
\end{split}
\label{fs:eq:ikdensity}
\end{equation}
Thus the spectral construction lifts the full two-parameter density
map \eqref{fs:eq:secant-ik}.

\subsection{Joint analyticity at the permutation point}

The apparent singularity at \(s=1\) is removable. Define the entire function
\[
\operatorname{exprel}(z)=\sum_{n=0}^{\infty}\frac{z^n}{(n+1)!},
\qquad
\Phi(s,u)=u\operatorname{exprel}(\lambda(s)u).
\]
Then
\begin{equation}
x-1=\lambda\Phi,\qquad
x-s^4=\lambda(s^2+\Phi),\qquad
\mathcal A=s^2+\Phi,\qquad \mathcal B=x+s^6.
\label{fs:eq:ikcancellation}
\end{equation}
Cancelling \(\lambda\) in \eqref{fs:eq:ikweights} gives
\begin{align}
b&=\frac{s^2\Phi}{\mathcal A},&
c_-&=\frac{s^2}{\mathcal A},&c_+&=xc_-,\nonumber\\
d&=\frac{s^3\lambda\Phi}{\mathcal A\mathcal B},&
d_-&=s^4d,&d_+&=-xd,\nonumber\\
f&=\frac{s^4(x+s^2)\Phi}{\mathcal A\mathcal B},&&&\nonumber\\
g&=\frac{s^2[\mathcal B\Phi+x(s^6+1)]}
{\mathcal A\mathcal B},&&&\nonumber\\
h_-&=\frac{s^2[\mathcal B+s^4\lambda\Phi]}
{\mathcal A\mathcal B},&&&\nonumber\\
h_+&=\frac{xs^2[\mathcal B-s^2\lambda\Phi]}
{\mathcal A\mathcal B}.&&&
\label{fs:eq:ikregularweights}
\end{align}
At \((s,u)=(1,0)\), the only denominators have
\(\mathcal A=1\), \(\mathcal B=2\), and \(s\ne0\).
All entries are therefore jointly analytic on a product neighborhood.
At \(s=1\), \(\Phi=u\) and the unbraided matrix becomes
\begin{equation}
\left.R_{\rm IK}(s,x(s,u))\right|_{s=1}=\frac{P+u\I}{1+u},\qquad
\check R_{\rm sec}(1,0;u)=\I+uP,
\label{fs:eq:iklimit}
\end{equation}
where the left side denotes the analytic extension just constructed.
Regularity holds throughout the neighborhood. The full Yang--Baxter
identity extends to \(s=1\) by analyticity, as do \eqref{fs:eq:ikdensity}
and the fixed total-\(q_*\) symmetry.

Finally, insert \(s(\tau)\) and \(r(s(\tau))\) from
\eqref{fs:eq:secantarc} into \eqref{fs:eq:iksecantR}. The resulting
\(\check R(\tau,u)\) is jointly analytic and regular, has the
Hermitian density \(h_{\rm sec}(\tau)\), and retains the same
nontrivial charge for every \((\tau,u)\). Together with the twists
in Appendix~\ref{fs:app:lifts}, this proves the secant sufficiency
statement for every allowed first derivative.

\section{Exact analytic continuation of the charge}
\label{fs:app:continuation}

We prove Theorem~\ref{fs:thm:main} by retaining the full analytic equations
near \(P\). The first-order classification supplies their leading
directions, but an additional argument is needed to control an arbitrary
given arc. We first construct genuine analytic equations with the required
initial terms. Nonzero intrinsic leaders are then treated by local
exhaustion with the exact families. The remaining leaders are treated by
equivariant elimination and a positive scalar identity.

\subsection{An equivariant analytic section}
\label{fs:app:equivariantsection}

Project \(h\) orthogonally onto \(\Span\{\I,P\}\), write the projection
as \(a\I+bP\), and put \(h_n=(h-a\I)/b\). Near \(P\), \(b\) is
real analytic and nonzero, with \(b(0)=1\). This normalization preserves
all one-site stabilizers and is implemented spectrally by a scalar factor
and a spectral rescaling.

Use invariant Hilbert--Schmidt inner products on Hermitian operator
spaces. Let \(\pi\) project onto the orthogonal complement of
\(\operatorname{im}\Delta\), and set \(r(h)=\pi F(h)\).
The normalized derivative \(L=Dr_P\) has the 36-dimensional odd kernel
\(\mathcal T\) of \eqref{fs:eq:chart}. On the 43-dimensional even complement
\(\mathcal N\), orthogonal to \(\I,P\), it is an isomorphism onto
\(\operatorname{im}L\). Solve the \(\operatorname{im}L\) projection of
\(r=0\) for the normal coordinate. The analytic implicit-function theorem
gives
\begin{equation}
\mathfrak h(x)=P+K(x)+n(x),\qquad n(x)\in\mathcal N,
\qquad n(x)=O(\norm{x}^2),
\label{fs:eq:equivariantsection}
\end{equation}
where \(x=(A,B;f,g)\) and \(K(x)\) is \eqref{fs:eq:decomposition}.
Every normalized solution lies in this section. Uniqueness and invariance
of the projections make \(n\) \(SU(3)\)-equivariant for every small
\(x\), including points that fail the remaining integrability equations.
This ambient equivariance will be essential below.

Let \(\mathcal I\) be the analytic ideal generated by the entries of
\(r(\mathfrak h(x))\), and let \(\operatorname{in}\mathcal I\) denote
its ideal of lowest homogeneous terms. Complex tensor coordinates mean
the complexification of this real-analytic construction; on real
parameters we always impose \eqref{fs:eq:realform}.

\begin{lemma}[Genuine initial equations]
\label{fs:lem:genuineinitial}
Every component of
\begin{equation}
A\cdot f,\quad A\cdot g,\quad B\cdot f,\quad B\cdot g,
\quad E_{70}(f,g),\quad [B,[A,B]]
\label{fs:eq:genuineinitial}
\end{equation}
is the initial term of an analytic element of \(\mathcal I\).
The respective degrees are two and three. More precisely, for each
component \(g_i\) there are analytic coefficients \(a_{i\nu}(x)\) such that
\begin{equation}
G_i(x)=\sum_\nu a_{i\nu}(x)r_\nu(\mathfrak h(x))
=g_i(x)+O(\norm x^{d_i+1}),\qquad d_i=\deg g_i.
\label{fs:eq:boundedrepresentatives}
\end{equation}
All coefficients are bounded on one sufficiently small closed ball.
No assertion that these components generate the full analytic ideal
is needed.
\end{lemma}

\begin{proof}
The two functionals in \eqref{fs:eq:twoactionrows} annihilate both
\(\Delta\) and the linearized image. Applied to \(F(\mathfrak h(x))\),
they therefore give analytic equations with the displayed quadratic
initial terms. Their common group orbits and duals supply all four action
blocks. Since these are finite-dimensional spaces, choose finitely many
orbit elements and constant linear combinations spanning every component.
This gives analytic action equations \(G_\alpha\) with quadratic terms
\(g_\alpha\), each a constant linear contraction of the residual.
Denote the polynomial ideal of the \(g_\alpha\) by \(\mathfrak a\).

The identities proving \eqref{fs:eq:fullnormal} hold over the polynomial
quotient by \(\mathfrak a\). Indeed, the pure intrinsic and pure dressing
identities are unconditional; the mixed terms use only
\([A_1+A_2,T]=[B_1+B_2,T]=0\), which are precisely the action relations.
Equivariance gives the same relations for the quadratic contraction
\(N_0\). Thus these statements hold in the quotient ring, without an
assumption that its zero set is reduced. To see explicitly how this
polynomial statement enters the analytic section, put
\(N_c=\operatorname{proj}_{\mathcal N}N_*\), and let \(P_L\) project
onto \(\operatorname{im}L\). The homogeneous quadratic vector
\[
D=\pi Q(K)+LN_c
\]
belongs componentwise to \(\mathfrak a\). All its generators have
degree two, so \(D=\sum_\alpha R_\alpha g_\alpha\) for constant
vectors \(R_\alpha\). The normal equation then gives the actual identity
\begin{equation}
n_2-N_c
=-(L|_{\mathcal N})^{-1}P_LD
=\sum_\alpha C_\alpha g_\alpha,
\qquad C_\alpha=-(L|_{\mathcal N})^{-1}P_LR_\alpha.
\label{fs:eq:normalideallift}
\end{equation}
Removing the scalar and \(P\) parts has no effect on the following
projections.

The analytic equation \(\ell F(\mathfrak h(x))\), with
\(\ell Y=Y_{000,122}\), has order at least three by \eqref{fs:eq:oddQ}.
Modulo linear multiples of \(\mathfrak a\), its cubic term is
\(\ell\mathscr R\), where
\(\mathscr R=F(K)+\mathscr B(K,N_*)\). Abbreviate tensor indices by
sorted triples, while retaining the upper or lower action convention.
Direct contraction gives the following identity for arbitrary complex
\(A,B,f,g\):
\begin{equation}
\begin{split}
\ell\mathscr R+6E_+^{0000}{}_2
={}&(B^0{}_2+2f^{001}+2g_{122})(A\cdot f)^{000}\\
&-2f^{000}\bigl((A\cdot f)^{001}+(A\cdot g)_{122}\bigr)\\
&+(4B^0{}_2+2f^{001}+2g_{122})(B\cdot f)^{000}\\
&-2f^{000}\bigl((B\cdot f)^{001}-(B\cdot g)_{122}\bigr).
\end{split}
\label{fs:eq:initialsyzygy}
\end{equation}
Here is a contraction derivation of \eqref{fs:eq:initialsyzygy}.
Introduce two scalar tags \(u,v\) and write
\[
K_{u,v}=T+u\Tel A+v\Hel B,\qquad
N_{u,v}=N_0+\tfrac12v^2(\Tel B)^2P
-\tfrac12uv([A,B]_1+[A,B]_2)+v[B_2,T].
\]
Apply \eqref{fs:eq:scalarcubic} to \(K_{u,v}\) and
\eqref{fs:eq:scalarbilinear} to \(K_{u,v},N_{u,v}\).
The latter automatically replaces \(N_{u,v}\) by its even part;
in particular the mixed term becomes
\(\tfrac12v[B_2-B_1,T]\). Put \(z=f^{001}+g_{122}\).
Collecting the powers of the two tags gives
\begin{equation}
\begin{array}{c|l}
(u\text{ degree},v\text{ degree})&
[u^rv^s]\,\ell\big(F(K_{u,v})+\mathscr B(K_{u,v},N_{u,v})\big)
\\ \hline
(0,0)&-6E_+^{0000}{}_2\\
(1,0)&2z(A\cdot f)^{000}
-2f^{000}\big((A\cdot f)^{001}+(A\cdot g)_{122}\big)\\
(0,1)&2z(B\cdot f)^{000}
-2f^{000}\big((B\cdot f)^{001}-(B\cdot g)_{122}\big)\\
(1,1)&B^0{}_2(A\cdot f)^{000}\\
(0,2)&4B^0{}_2(B\cdot f)^{000}\\
(2,0)\text{ and }r+s=3&0 .
\end{array}
\label{fs:eq:syzygydegrees}
\end{equation}
For the first row the contractions are
\eqref{fs:eq:thirdtable}. For the two rows linear in \(u\) or \(v\),
the surviving one-site factor acts on each of the three symmetric
tensor slots; a lower slot contributes its negative transpose.
This gives precisely the displayed upper and lower actions.
For the two mixed quadratic rows the contractions reduce, respectively,
to \(3B^0{}_2\sum_j A^0{}_jf^{00j}\) and
\(12B^0{}_2\sum_j B^0{}_jf^{00j}\).
The terms with two \(A\)'s and all terms without \(f,g\) cancel
between the two scalar contraction formulas.
Summing \eqref{fs:eq:syzygydegrees} at \(u=v=1\) proves
\eqref{fs:eq:initialsyzygy}. This derivation is polynomial in all
coordinates; it uses neither a stabilizer assumption nor a statement
about the radical of the action ideal.

Equation~\eqref{fs:eq:normalideallift} makes the additional cubic term
\(\ell\mathscr B(K,n_2-N_c)\) a sum of linear polynomials times
\(g_\alpha\). Combining these polynomials with
\eqref{fs:eq:initialsyzygy}, call their coefficients \(t_\alpha(x)\).
Then the explicit analytic replacement
\[
G_E=-\frac16\left(\ell r(\mathfrak h(x))
                   -\sum_\alpha t_\alpha(x)G_\alpha(x)\right)
\]
has initial term \(E_+^{0000}{}_2\). A finite choice of group orbit
elements and duals gives every component of \(E_{70}\). This is an
identity in the analytic residual ideal, rather than an inference from
its zero set.

Finally, \(\mu_3F(\mathfrak h(x))\) is an analytic linear combination
of the residual equations because \(\mu_3\Delta=0\). Its order is at
least three: \(\mu_2P=3\I\), \(\mu_2P^2=\I\), and
\(\{P,K(x)\}=0\). Equations \eqref{fs:eq:multiplicationcontractions}
give cubic term \(-18[B,[A,B]]\) modulo linear multiples of
\(\mathfrak a\). Subtracting the corresponding multiples of the
analytic action equations supplies the final representatives.
Every operation above is a constant linear contraction, a group action,
or multiplication by a linear coordinate polynomial. It therefore
produces the coefficients in \eqref{fs:eq:boundedrepresentatives}.
There are finitely many of them; continuity on a smaller closed ball
gives a common bound.
\end{proof}

The same construction controls the equations away from their zero set.
Choose invariant norms and an invariant ball. If a representative \(G\)
has an initial covariant with orthogonal output representation \(\rho_0\),
its average over any compact subgroup \(H\subset SU(3)\) obeys
\begin{equation}
\left\lVert\int_H\rho_0(k)^{-1}G(kx)\,dk\right\rVert
\le C\norm{r(\mathfrak h(x))}.
\label{fs:eq:averagedresidualbound}
\end{equation}
Indeed \(r(\mathfrak h(kx))\) is the unitary transform of
\(r(\mathfrak h(x))\), and the coefficient bound is uniform on the
ball. The average is analytic and retains the initial covariant.
This estimate requires no closure of the analytic ideal under integration.

If the coordinate arc \(x(\tau)\) is identically zero, the normalized
density is \(P\) and the charge theorem is immediate. Otherwise
analyticity gives a finite \(m\ge1\) and
\begin{equation}
x(\tau)=\tau^m\bigl(p+O(\tau)\bigr),\qquad
p=(A_0,B_0;f_0,\bar f_0)\ne0.
\label{fs:eq:firstnonzeroleader}
\end{equation}
Lemma~\ref{fs:lem:genuineinitial} imposes \eqref{fs:eq:genuineinitial} on
\(p\). The Hermitian argument of Appendixes~\ref{fs:app:obstructions}
and \ref{fs:app:intrinsic} consequently gives the same classification as
\eqref{fs:eq:classification}, including \([A_0,B_0]=0\).
This conclusion does not require replacing \(\tau\) by a fractional
power.

\subsection{Every nonzero intrinsic direction is a smooth incidence point}
\label{fs:app:intrinsicranks}

\begin{lemma}
\label{fs:lem:intrinsicrank}
At every nonzero Hermitian solution of \(E_{70}=0\), the intrinsic
differential \(dE_{70}\) has real rank 12.
\end{lemma}

\begin{proof}
The differential can be taken in the three polynomial expressions
\eqref{fs:eq:polynomialcovariant}. If \(p=\nabla\phi\), \(H=\nabla^2\phi\),
and dots denote arbitrary variations of \(f,g\), its raw part is
\begin{equation}
\begin{split}
\dot{\mathcal P}_e={}&
\tfrac13\dot\phi\,\tr(G_eH)
+\tfrac13\phi\,\tr(\dot G_eH+G_e\nabla^2\dot\phi)\\
&-\tfrac43(\nabla\dot\phi)^{\mathsf T}G_ep
-\tfrac23p^{\mathsf T}\dot G_ep\\
&+\sum_r\big[
(\partial_r\dot\phi)\tr(\Omega G_r\Omega G_e)
+(\partial_r\phi)\tr(\Omega\dot G_r\Omega G_e
                      +\Omega G_r\Omega\dot G_e)\big].
\end{split}
\label{fs:eq:polynomialdifferential}
\end{equation}
The projected differential is
\(\dot{\mathcal E}_e=\dot{\mathcal P}_e
-x_e\operatorname{div}\dot{\mathcal P}/6\).
This is a linear formula in the 20 tensor variations, with coefficient
extraction normalized as after \eqref{fs:eq:polynomialcovariant}.

At a triangle write \(f^{012}=a\), \(g_{012}=b\), with
\(b=\bar a\ne0\). Differentiate \eqref{fs:eq:covariant} and denote
variations by dots. Its entries include
\begin{equation}
\begin{aligned}
\dot E_+^{iiii}{}_i&=4b(2a+b)\dot f^{iii},\\
\dot E_-{}_{iiii}{}^i&=4a(a+2b)\dot g_{iii},\\
\dot E_+^{iiij}{}_i&=(3a+2b)
\bigl(b\dot f^{iij}+a\dot g_{jkk}\bigr),
\qquad\{i,j,k\}=\{0,1,2\}.
\end{aligned}
\label{fs:eq:trianglephasejacobian}
\end{equation}
For instance, at \(\phi=6ax_0x_1x_2\),
\(g(x)=6bx_0x_1x_2\), the variation \(\dot\phi=x_0^3\) gives
\[
\dot{\mathcal P}_0=(4ab+6b^2)x_0^4,\qquad
\operatorname{div}\dot{\mathcal P}=-12b(2a-b)x_0^3,
\]
so its projected coefficient is \(4b(2a+b)\).
For the variations \(\dot\phi=3x_0^2x_1\) and
\(\dot g(x)=3x_1x_2^2\), the pairs
\(([x_0^3x_1]\dot{\mathcal P}_0,
[x_0^2x_1]\operatorname{div}\dot{\mathcal P})\) are, respectively,
\[
\big(-6b(2a-b),-12b(12a+b)\big),\qquad
\big(12ab,-24a(3a-b)\big).
\]
Subtracting one sixth of the second entry and dividing by four gives
\(b(3a+2b)\) and \(a(3a+2b)\), the third line of
\eqref{fs:eq:trianglephasejacobian}. Cyclic relabeling and duality
supply the other rows.

The first two lines give six independent pure-cube conditions and the
last line gives six disjoint pairs. Every prefactor is nonzero:
\(|a|=|b|\ne0\), while the coefficients in \(2a+b\), \(a+2b\),
and \(3a+2b\) have unequal absolute values. This covers every complex
triangle phase.

At the unit secant \(f^{012}=g_{012}=1\),
\(f^{222}=g_{222}=3\), group variations by the charge
\(\operatorname{diag}(1,-1,0)\). With tensor indices abbreviated by
triples, selected derivative rows have the coefficient blocks
\begin{equation}
\begin{array}{c|c|c}
\text{weight}&\text{variables}&\text{coefficient block}\\ \hline
3&(\dot f^{000},\dot g_{111})&
\begin{pmatrix}9&6\\6&9\end{pmatrix}\\[4pt]
2&(\dot f^{002},\dot g_{112})&\begin{pmatrix}1&1\end{pmatrix}\\[3pt]
1&(\dot f^{001},\dot f^{022},\dot g_{011},\dot g_{122})&
\begin{pmatrix}0&-4&5&-6\\5&-6&0&-4\end{pmatrix}\\[4pt]
0&(\dot f^{012},\dot f^{222},\dot g_{012},\dot g_{222})&
\begin{pmatrix}-1&0&1&0\\-3&1&-3&1\end{pmatrix}.
\end{array}
\label{fs:eq:continuationsecantblocks}
\end{equation}
In order, the selected rows are
\[
\begin{gathered}
\dot E_+^{0000}{}_0,\quad \dot E_-{}_{1111}{}^1;\qquad
-\dot E_+^{0001}{}_2/9;\\
\dot E_+^{2222}{}_1/6,\quad \dot E_-{}_{2222}{}^0/6;\\
\dot E_+^{0011}{}_2/12,\quad
(\dot E_+^{0222}{}_0+\dot E_-{}_{0222}{}^0)/12.
\end{gathered}
\]
To make the secant entries equally explicit, its polynomial and
Hessian are
\[
\phi=g(x)=6x_0x_1x_2+3x_2^3,\qquad
\nabla^2\phi=6
\begin{pmatrix}0&x_2&x_1\\x_2&0&x_0\\x_1&x_0&3x_2\end{pmatrix}.
\]
For \(\dot\phi=x_0^3\) and \(\dot g(x)=x_1^3\), respectively,
\eqref{fs:eq:polynomialdifferential} gives
\[
\begin{array}{c|c|c}
&\dot{\mathcal P}&\operatorname{div}\dot{\mathcal P}\\ \hline
\dot f^{000}=1&
(10x_0^4,-26x_0^3x_1-36x_0^2x_2^2,-8x_0^3x_2)&6x_0^3\\
\dot g_{111}=1&
(0,-18x_0^3x_1-27x_0^2x_2^2,-18x_0^3x_2)&-36x_0^3 .
\end{array}
\]
The projected \(x_0^4\) coefficients are \(10-6/6=9\) and
\(0+36/6=6\). Applying the same derivative and coefficient rule to
\(3x_0^2x_2\), \(3x_0^2x_1\), \(3x_0x_2^2\),
\(6x_0x_1x_2\), and \(x_2^3\), and their duals, gives the remaining
selected rows in \eqref{fs:eq:continuationsecantblocks}.

The conjugate blocks give the negative weights. The ranks are
\(2(2+1+2)+2=12\). Cubic homogeneity multiplies the entire differential
by \(c^2\) at a nonzero real secant amplitude \(c\).

In both cases the orbit and amplitude manifold has real dimension eight:
six orbit directions and two amplitudes for triangles, seven orbit
directions and one amplitude for secants. Its tangent gives the matching
rank upper bound in the 20-dimensional intrinsic space. Equivariance
covers all normal-form representatives. The differential respects
Hermitian conjugation, so its complex rank equals the rank on the real
form. This proves the assertion.
\end{proof}

Let \(s\) be the intrinsic stabilizer dimension, equal to two or one.
The allowed pairs \((A,B)\) form two copies of this abelian algebra.
The resulting Hermitian incidence manifold \(M\) has dimension
\begin{equation}
d=8+2s=12\ \text{or}\ 10.
\label{fs:eq:incidencedimension}
\end{equation}
At every \(p\in M\), the action equations have rank \(2(8-s)\) in
the \((A,B)\) columns. The intrinsic equations have rank 12 and no
entries in those columns. A block-triangular minor therefore gives rank
at least \(36-d\) for \eqref{fs:eq:genuineinitial}. Since \(M\) lies
in their common zero set, its tangent gives the opposite inequality.
The rank is exactly \(36-d\), for arbitrary allowed dressing data.

\subsection{Exact families exhaust arcs with a nonzero intrinsic leader}
\label{fs:app:familyfilling}

\begin{proposition}
\label{fs:prop:familyfilling}
If \(f_0\ne0\) in \eqref{fs:eq:firstnonzeroleader}, the given normalized
density arc is locally a member of the corresponding exact family in
Appendix~\ref{fs:app:lifts}, with parameters and a common unitary basis
depending analytically on \(\tau\). Its space of traceless Hermitian
one-site charges has dimension exactly two in the triangle case and one
in the secant case for small nonzero \(\tau\). Every charge of the
leading direction extends analytically, with constant spectrum.
\end{proposition}

\begin{proof}
Select \(36-d\) homogeneous equations \(g_i\) from
\eqref{fs:eq:genuineinitial} with independent differentials at \(p\).
Lemma~\ref{fs:lem:genuineinitial} supplies analytic representatives
\(G_i\in\mathcal I\), with initial degree \(d_i=2\) or \(3\).
Introduce an independent scale \(\sigma\) and define
\begin{equation}
\widetilde G_i(\sigma,z)=\sigma^{-d_i}G_i(\sigma z).
\label{fs:eq:analyticrescaling}
\end{equation}
This is analytic at \(\sigma=0\), with value \(g_i(z)\). Its common
zero set \(Y\) is smooth of dimension \(d+1\) near \((0,p)\).
Every rescaled actual density arc lies in \(Y\).

Here is a local parameterization of the exact families used to fill
\(Y\). Let \(\mathfrak t\) be the real traceless diagonal matrices
and \(q_*=\operatorname{diag}(1,-1,0)\). Before a common basis change,
the parameters and their real dimensions are
\begin{equation}
\begin{array}{c|c|c|c|c}
\text{type}&f&A&B&\text{dimensions: orbit, amplitude, }A,B\\ \hline
\text{triangle}&6(a+ib)x_0x_1x_2&\mathfrak t&i\mathfrak t&6+2+2+2=12\\
\text{secant}&3c\,x_2(2x_0x_1+x_2^2)&\mathbb Rq_*&i\mathbb Rq_*&7+1+1+1=10
\end{array}
\label{fs:eq:familyparameters}
\end{equation}
Here \(a+ib\ne0\) and \(c\in\mathbb R\setminus\{0\}\).
A local analytic section of \(SU(3)/H\), with \(H\) the intrinsic
stabilizer, supplies the orbit coordinates. Its connected stabilizer
acts trivially on the allowed abelian \(A,B\); any discrete
identifications are removed by choosing a local parameter chart.
Thus these are \(d\) independent coordinates on \(M\).

For the intrinsic density \(k_\sigma\), use
\(\kappa=6a\sigma\), \(p=e^{2ib\sigma}\) in the triangle formulas,
or \(s=1-3c\sigma/2\) and \(r=r(s)\) in the secant formulas.
The hyperbolic quotient is removable, and the secant denominators are
nonzero near \(s=1\). These formulas and their spectral matrices are
jointly analytic in all the displayed parameters, including \(a=0\)
or \(b=0\). Apply \eqref{fs:eq:combined-density} with
\(D_\sigma=e^{-\sigma B}\), then the common unitary from the orbit
section. This supplies every allowed dressing parameter and preserves
the fixed standard charges. Its derivative at \(\sigma=0\) is
exactly \(K(p')\). Scalar and exchange normalization leave that
derivative unchanged, since it is exchange-odd. By uniqueness of
\eqref{fs:eq:equivariantsection}, the exact density has coordinates
\[
x_*(\sigma,p')=\sigma p'+O(\sigma^2).
\]
Consequently
\begin{equation}
\Xi(\sigma,p')=(\sigma,x_*(\sigma,p')/\sigma)
\label{fs:eq:familyfilling}
\end{equation}
extends analytically to \(\sigma=0\) and takes values in \(Y\).
Its differential in \(p'\) is the inclusion of \(T_pM\), and its
\(\sigma\) differential has first component one. Its rank is
\(d+1\), so the analytic inverse-function theorem makes its image a
neighborhood of \((0,p)\) in \(Y\).

Apply this inverse to the analytic path
\[
(\sigma,z)=\bigl(\tau^m,x(\tau)/\tau^m\bigr).
\]
It gives analytic parameters \(p'(\tau)\) and leaves the scale equal
to \(\tau^m\). Thus the given density equals a member of the exact
family. This argument uses genuine equations and an exact family filling
their smooth solution space; no higher-degree generation statement or
fractional reparametrization is involved.

Transport the fixed standard charges of that family by its analytic
common unitary. These charges have constant spectrum. For the exact
count, divide the normalized commutator map
\(q\mapsto[h_n(\tau),q\otimes\I+\I\otimes q]\) by \(\tau^m\).
At zero its kernel is the \(s\)-dimensional stabilizer of \(K(p)\).
A nonzero minor bounds the nearby kernel dimension above by \(s\),
while the transported charges give \(s\) independent elements.
They form an analytic basis and extend every leading charge.
\end{proof}

\subsection{A nonzero helix leader eliminates all charged variables}
\label{fs:app:helixcontinuation}

Consider a remaining leader with \(f_0=0\), \(B_0\ne0\).
Let \(T_0\) be the compact closure of \(\exp(\theta B_0)\) in
\(SU(3)\). The point \(p\) is fixed by \(T_0\), since
\([A_0,B_0]=0\). We show that an analytic common basis change makes
the entire density fixed by this group.

First make the analytic representatives of the \(B\)-action and
double-commutator equations equivariant under \(T_0\). If \(G\) has
initial covariant \(g\), with output representation \(\rho\), replace it
on an invariant neighborhood by
\begin{equation}
\widehat G(x)=\int_{T_0}\rho(k)^{-1}G(kx)\,dk.
\label{fs:eq:equivariantaverage}
\end{equation}
The integral is real analytic, preserves the initial covariant, and
vanishes on every actual solution because that solution set is invariant.
Only this vanishing is needed; no ideal-closure assertion for the integral
is assumed. Rescale these representatives as in
\eqref{fs:eq:analyticrescaling}.

The anti-Hermitian matrix \(B\) in \(z(\tau)=x(\tau)/\tau^m\)
equals \(B_0\) at zero. Separate the distinct eigenspaces of \(B_0\).
Their spectral clusters remain separated nearby. Analytic spectral
projections and analytic orthonormal frames give a common unitary equal
to \(\I\) at zero, in which
\[
[B,B_0]=0.
\]
Explicitly, if \(P_\alpha(\tau)\) is a cluster projection and
\(E_\alpha\) is a fixed orthonormal frame for its initial range, use
\(W_\alpha=P_\alpha(\tau)E_\alpha\) and the frame
\(W_\alpha(W_\alpha^\dagger W_\alpha)^{-1/2}\). The inverse square
root is analytic near \(\I\). Concatenating these frames and comparing
with the initial basis gives the required unitary.
This block gauge does not label eigenvalues within a repeated block.
It therefore requires no simple-spectrum hypothesis.

In this slice, split coordinates into the \(T_0\)-fixed variables \(u\)
and charged variables \(v\). The latter are the nonzero-weight parts
of \(A\) and \((f,g)\); all \(B\) variables are fixed. Project the
rescaled \(B\)-action and double-commutator equations onto their
charged outputs. At \((\sigma,z)=(0,p)\), their derivative in \(v\) is
\begin{equation}
\delta(f,g)_{\rm ch}\longmapsto B_0\cdot\delta(f,g)_{\rm ch},
\qquad
\delta A_{\rm ch}\longmapsto-\ad_{B_0}^2\delta A_{\rm ch}.
\label{fs:eq:chargedhelixderivative}
\end{equation}
Both maps are invertible on their charged spaces. Variations of \(B\)
within its blocks produce no charged derivative: \(f_0=g_0=0\) and
\([B_0,[A_0,\delta B]]=0\).

At every fixed input \((\sigma,u,0)\), the charged outputs vanish
exactly by equivariance. The analytic implicit-function theorem makes
\(v=0\) the unique nearby solution for \(v\), given \(\sigma,u\).
Hence every actual arc in the block gauge has fixed tangent coordinates.
Equation \eqref{fs:eq:equivariantsection} then makes its complete density
fixed by \(T_0\). Undoing the basis change extends the Hermitian charge
\(-iB_0\), with constant spectrum. No classification of the remaining
neutral equations is required.

\subsection{A pure-telescope leader and a positive analytic scalar}
\label{fs:app:telescopecontinuation}

The last case is \(p=(A_0,0;0,0)\), with \(A_0\ne0\).
The linearized double commutator does not eliminate \(B\) here.
Instead we use the exact identity \eqref{fs:eq:multiplication} and
positivity on the Hermitian real form.

Let \(T_0\) be the compact closure of \(\exp(i\theta A_0)\).
As above, an analytic common basis change puts the \(A\) coordinate
of \(z(\tau)\) in the block slice \([A,A_0]=0\).
Write its coordinates as the fixed part \(u\), the charged
anti-Hermitian part \(v=B_{\rm ch}\), and the charged intrinsic part
\(w=(f,g)_{\rm ch}\). At the base point, \(u=u_0\) and \(v=w=0\).

Average the analytic \(A\)-action equations as in
\eqref{fs:eq:equivariantaverage}, rescale by \(\sigma^2\), and project
to the charged intrinsic output. Their derivative in \(w\) is the
invertible action of \(A_0\) on its nonzero weights. They therefore
have a unique equivariant analytic solution
\begin{equation}
w=\varphi(\sigma,u,v),\qquad
\varphi(\sigma,u,0)=0,\qquad \varphi(0,u_0,v)=0.
\label{fs:eq:intrinsicelimination}
\end{equation}
The second equality follows from equivariance. The third holds because
at \(\sigma=0,u=u_0\) the equation is simply \(A_0\cdot w=0\),
independent of \(B\). Every actual arc satisfies this elimination.

On the whole implicit section define the real scalar
\begin{equation}
\mathcal S(\sigma,z)
=-\frac{1}{18\sigma^3}\operatorname{Re}\tr\!\left(
A_0[\mu_2(\mathfrak h(\sigma z)^2),
       \mu_2(\mathfrak h(\sigma z))]\right).
\label{fs:eq:coercivescalar}
\end{equation}
It extends analytically through \(\sigma=0\): the two base contractions
are scalar and \(\{P,K(z)\}=0\), so the commutator has no term of
degree below three. On an actual Reshetikhin density it is exactly zero
by \eqref{fs:eq:multiplication}.

The scalar is \(T_0\)-invariant. More strongly, it vanishes at every
\(T_0\)-fixed input \(z\), whether or not that input is integrable.
Indeed, ambient equivariance of \eqref{fs:eq:equivariantsection} makes
both contractions commute with \(A_0\). For such contractions \(M,N\),
\(\tr(A_0[M,N])=\tr([A_0,M]N)=0\). This exact vanishing controls
the neutral remainder in the positivity argument.

For pure dressing data, the quadratic normal term is \(N_p\) in
\eqref{fs:eq:purenormal}, up to scalar and \(P\) normalization. Thus
\eqref{fs:eq:multiplicationcontractions} gives
\[
\mu_2K_p=-3B,\qquad
\mu_2(K_p^2+\{P,N_p\})=-6[A,B].
\]
The normalization changes the second expression only by a scalar.
At \(\sigma=0,A=A_0,B=v,f=g=0\), the leading scalar is therefore
\begin{equation}
\mathcal S(0,A_0,v;0,0)=\tr([A_0,v]^2)
=2\sum_{i<j}(\lambda_i-\lambda_j)^2|v_{ij}|^2,
\label{fs:eq:positivechargedform}
\end{equation}
where \(\lambda_i\) are the eigenvalues of \(A_0\).
Here \([A_0,v]\) is Hermitian. Since \(v\) contains only entries
between distinct eigenspaces, this form is strictly positive unless
\(v=0\), including when \(A_0\) has repeated eigenvalues.

\begin{lemma}[Uniform positivity in charged variables]
\label{fs:lem:equivariantpositivity}
Let a compact group \(H\) act orthogonally on a finite-dimensional real
space \(V\), with \(V^H=0\), and trivially on a finite-dimensional real parameter space
\(\Theta\). Let \(s(\theta,v)\) be a real \(C^2\), \(H\)-invariant
function near \((\theta_0,0)\). Suppose \(s(\theta,0)=0\) for every
nearby \(\theta\), and \(D_v^2s(\theta_0,0)\) is positive definite.
Then, after shrinking the neighborhood, some \(c>0\) satisfies
\(s(\theta,v)\ge c\norm v^2\) uniformly in \(\theta\).
\end{lemma}

\begin{proof}
The linear functional \(D_vs(\theta,0)\) is invariant. For every \(v\),
its value equals its value on \(\int_Hkv\,dk\), which is zero because
\(V^H=0\). Thus the constant and linear charged terms vanish exactly
at every neutral parameter. Continuity of the Hessian and compactness
of the unit sphere in \(V\) give a product neighborhood, convex in \(v\),
on which \(D_v^2s(\theta,v)[\xi,\xi]\ge2c\norm\xi^2\).
Taylor's integral formula gives
\[
s(\theta,v)=\int_0^1(1-t)D_v^2s(\theta,tv)[v,v]\,dt
\ge c\norm v^2.
\]
Equivalently, for \(f(t)=s(\theta,tv)\), two mean-value arguments
integrate \(f''(t)\ge2c\norm v^2\), using \(f(0)=f'(0)=0\).
This is a bound on the complete function, with no power-series truncation.
\end{proof}

Apply the lemma with \(\theta=(\sigma,u)\), \(V\) the charged
anti-Hermitian matrices, and
\[
s(\sigma,u,v)=\mathcal S(\sigma,u,v,\varphi(\sigma,u,v)).
\]
Equivariance of \(\varphi\) and the ambient section gives invariance
and exact neutral vanishing for all nearby \((\sigma,u)\).
The identity \(\varphi(0,u_0,v)=0\) identifies the Hessian at the
leader with twice the quadratic form in
\eqref{fs:eq:positivechargedform}. Explicitly, choose \(\gamma>0\)
no larger than any nonzero eigenvalue gap of \(A_0\). For charged \(v\),
\begin{equation}
\tr([A_0,v]^2)=\norm{[A_0,v]}_\HS^2
=\sum_{i,j}(\lambda_i-\lambda_j)^2|v_{ij}|^2
\ge\gamma^2\norm v_\HS^2.
\label{fs:eq:chargedgapbound}
\end{equation}
Equal eigenvalues belong to a single neutral block; all within-block
entries of \(v\) vanish. Thus repeated eigenvalues cause no loss of
positivity on the charged space. The lemma supplies
\begin{equation}
s(\sigma,u,v)\ge c\norm v^2
\label{fs:eq:fullanalyticpositivity}
\end{equation}
for the full analytic scalar, uniformly over the nearby neutral inputs.

On an actual arc, \(s=0\), so \(v=0\), and then \(w=0\) by
\eqref{fs:eq:intrinsicelimination}. All coordinates are fixed by \(T_0\),
and ambient equivariance makes the density fixed as well. Undoing the
analytic basis change extends the charge \(A_0\), with constant
spectrum.

\begin{proof}[Proof of Theorem~\ref{fs:thm:main}]
Every nonconstant normalized analytic arc has a first nonzero coefficient
\eqref{fs:eq:firstnonzeroleader}. If \(f_0\ne0\), use
Proposition~\ref{fs:prop:familyfilling}. If \(f_0=0\) and \(B_0\ne0\),
use Subsection~\ref{fs:app:helixcontinuation}; otherwise use
Subsection~\ref{fs:app:telescopecontinuation}. Each case transports a fixed
nonzero traceless Hermitian charge by an analytic common unitary. A
constant initial basis adjustment makes that unitary equal to \(\I\)
at zero. The normalized constant arc is immediate, and restoring scalar
and scale terms preserves every charge. This proves the theorem.
\end{proof}
\section{Defects and normalization}\label{main:app:defects}
Let $P$ exchange two copies of $\mathbb C^3$. All local Hilbert--Schmidt
norms use the unnormalized trace. For a Hermitian two-site density set
\begin{align}
\delta(h)&=\min_{\substack{q=q^\dagger,\ \tr q=0\\\norm q_\HS=1}}
  \norm{[h,q\otimes\I+\I\otimes q]}_\HS,\label{st:eq:delta}\\
F(h)&=[h_{12}+h_{23},[h_{12},h_{23}]],\qquad
r(h)=\pi_\Delta F(h),\quad \varepsilon(h)=\norm{r(h)}_\HS,
\label{st:eq:residual}
\end{align}
where $\pi_\Delta$ projects orthogonally away from
$\operatorname{im}\Delta$, with $\Delta X=X_{23}-X_{12}$.
Write the orthogonal projection of $h$ onto
$\operatorname{span}\{\I,P\}$ as $a\I+bP$, and put
$h_n=(h-a\I)/b$. Then
\begin{equation}
\delta(h)=|b|\delta(h_n),\qquad
\varepsilon(h)=|b|^3\varepsilon(h_n).
\label{st:eq:scale}
\end{equation}

Theorem~\ref{st:thm:sharp} is stated in the main text.

\begin{corollary}[Fixed-density forced charge]
Every Hermitian density in this neighborhood with $r(h)=0$ has a nonzero
traceless Hermitian one-site charge. In particular this holds for each
individual density admitting a regular analytic additive difference-form
$R$ matrix, without requiring a jointly analytic model-deformation family
through $P$.
\end{corollary}
This follows by setting $\varepsilon=0$. It removes a hypothesis about
the additional model parameter, not spectral regularity. It gives neither
a common charge nor a globally analytic choice of charge over the
neighborhood.

\section{Proof retaining the full analytic error}\label{main:app:stability}
We use the analytic inputs proved above. Appendix~\ref{fs:app:continuation} constructs
an $SU(3)$-equivariant section
\begin{equation}
\mathfrak h(x)=P+K(x)+n(x),\qquad n(x)=O(\norm x^2),
\quad x=(A,B;f,\bar f)\in\mathbb R^{36},\label{st:eq:section}
\end{equation}
by solving the projection of $r=0$ onto $\operatorname{im}Dr_P$.
The derivative in the 43 normal coordinates is invertible. The same
appendix constructs genuine analytic residual equations whose initial
terms are the action equations, the intrinsic cubic covariant, and
$[B,[A,B]]$. Their degrees are two or three. On the Hermitian real form
their nonzero common zeros have either a nonzero intrinsic component,
a nonzero helix component $B$, or only a telescope component $A$.

\subsection{Normal reduction and compactness}
Write a normalized density as $h_n=P+K(x)+y$. Invertibility of the
normal derivative and a mean-derivative estimate imply, on a smaller
neighborhood,
\begin{equation}
\norm{h_n-\mathfrak h(x)}\le C_0\varepsilon(h_n),\qquad
\varepsilon(\mathfrak h(x))\le C_1\varepsilon(h_n).
\label{st:eq:normalerror}
\end{equation}
Here the averaged derivative is kept within half the smallest singular
value of the fixed invertible derivative at $P$; pointwise
invertibility alone would not suffice. This estimate applies off the
zero set. Also
$|\delta(h)-\delta(g)|\le4\norm{h-g}_\HS$.
It therefore suffices to prove the bound in Theorem~\ref{st:thm:sharp} on the section.

Choose an invariant tangent norm and put $x=\sigma z$,
$\sigma=\norm x>0$, $\norm z=1$, and
$\rho=\varepsilon(\mathfrak h(\sigma z))$. Since $P$ commutes with
every total one-site operator,
\begin{equation}
\delta(\mathfrak h(\sigma z))\le C_2\sigma.\label{st:eq:trivial}
\end{equation}
Equation~\eqref{fs:eq:boundedrepresentatives} gives the actual analytic
coefficients expressing each \(G_i\) in the residual components, and a
common bound for them on a small invariant ball. If its initial degree is $d_i\in\{2,3\}$, then
\begin{equation}
|\widetilde G_i(\sigma,z)|
=|\sigma^{-d_i}G_i(\sigma z)|\le C_i\rho/\sigma^{d_i}.
\label{st:eq:scalederror}
\end{equation}

If no uniform local bound existed, one could choose section points
approaching $P$ with $\delta_n^3>n\rho_n$, including possible zero
values of $\rho_n$. Equation
\eqref{st:eq:trivial} would force $\rho/\sigma^3\to0$.
Passing to a subsequence, $z\to p$ on the unit sphere. Equation
\eqref{st:eq:scalederror} forces all the genuine initial equations to vanish
at $p$. We now prove the cubic bound in a neighborhood of every possible
such $p$. This contradicts that sequence and establishes uniformity.

\subsection{Nonzero intrinsic leaders}
The family-filling proposition in Appendix~\ref{fs:app:continuation} selects independent scaled
equations whose zero set $Y$ is smooth near $(0,p)$ and is filled by
exact densities with nonzero one-site charges. Quantitative implicit
inversion, at fixed $\sigma$, gives
\[
\operatorname{dist}(z,Y_\sigma)\le C_p\rho/\sigma^3.
\]
The section has derivative $O(\sigma)$ with respect to $z$. Comparing
with a density on $Y_\sigma$ thus gives
\begin{equation}
\delta\le C_p\rho/\sigma^2.\label{st:eq:linearerror}
\end{equation}
Multiplication by the square of \eqref{st:eq:trivial} yields
$\delta^3\le C_p\rho$. If $\rho/\sigma^3$ is not small, the trivial
estimate alone suffices. The argument uses exact filling of $Y$, not
merely equality of tangent spaces. More explicitly, use
$(\sigma,u,v)\mapsto(\sigma,u,\widetilde G)$ as local coordinates and
set the last coordinates to zero while preserving $\sigma,u$.

\subsection{Nonzero helix leaders}
Use analytic spectral cluster frames to impose $[B,B_0]=0$, allowing
repeated eigenvalues within each cluster. Split the remaining variables
into fixed coordinates $u$ and charged coordinates $v$ under the compact
closure of $\exp(\theta B_0)$. Average the analytic representatives over
this group. Equation~\eqref{fs:eq:averagedresidualbound} shows that the averaged
equations obey \eqref{st:eq:scalederror} with a bounded constant,
without requiring ideal membership of the group integral.

The charged output vanishes exactly at $v=0$ and has invertible
derivative in $v$ at the leader. Hence
$\norm v\le C_p\rho/\sigma^3$. Ambient equivariance of the section
makes every $v=0$ density commute with $-iB_0$, even off the remaining
zero set. It follows that $\delta\le C_p\sigma\norm v$, which gives
\eqref{st:eq:linearerror} and the cubic bound.

\subsection{Pure telescope leaders}
Now $p=(A_0,0;0,0)$ with $A_0\ne0$. In a spectral block frame with
$[A,A_0]=0$, let $v=B_{\rm ch}$ and $w=(f,\bar f)_{\rm ch}$.
The scaled quadratic action equations have an invertible derivative in
$w$. Their exact implicit solution is $w=\varphi(\sigma,u,v)$, with
$\varphi(\sigma,u,0)=0$. For arbitrary inputs, mean-derivative inversion
therefore gives
\begin{equation}
w=\varphi(\sigma,u,v)+e,\qquad
\norm e\le C_p\rho/\sigma^2,\qquad
\norm\varphi\le C_p\norm v.\label{st:eq:werror}
\end{equation}

The complete analytic scalar constructed in Appendix~\ref{fs:app:continuation} is a bounded
linear contraction of the residual divided by $\sigma^3$. Thus its
absolute value at the actual input is at most $C_p\rho/\sigma^3$.
Its extension at $\sigma=0$ is analytic in the ambient section:
the multiplication contractions satisfy
$\mu_2\mathfrak h=3\I+O(\sigma)$ and
$\mu_2\mathfrak h^2=\I+O(\sigma^2)$ because $\{P,K\}=0$.
Their commutator thus starts at order three before imposing the
remaining equations. In particular its derivative with respect to
$w$ stays bounded. Replacing $w$ by $\varphi$ changes the scalar by
$O(\norm e)$.
After this elimination, Lemma~\ref{fs:lem:equivariantpositivity},
applied to the full analytic scalar, gives
\[
s(\sigma,u,0)=0,\qquad D_vs(\sigma,u,0)=0,\qquad
s(\sigma,u,v)\ge c_p\norm v^2.
\]
The first identity holds at all neutral inputs, not only solutions;
the second follows from compact invariance. The quadratic term at the leader
is the positive form $\tr([A_0,v]^2)$. Consequently
\begin{equation}
\norm v^2\le C_p\rho/\sigma^3.\label{st:eq:verror}
\end{equation}
Ambient equivariance, \eqref{st:eq:werror}, and \eqref{st:eq:verror} imply
\begin{equation}
\delta\le C_p\sigma(\norm v+\norm e),\qquad
\delta^2\le C_p\left(\rho/\sigma+\rho^2/\sigma^2\right).
\end{equation}
For $\rho\le\sigma^3$ and $\sigma\le1$ this gives
$\delta^2\le C_p\rho/\sigma$. Multiplying \eqref{st:eq:trivial}
proves $\delta^3\le C_p\rho$. The complementary region is handled
by \eqref{st:eq:trivial} alone.

The final exponent follows from explicit scalar inequalities.
Write \(d=\delta\), \(s=\sigma\), \(r=\rho\), and choose
nonnegative local constants \(a,b,c\) so that \(d\le as\).
In a linear chart, \(ds^2\le br\) implies
\[
d^3\le a^2(ds^2)\le a^2br.
\]
In a telescope chart, the displayed estimate has the form
\(d^2s^2\le c(rs+r^2)\). If \(0\le r\le s^3\) and
\(0<s\le1\), then \(r\le s\) and \(r^2\le rs\), giving
\[
d^2s\le2cr,\qquad d^3\le asd^2\le2acr.
\]
If \(s^3\le r\), the trivial bound gives \(d^3\le a^3r\).
Thus \(\max\{a^3,a^2b,2ac\}\) suffices on any such local chart.
The compactness argument supplies finitely many charts or,
equivalently, excludes a sequence on which their bounds fail.

Normal errors can also be restored with an explicit constant.
Suppose \(d\le d_0+L\varepsilon\),
\(d_0^3\le C\varepsilon\), and \(0\le\varepsilon\le1\),
with all quantities nonnegative. Then
\begin{equation}
d^3\le4\bigl(d_0^3+L^3\varepsilon^3\bigr)
\le4(C+L^3)\varepsilon.
\label{st:eq:normalrestorationconstant}
\end{equation}
The elementary inequality used here follows from
\(4(x^3+y^3)-(x+y)^3=3(x+y)(x-y)^2\ge0\).
In \eqref{st:eq:normalerror}, one may take \(L=4C_0\) and absorb
\(C_1\) into \(C\). Scalar shifts and density rescaling then cancel
exactly as in \eqref{st:eq:scale}.

These three charts complete the compactness argument. Equation
\eqref{st:eq:normalerror} restores arbitrary normal errors, whose cubic
contribution is $O(\varepsilon^3)$ and is absorbed locally. Equation
\eqref{st:eq:scale} restores scalar and scale directions. The explicit
perturbation below proves optimality and completes
Theorem~\ref{st:thm:sharp}.

\section{A length-independent physical residual norm}\label{main:app:norm}
Fix the density representative. For a periodic chain with $L\ge5$, set
\begin{equation}
H_L=\sum_jh_{j,j+1},\qquad
J_L=i\sum_j[h_{j,j+1},h_{j+1,j+2}],\quad
C_J(t)=\frac{\Tr(J_L(t)J_L)}{L3^L}.
\end{equation}
Jacobi cancellation of four-site terms gives
$[H_L,J_L]=i\Sigma_LF(h)$. In particular
\begin{equation}
\kappa_J:=-C_J''(0)=\frac{\norm{\Sigma_LF(h)}_\HS^2}{L3^L}.
\end{equation}
Thus $i[H_L,J_L]=-\Sigma_Lr(h)$. This local energy-current identity
is standard; see, for example, Appendix C of \cite{surace}.
We distinguish the force variance $\kappa_J$ from the relative
curvature $\Gamma_J=-[C_J(t)/\chi_J]''_{t=0}=\kappa_J/\chi_J$,
where $\chi_J=C_J(0)$.

\begin{lemma}[Periodic quotient norm]\label{st:lem:norm}
For a traceless three-site operator $F$, write
$r=\pi_\Delta F=r_1+r_2+r_3$ according to the minimum support span of
its identity/traceless tensor words. Then for $L\ge5$,
\begin{equation}
\frac{\norm{\Sigma_LF}_\HS^2}{L3^L}
=\frac{3\norm{r_1}_\HS^2+2\norm{r_2}_\HS^2+\norm{r_3}_\HS^2}{27}.
\label{st:eq:norm}
\end{equation}
Consequently $\varepsilon^2/27\le\kappa_J\le\varepsilon^2/9$.
\end{lemma}
\begin{proof}
Use an orthonormal single-site basis $\I,E_1,\ldots,E_8$ for the
normalized trace inner product. A word of support span $s$ has
$k=4-s$ padding positions in three sites. The image of $\Delta$
spans differences of these positions, so the orthogonal quotient
representative assigns their common average coefficient $a$.
Its local squared norm is $k|a|^2$; the periodic squared norm per site
is $k^2|a|^2$. Different translation orbits are orthogonal for $L\ge5$.
Span three includes words with an identity at the middle site.
The traceless assumption excludes the constant word. Summing gives
\eqref{st:eq:norm}.
\end{proof}

This quotient removes boundary terms from $F$, not from the choice of
$h$. Adding a telescope to $h$ can change the instantaneous current
representative and its curvature. The examples below have scalar partial
traces on both sites and use explicitly fixed densities.

\begin{corollary}[Sharp sixth-power initial response bound]\label{st:cor:sixth}
In the neighborhood of Theorem~\ref{st:thm:sharp},
\begin{equation}
\kappa_J(h)\ge\frac{\delta(h)^6}{27C^2}.
\label{st:eq:sixthforce}
\end{equation}
On a smaller normalized neighborhood choose
$0<\chi_-\le\chi_J(h_n)\le\chi_+$. Then
\begin{equation}
\frac{\Gamma_J(h)}{|b|^2}\ge
\frac1{27C^2\chi_+}\left(\frac{\delta(h)}{|b|}\right)^6.
\label{st:eq:sixthresponse}
\end{equation}
At fixed exchange scale this is $\Gamma_J\ge c\delta^6$.
The power six cannot be replaced by a smaller positive power.
No scalar-partial-trace assumption is required for these bounds.
\end{corollary}
\begin{proof}
Combine $\varepsilon\ge\delta^3/C$ with
$\kappa_J\ge\varepsilon^2/27$.
The susceptibility is continuous and equals $16/9$ at $P$, giving
the bounds on $\chi_J(h_n)$. Under $h=a\I+bh_n$,
$\chi_J(h)=|b|^4\chi_J(h_n)$ and
$\Gamma_J(h)=|b|^2\Gamma_J(h_n)$, which prove
\eqref{st:eq:sixthresponse}. Sharpness follows from the Fermat family below.
\end{proof}
This is a physical restatement of the sharp structural estimate, not
an independent relaxation mechanism. In particular,
$\delta^3\le C\varepsilon$ does not convert a persistence bound
proportional to $\varepsilon^{-1}$ into a general lower bound
proportional to $\delta^{-3}$: the inverse inequality runs in the
opposite direction. Such a conversion requires a separate upper bound
on $\varepsilon$, supplied by the explicit Fermat construction below.

\begin{corollary}[Sharp response constraint]
Suppose $h$ is in the neighborhood of Theorem~\ref{st:thm:sharp} and
$\Tr_1h=\Tr_2h=c\I$. For $Q_L=\sum_jq_j$ with
$\tr q=0$, $\norm q_\HS=1$, define
$C_q(t)=\Tr(Q_L(t)Q_L)/(L3^L)$ and
$\Gamma_{\rm one}=\min_q[-C_q''(0)/C_q(0)]$. Then
\begin{equation}
\boxed{\Gamma_{\rm one}=\frac{\delta(h)^2}{3}
\le C^{2/3}\kappa_J(h)^{1/3}.}\label{st:eq:response}
\end{equation}
The exponent $1/3$ is optimal.
\end{corollary}
\begin{proof}
The two partial traces of $[h,q_1+q_2]$ vanish. Distinct translates
are therefore orthogonal in the infinite-temperature inner product.
Since $C_q(0)=1/3$, its relative curvature is
$\norm{[h,q_1+q_2]}_\HS^2/3$. Apply Theorem~\ref{st:thm:sharp}
and Lemma~\ref{st:lem:norm}. Optimality follows below.
\end{proof}

\section{An explicit saturating perturbation}\label{main:app:fermat}
Let the sums run over cyclic triples $(a,b,c)=(0,1,2),(1,2,0),(2,0,1)$,
and define
\begin{align}
T&=\sum_{\rm cyc}\bigl[|aa\rangle(\langle bc|-\langle cb|)
 +( |bc\rangle-|cb\rangle)\langle aa|\bigr],\\
C_+&=\sum_{\rm cyc}\bigl[|aa\rangle(\langle bc|+\langle cb|)
 +( |bc\rangle+|cb\rangle)\langle aa|\bigr],\qquad
D=\sum_a|aa\rangle\langle aa|,\\
N&=\I+2P-5D-C_+,\qquad h_\lambda=P+\lambda T+\lambda^2N.
\end{align}
Its conversion tangent already occurs in the Leigh--Strassler
Hamiltonian \cite[Eq.~(5.2)]{mansson}; the saturation below uses the
particular second-order coefficient $N$ displayed here.
For each $\lambda\ne0$ the residual below is nonzero, excluding a
regular analytic additive difference-form $R$ matrix. Put $u=\lambda^2$ and
\begin{align}
p(u)&=171u^3+378u^2+320u+108,\\
v(u)&=81u^4+72u^3+18u^2+12u+2.
\end{align}
The Weyl contractions derived below give
\begin{equation}
\boxed{\varepsilon^2=216\lambda^6p(u),\quad
\kappa_J=8\lambda^6p(u),\quad
\chi_J=\frac89v(u),\quad
\chi_E=\frac29(27u^2+12u+4).}\label{st:eq:polynomials}
\end{equation}
Here $\chi_J=C_J(0)$ and $\chi_E$ is the variance per site of the
centered energy density $e=h_\lambda-\I/3$. All three single-site
partial traces of the local current vanish, as do both partial traces
of $e$. The quotient residual has only span-three terms. These facts
make \eqref{st:eq:polynomials} exact for every $L\ge5$.

The generalized local charge-Gram spectrum is
\[
36u+36u^2\ (2),\qquad 12u+108u^2\ (4),\qquad
12u+144u^2\ (2).
\]
Thus every uniform one-site charge is broken for $\lambda\ne0$.
For $u<1/3$,
\begin{equation}
\delta=|\lambda|\sqrt{12+108u},\qquad
\Gamma_{\rm one}=4u+36u^2,\qquad
\frac{\kappa_J}{\chi_J}=\frac{9\lambda^6p(u)}{v(u)}
=486\lambda^6+O(\lambda^8).\label{st:eq:separation}
\end{equation}
Equations \eqref{st:eq:polynomials} and \eqref{st:eq:separation} prove both
the stability and balanced-response optimality assertions. They also give
$\kappa_J/\delta^6\to1/2$ and $\Gamma_J/\delta^6\to9/32$, proving
Corollary~\ref{st:cor:sixth}'s sharp exponent.
The normalization in \eqref{st:eq:scale} is regular:
here $a=-\lambda^2/4$ and $b=1+3\lambda^2/4$.

\subsection{Weyl reduction of the density}

The following derivation gives the coefficients in
\eqref{st:eq:polynomials} and the complete charge spectrum by scalar
contractions. Let \(\omega=e^{2\pi i/3}\),
\(X|j\rangle=|j+1\rangle\), \(Z|j\rangle=\omega^j|j\rangle\), and
\[
W_\alpha=\omega^{2\alpha_1\alpha_2}X^{\alpha_1}Z^{\alpha_2},
\qquad \alpha\in\mathbb F_3^2.
\]
All lattice arithmetic below is modulo three. These matrices obey
\[
W_\alpha^\dagger=W_{-\alpha},\quad
\tr(W_\alpha^\dagger W_\beta)=3\delta_{\alpha\beta},\quad
W_\alpha W_\beta=\omega^{2(\alpha_2\beta_1-\alpha_1\beta_2)}
W_{\alpha+\beta}.
\]
Define \(\eta(\alpha,\beta)\in\{0,1,-1\}\) by
\([W_\alpha,W_\beta]=i\sqrt3\,\eta(\alpha,\beta)W_{\alpha+\beta}\).
Explicitly, if \(k=2(\alpha_2\beta_1-\alpha_1\beta_2)\bmod3\),
then \(\eta=(0,1,-1)_k\). In particular,
\(\eta^2=1\) precisely when the two labels are linearly independent.

The density commutes with both \(X\otimes X\) and \(Z\otimes Z\), so
\begin{equation}
h_\lambda=\sum_{\alpha\in\mathbb F_3^2}
c_\alpha W_\alpha\otimes W_{-\alpha},\qquad
c_0=\tfrac13,\qquad c_{-\alpha}=\overline{c_\alpha}.
\label{st:eq:weyldensity}
\end{equation}
There are four nonzero lattice lines. Choose their representatives
\(a=(1,0)\), \(b=(0,1)\), \(c=(1,1)\), \(d=(1,2)\). Taking the
two-site trace against \(W_{-\alpha}\otimes W_\alpha\) gives
\begin{equation}
\begin{array}{c|rrrr}
\alpha&a&b&c&d\\ \hline
c_\alpha&\tfrac13&\tfrac13-u&
\tfrac13+u-i\lambda/\sqrt3&\tfrac13+u+i\lambda/\sqrt3 .
\end{array}
\label{st:eq:weylcoefficients}
\end{equation}
The table follows from a single three-dimensional block: in the ordered
two-site basis \((00,12,21)\) the density is
\[
\begin{pmatrix}
1-2u&\lambda-u&-\lambda-u\\
\lambda-u&u&1+2u\\
-\lambda-u&1+2u&u
\end{pmatrix},
\]
and the two other blocks are its translates by \(X\otimes X\).
Multiplication by the three-by-three Weyl matrices gives
\eqref{st:eq:weylcoefficients}.

\subsection{Susceptibilities and the eight one-site charges}

Write \(w_\alpha=|c_\alpha|^2\) for one representative of each line.
Equation~\eqref{st:eq:weylcoefficients} gives
\[
(w_a,w_b,w_c,w_d)=
\left(\tfrac19,\ \tfrac19-\tfrac23u+u^2,\
\tfrac19+u+u^2,\ \tfrac19+u+u^2\right).
\]
Orthogonality of Weyl words immediately yields
\(\chi_E=2\sum_{\alpha=a,b,c,d}w_\alpha\).
The local current is
\[
j=i[h_{12},h_{23}]
=-\sqrt3\sum_{\alpha,\beta}
\eta(-\alpha,\beta)c_\alpha c_\beta\,
W_\alpha\otimes W_{\beta-\alpha}\otimes W_{-\beta}.
\]
Each nonzero term has three nonzero labels, so all its one-site
partial traces vanish. The pair \((\alpha,\beta)\) specifies its word
uniquely. There are four choices of signs for each ordered pair of
distinct lines. Consequently
\begin{equation}
\chi_J=3\sum_{\alpha,\beta}\eta(-\alpha,\beta)^2
 |c_\alpha c_\beta|^2
=12\left[\left(\sum_{\alpha=a,b,c,d}w_\alpha\right)^2
 -\sum_{\alpha=a,b,c,d}w_\alpha^2\right].
\label{st:eq:weylsusceptibilities}
\end{equation}
These two scalar formulas give exactly the two susceptibilities in
\eqref{st:eq:polynomials}.

The charge-Gram operator commutes with conjugation by \(X\) and \(Z\).
Their eight nontrivial simultaneous characters on traceless matrices
are the eight \(W_\alpha\), so no characteristic polynomial is needed.
The coefficient of
\(W_{\beta+\alpha}\otimes W_{-\beta}\) in
\([h_\lambda,W_\alpha\otimes\I+\I\otimes W_\alpha]\) is
\(i\sqrt3\,\eta(\beta,\alpha)(c_\beta-c_{\beta+\alpha})\).
Since \(\norm{W_\alpha}_\HS^2=3\), its Gram eigenvalue is
\begin{equation}
g_\alpha=9\sum_{\beta\in\mathbb F_3^2}
\eta(\beta,\alpha)^2
|c_\beta-c_{\beta+\alpha}|^2.
\label{st:eq:weylcharge}
\end{equation}
Only six terms occur in each sum. Substituting the four coefficients gives
\[
g_a=12u+144u^2,\qquad
g_b=36u+36u^2,\qquad
g_c=g_d=12u+108u^2.
\]
The Hermitian matrices
\((W_\alpha+W_\alpha^\dagger)/\sqrt6\) and
\((W_\alpha-W_\alpha^\dagger)/(i\sqrt6)\) are an orthonormal real
pair with the same eigenvalue. This proves all eight eigenvalues,
their multiplicities, and the formula for \(\delta\).

\subsection{Six scalar residual contractions}

Expand the double commutator in the same basis:
\[
F(h_\lambda)=\sum_{\alpha,\beta}\rho_{\alpha\beta}
W_\alpha\otimes W_{\beta-\alpha}\otimes W_{-\beta}.
\]
Multiplying the current words once more by a density word gives the
explicit nine-term convolution
\begin{equation}
\begin{split}
\rho_{\alpha\beta}=-3\sum_{\gamma\in\mathbb F_3^2}c_\gamma
\big[&
c_{\alpha-\gamma}c_\beta
 \eta(\gamma-\alpha,\beta)\eta(\gamma,2\alpha-\beta)\\
&+c_\alpha c_{\beta-\gamma}
 \eta(-\alpha,\beta-\gamma)\eta(\gamma,2\beta-\alpha)
\big].
\end{split}
\label{st:eq:weylresidual}
\end{equation}
The first line comes from the commutator with \(h_{12}\); the
second comes from \(h_{23}\). This formula also fixes every sign.
For example, writing \(A=c_a\), \(B=c_b\), \(C=c_c\), and
\(\bar C=c_d\), the first contraction is
\[
\rho_{ba}=3\left[A B(A-B)+A\bar C^2+
B^2(C-\bar C)-B C\bar C\right]
=\lambda^3\left[\lambda(3u+2)+2i\sqrt3(1-u)\right].
\]
The remaining contractions require the same nine-term sum.
They reduce to the following six rows:
\begin{equation}
\begin{array}{c|c|r}
(\alpha,\beta)&\rho_{\alpha\beta}/\lambda^3&
\text{multiplicity}\\ \hline
(b,a)&\lambda(3u+2)+2i\sqrt3(1-u)&8\\
(b,c)&-2\lambda(3u+1)-i\sqrt3(u+4)&8\\
(b,d)&\lambda(4-6u)+i\sqrt3(5u+2)&8\\
(a,c)&-\lambda(9u+4)+i\sqrt3(3u+2)&8\\
(a,-d)&\lambda(3u+8)+i\sqrt3u&8\\
(c,d)&-2i\sqrt3(5u+2)&4
\end{array}
\label{st:eq:weylresidualtable}
\end{equation}
Here multiplicities can be counted without additional contractions.
The involutions
\[
(\alpha,\beta)\mapsto(-\alpha,-\beta),\qquad
(\alpha,\beta)\mapsto(J\alpha,J\beta),\qquad
(\alpha,\beta)\mapsto(\beta,\alpha),\qquad
J(r,t)=(r,-t),
\]
send \(\rho\) to its complex conjugate, its complex conjugate, and
its negative, respectively. The first two follow from the Weyl
coefficients. For the third, spatial reflection sends
\(F(h_\lambda)\) to \(-F(h_{-\lambda})\), while
\(c_\alpha(-\lambda)=\bar c_\alpha(\lambda)\).
The first five orbits have eight elements and the last has four.
Among the 64 pairs with nonzero endpoints, the other 20 coefficients
vanish: 16 have \(\beta=\pm\alpha\), and four are the orbit of
\((c,-d)\). Their vanishing follows directly from
\eqref{st:eq:weylresidual}.

Antisymmetry also gives \(\rho_{0\beta}=-\rho_{\beta0}\).
Those terms are exactly a boundary difference, whereas every word
with both endpoints nonzero is orthogonal to all boundary differences.
Thus the latter words are the quotient residual. Its only support
span is three, and
\[
\frac{\varepsilon^2}{27}=\kappa_J
=\sum_{\alpha,\beta\ne0}|\rho_{\alpha\beta}|^2.
\]
Taking squared moduli in \eqref{st:eq:weylresidualtable}, with the
listed multiplicities, gives
\[
\kappa_J=8\lambda^6(108+320u+378u^2+171u^3),
\]
which proves the remaining identities in \eqref{st:eq:polynomials}.
Weyl orthogonality and the periodic word lemma make these identities
exact for every \(L\ge5\).

The sharp exponents also have a shorter structural proof. For the
Fermat tensor, the contractions of Appendix~\ref{fs:app:obstructions}
give \(M=\I\), \(m=3\), \(\mathsf U=D\),
\(\mathsf F+\mathsf G=C_+\), and \(S=2\I\).
Hence the particular quadratic correction \(N_0\) is exactly \(N\).
Equation~\eqref{fs:eq:intrinsicnormal} therefore gives
\(\varepsilon(h_\lambda)=O(|\lambda|^3)\), while the absence of
one-site charges for \(T\), compactness of the unit charge sphere,
and continuity give
\(\delta(h_\lambda)=|\lambda|(\delta(T)+O(|\lambda|))\).
These two orders alone exclude any stability exponent greater than
\(1/3\). The norm comparison in Appendix~\ref{main:app:norm}
and \(\chi_J(P)>0\) similarly exclude an exponent below six in a
uniform lower bound for current curvature in terms of \(\delta\).
The six contractions above additionally determine the exact constants.

\subsection{The same conversion without compensation}
\label{st:app:rawcomparison}

For \(h_\lambda^{\rm raw}=P+\lambda T\), remove only the quadratic
terms from \eqref{st:eq:weylcoefficients}. The four coefficients are
\[
(c_a,c_b,c_c,c_d)=
\left(\tfrac13,\tfrac13,
\tfrac13-i\lambda/\sqrt3,\tfrac13+i\lambda/\sqrt3\right),
\]
with \(c_0=1/3\). This can also be read directly from its block
\(\left(\begin{smallmatrix}1&\lambda&-\lambda\\
\lambda&0&1\\-\lambda&1&0\end{smallmatrix}\right)\)
in the basis \((00,12,21)\). The weights are
\((1/9,1/9,1/9+u/3,1/9+u/3)\), so
\eqref{st:eq:weylsusceptibilities} gives
\begin{equation}
\chi_E^{\rm raw}=\frac49(2+3u),\qquad
\chi_J^{\rm raw}=\frac89(2+6u+3u^2).
\label{st:eq:rawsusceptibilities}
\end{equation}
Equation~\eqref{st:eq:weylcharge} gives \(12u\) on six Hermitian
charge directions and \(36u\) on the other two. Hence
\(\delta_{\rm raw}^2=12u\) and \(\Gamma_{\rm one}^{\rm raw}=4u\).

The same convolution \eqref{st:eq:weylresidual} and the same orbit
counting give the following four squared-modulus classes:
\begin{equation}
\begin{array}{c|c|r}
\text{representative orbits}&
|\rho_{\alpha\beta}^{\rm raw}|^2/\lambda^4&\text{multiplicity}\\ \hline
(b,a)&4/9&8\\
(b,c),\ (b,d)&(16+3u)/9&16\\
(a,c),\ (a,-d)&(4+3u)/9&16\\
(c,d)&4u/3&4
\end{array}
\label{st:eq:rawresidualtable}
\end{equation}
For example, the first, second, third, and last representatives have
\(\rho^{\rm raw}/\lambda^2=-2/3\),
\(-4/3-i\lambda/\sqrt3\), \(-2/3+i\lambda/\sqrt3\), and
\(2i\lambda/\sqrt3\), respectively. The other 20 coefficients with
nonzero endpoints vanish. The coefficients with a zero endpoint remain
an antisymmetric boundary difference, so the quotient residual again
has only span three. Summing the four rows gives
\begin{equation}
\kappa_J^{\rm raw}=\frac{16}{9}\lambda^4(22+9u),\qquad
\varepsilon_{\rm raw}^2=48\lambda^4(22+9u).
\label{st:eq:rawresidual}
\end{equation}
Equivalently, writing the projected residual as
\(\lambda^2R_2+\lambda^3R_3\), its three Gram contractions are
\(\tr(R_2^2)=1056\), \(\tr(R_3^2)=432\), and
\(\tr(R_2R_3)=0\). These identities are also obtained directly from
the local matrices by the accompanying comparison program.

Both paths have scalar partial traces and currents with vanishing
one-site partial traces. Thus the same periodic word lemma applies,
with no change of density representative or normalization. Dividing
\eqref{st:eq:rawresidual} by \eqref{st:eq:rawsusceptibilities} gives
\begin{equation}
\Gamma_J^{\rm raw}
=\frac{2\lambda^4(22+9\lambda^2)}{2+6\lambda^2+3\lambda^4}
=22\lambda^4+O(\lambda^6),
\label{st:eq:rawcurvature}
\end{equation}
which proves the raw-current curvature formula stated in Sec.~\ref{main:sec:model} for every \(L\ge5\).
The spectral inequality \eqref{st:eq:timebound} then gives
the two half-retention scalings stated in Sec.~\ref{main:sec:model}. These are certificates of retention; neither
curve determines an actual decay time.

\section{Resonant obstruction and local charge repair}\label{main:app:resonance}
We now prove a statement that does not use Appendix~\ref{fs:app:continuation}. For a completely
symmetric complex tensor $f\in\operatorname{Sym}^3\mathbb C^3$, define
the Hermitian intrinsic density
\begin{equation}
(T_f)^{ab}{}_{ij}
=\epsilon_{ijm}f^{abm}+\epsilon^{abm}\bar f_{ijm},\qquad
\norm{T_f}_\HS^2=4\norm f^2,\label{st:eq:intrinsic}
\end{equation}
where $\norm f^2=\sum_{a,b,c}|f^{abc}|^2$ includes all symmetric-index
multiplicities. It obeys $PT_fP=-T_f$. Put
$H_{0,L}=\sum_jP_{j,j+1}$, $V_{f,L}=\sum_j(T_f)_{j,j+1}$, and
$Q_L(q)=\sum_jq_j$ on a periodic chain.

\begin{lemma}[Finite-dimensional repair formula]\label{st:lem:pinching}
For a Hermitian $H_0=\sum_EEP_E$, let
$\mathcal P(A)=\sum_EP_EAP_E$. For any $A$,
\begin{equation}
\inf_X\norm{[H_0,X]+A}_{\rm op}=\norm{\mathcal P(A)}_{\rm op}.
\label{st:eq:pinching}
\end{equation}
If $A$ is anti-Hermitian, the infimum may be restricted to Hermitian $X$.
\end{lemma}
\begin{proof}
For each distinct energy \(E\), let \(S_E=2P_E-\I\). It is a
Hermitian unitary, and the map
\[
\mathcal C_E(A)=\tfrac12(A+S_EAS_E)
\]
has operator norm at most one. On a block \(P_FA P_G\), it retains
the block exactly when either both \(F,G\) equal \(E\) or neither
does. Consequently the finite product of the maps \(\mathcal C_E\)
is precisely \(\mathcal P\). This proves operator-norm
contractivity while keeping repeated eigenvalues in entire blocks.
Moreover,
\[
\mathcal P([H_0,X]+A)=\mathcal P(A),
\]
so every residual has norm at least \(\norm{\mathcal P(A)}_{\rm op}\).

Set the diagonal blocks of \(X_*\) to zero and choose
\[
P_EX_*P_{E'}=-\frac{P_EAP_{E'}}{E-E'}\qquad(E\ne E').
\]
Then \([H_0,X_*]+A=\mathcal P(A)\), which attains the lower bound.
When \(A^\dagger=-A\), conjugation of the \((E',E)\) block changes
the sign of both the numerator and the energy difference, giving
\((P_{E'}X_*P_E)^\dagger=P_EX_*P_{E'}\). Thus \(X_*\) is
Hermitian. In particular, exact repair is possible if and only if
\(\mathcal P(A)=0\). A unitary eigenbasis change preserves the
operator norm and commutator, so the construction applies to any
finite-dimensional Hermitian \(H_0\).
\end{proof}
The obstruction from degenerate blocks is standard linear algebra,
also present in finite-periodic generator analyses \cite{vanovac}.
Theorem~\ref{st:thm:resonance} gives an explicit bound for every intrinsic
direction, every one-site charge, and every periodic length.

Theorem~\ref{st:thm:resonance} is stated in the main text.
\begin{proof}[Proof of Theorem~\ref{st:thm:resonance}]
The normalized one-magnon states
\[
|k,a\rangle=L^{-1/2}\sum_{j=0}^{L-1}e^{ikj}
|0\cdots a_j\cdots0\rangle,\qquad a=1,2,
\]
have the same exact exchange energy $E_L(k)=L-2+2\cos k$. The same
holds after any common $SU(3)$ rotation. Every matrix element of
$[H_{0,L},Q_1]$ between these rotated states vanishes, regardless of
$Q_1$.

Set $(A_f)_{ab}=\langle a0|T_f|0b\rangle$ for $a,b=1,2$. The
on-site contributions of the two adjacent bonds cancel by flip
oddness; hopping contributions carry phases $e^{ik}$ and $-e^{-ik}$.
Consequently
\begin{equation}
\langle k,a|V_{f,L}|k,b\rangle=2i\sin k\,(A_f)_{ab},\quad
A_f=\begin{pmatrix}
f^{012}-\bar f_{012}&-f^{011}-\bar f_{022}\\
f^{022}+\bar f_{011}&-f^{012}+\bar f_{012}
\end{pmatrix}.\label{st:eq:magnoncompression}
\end{equation}
This is a compression to degenerate states, not an invariant-subspace
assertion for $V_{f,L}$.

The complex representation $\operatorname{Sym}^3\mathbb C^3$ is
irreducible of dimension ten and is not self-dual. Schur orthogonality
and the multiplicity-three tensor normalization give
\[
\int|f_U^{011}|^2\,dU=\int|f_U^{022}|^2\,dU=\norm f^2/30,
\qquad \int f_U^{011}f_U^{022}\,dU=0.
\]
Thus $\int|(A_{f_U})_{12}|^2\,dU=\norm f^2/15$. For some rotation,
the corresponding resonant matrix element has modulus at least
$2|\sin k|\norm f/\sqrt{15}$.

The Hermitian density $i[T_f,q_1+q_2]$ is again intrinsic, with tensor
$f_q$ satisfying $2\norm{f_q}=\norm{[T_f,q_1+q_2]}_\HS$. Apply the
preceding bound to it. The $Q_1$ matrix element vanishes in the same
states, proving \eqref{st:eq:resonance}. For every $L\ge3$ there is an
allowed $k$ with $\sin k\ne0$, so the repair equation implies local
commutation. The converse takes $Q_1=0$.
\end{proof}

For the Fermat tensor $f^{000}=f^{111}=f^{222}=1$, the infinitesimal
action obeys
\[
\norm{q\cdot f}^2
=9\sum_a|q_{aa}|^2+3\sum_{a\ne b}|q_{ab}|^2
\ge3\norm q_\HS^2.
\]
Choosing an allowed momentum with $|\sin k|\ge\sqrt3/2$ gives, for
every $L\ge3$ and every $\norm q_\HS=1$,
\begin{equation}
\inf_{Q_1}\norm{[H_{0,L},Q_1]+[V_{f,L},Q_L(q)]}_{\rm op}
\ge\sqrt{3/5}.\label{st:eq:fermatobstruction}
\end{equation}
Such a momentum is obtained by taking the nearest integer to $L/4$
for $L\ge6$; the three smaller lengths satisfy the same bound directly.
No charge $Q_L(q)+\lambda Q_1+O(\lambda^2)$ can therefore have an
$O(\lambda^2)$ commutator with
$H_{0,L}+\lambda V_{f,L}+\lambda^2W_L$, for any $W_L$.
In particular, the compensation $N$ that suppresses current breaking
does not repair these charges. A translation-invariant bulk repair of
any fixed finite range would periodize to all sufficiently long rings
and is likewise excluded. This operator-norm obstruction is not a
per-site infinite-temperature decay-rate bound.

\subsection{Why an open-chain bilocal generator does not repair the ring}
Set $r_f=-iT_fP/2$ and $X_L=\sum_{a<b}(r_f)_{ab}$. Then $r_f$ is
Hermitian and $i[r_f,P]=T_f$. For an open chain, each adjacent
permutation commutes with the paired sum of terms connecting that bond
to any external site. Only the term on the bond remains, giving
\begin{equation}
i[X_L,H_{0,L}^{\rm open}]=V_{f,L}^{\rm open}.\label{st:eq:bilocal}
\end{equation}
Jacobi therefore makes $Y_L(q)=i[X_L,Q_L(q)]$ an open-chain
first-order charge correction. Contraction of a symmetric tensor index
pair against the Levi--Civita tensor makes both partial traces of
$r_f$ vanish, as do those of its charge commutator. Distinct site pairs
are orthogonal, so
\begin{equation}
\frac{\Tr Y_L(q)^2}{3^L}
=\frac{\binom L2}{36}\norm{[T_f,q_1+q_2]}_\HS^2.
\label{st:eq:bilocalvariance}
\end{equation}
For a unit Fermat charge this lies between $\binom L2/3$ and
$\binom L2$, growing quadratically rather than extensively.
This statement concerns the displayed representative. Another
open-chain solution is $Y_L+K$ with $[H_{0,L}^{\rm open},K]=0$;
the minimum normalized squared Hilbert--Schmidt norm over these solutions is
$3^{-L}\norm{(1-\mathcal P_{\rm open})Y_L}_\HS^2$, whose volume scaling
has not been determined here. The periodic closing bond also spoils
\eqref{st:eq:bilocal} by reversing
the ordering of intervening pairs. Theorem~\ref{st:thm:resonance} rules
out any alternative finite-periodic correction of these charges.

\subsection{Local repair of the invariant commuting charges}
The same bilocal construction does give local corrections to the
$SU(3)$-invariant charges. We make the mechanism and its range explicit.

Proposition~\ref{st:prop:scalarrepair} is stated in the main text.
\begin{proof}[Proof of Proposition~\ref{st:prop:scalarrepair}]
Expand $r_f=\sum_{a,b}c_{ab}t^a\otimes t^b$ in traceless single-site
generators. For any interval $K\supset I$, the terms of $X_K$
connecting an external site $k<\min I$ to $I$ sum to
$\sum_{a,b}c_{ab}t_k^a\sum_{j\in I}t_j^b$ and thus commute with
$A_I$. Sites to the right of $I$ give the same cancellation. Terms
fully outside $I$ commute as well. Therefore
\begin{equation}
i[X_K,A_I]=i[X_I,A_I],\qquad
\norm{\mathcal D_f A_I}_{\rm op}
\le |I|(|I|-1)\norm{r_f}_{\rm op}\norm{A_I}_{\rm op}.
\label{st:eq:scalarlocality}
\end{equation}
The generators $t^a,t^b$ need not commute with each other; the target
density must commute with each total generator on its support.

For two invariant local operators, choose an interval containing both
supports and apply Jacobi using its $X_K$. This gives
$\mathcal D_f[A,B]=[\mathcal D_fA,B]+[A,\mathcal D_fB]$.
The map also commutes with translations and is therefore well-defined
on invariant local densities modulo total differences. In a local
operator-word basis, a vanishing total commutator has zero sums of
coefficients on each translation orbit. Its density is thus a finite
total difference, whose primitive can be averaged to be invariant.
Applying $\mathcal D_f$ to the bulk relations
$[C_\alpha^{(0)},C_\beta^{(0)}]=0$ proves
\eqref{st:eq:mutualrepair}. Since $\mathcal D_fP=i[r_f,P]=T_f$,
the first identity in Eq.~\eqref{st:eq:mutualrepair} follows. These are finite-range density
identities and can be periodized once the ring exceeds their support
spans; for a collection with maximum range $R$, $L>2R$ suffices.
This does not require periodizing the bilocal generator itself.
If only the total charges are given as invariant, their densities can
first be averaged over $SU(3)$ without changing the charges or ranges.
The opposite one-site conclusion is Theorem~\ref{st:thm:resonance}.
\end{proof}

For the exchange current $j_0=i[P_{12},P_{23}]$, this construction gives
$\mathcal D_fj_0=i([T_{12},P_{23}]+[P_{12},T_{23}])$, exactly its
first-order current correction. It also applies to every other
finite-range invariant commuting charge, so the first-order effect
is not specific to the energy current. Within the degenerate energy
blocks, the first identity in Eq.~\eqref{st:eq:mutualrepair} further implies
$[\mathcal P_L(V_f),C_\alpha^{(0)}]=0$; the same resonant perturbation
need not commute with the uniform one-site $SU(3)$ charges.

This is an application of bilocal localization, closely related to
Eq.~(12) of \cite{surace} and the finite-periodic distinction in
Sec.~II.3 of \cite{vanovac}. A recent quantum-group construction also
describes long-range deformations to first order and discusses extensions
using Lie-algebra charges \cite[Secs.~3.1--3.3.3]{schouten}.
This local repair is not evidence that the perturbation
lies outside weak integrability-breaking constructions. Moreover,
$\mathcal D_f A_I$ is generally no longer invariant, so the locality
argument cannot simply be iterated. Second-order repair of the entire
commuting hierarchy is not proved here. The compensated Fermat
current's $O(\lambda^3)$ nonconservation remains a separate exact result.

\section{Guaranteed current retention and energy spreading}\label{main:app:timebounds}
\subsection{Initial-charge drift along exact arcs}

Use the exact analytic charge of Theorem~\ref{fs:thm:main} and the
notation of Sec.~\ref{main:sec:exact}, with $\norm{q_0}_{\rm op}=1$.
Conservation of $Q_L(\tau)$ gives
\begin{equation}
\frac1L\norm{Q_L^0(t)-Q_L^0}_{\rm op}
\le2\norm{q(\tau)-q_0}_{\rm op}\le2M_q|\tau|.
\label{main:eq:alltime}
\end{equation}
Indeed, the difference is the change under time evolution of
$Q_L^0-Q_L(\tau)$, whose norm is at most
$L\norm{q_0-q(\tau)}_{\rm op}$. Analyticity supplies a finite $M_q$.
Differentiating local conservation at $\tau=0$ also gives
$[K,q_0\otimes\I+\I\otimes q_0]=0$, since $P$ commutes with every
total one-site operator. Write $h(\tau)=P+\tau K+\tau^2r(\tau)$
with $\norm{r(\tau)}_{\rm op}\le M_r$. The commutator inequality
then yields $\norm{[H_L(\tau),Q_L^0]}_{\rm op}/L\le4M_r\tau^2$.
Its time integral, combined with \eqref{main:eq:alltime}, proves
\eqref{main:eq:drift}. The time-independent term uses the exact analytic
charge, beyond the first-order restriction.

\subsection{Current retention and energy spreading}

For any finite Hermitian $H,J$, positivity of the energy-basis spectral
weights and $0\le1-\cos x\le x^2/2$ imply
\begin{equation}
0\le\chi_J-C_J(t)\le\frac12\kappa_Jt^2.\label{st:eq:timebound}
\end{equation}
For the Fermat model this guarantees $C_J(t)\ge\chi_J/2$ whenever
\begin{equation}
|t|\le t_{\rm cert}(\lambda):=
\sqrt{\frac{v(u)}{9\lambda^6p(u)}}
\sim\frac{|\lambda|^{-3}}{\sqrt{486}}.\label{st:eq:window}
\end{equation}
The statement is uniform in $L\ge5$ and supplies a lower bound on a
possible half-decay time, without asserting that half-decay occurs.

On the infinite chain define
$M_2(t)=\chi_E^{-1}\sum_xx^2\langle e_x(t)e_0\rangle_0$.
The continuity equation and discrete summation by parts give
$M_2''(t)=2C_J^\infty(t)/\chi_E$, with
$C_J^\infty(t)=\sum_x\langle j_x(t)j_0\rangle_0$ as in the main
text; locality justifies the spatial sums at every fixed time~\cite{bravyi}.
The scalar marginals give $M_2(0)=0$, and cyclicity and translation invariance
make $M_2(t)$ even. Integrating \eqref{st:eq:timebound} yields
\begin{equation}
\frac{\chi_J}{\chi_E}t^2-\frac{\kappa_J}{12\chi_E}t^4
\le M_2(t)\le\frac{\chi_J}{\chi_E}t^2.\label{st:eq:spreading}
\end{equation}
Throughout \eqref{st:eq:window}, the lower bound is at least
$(11/12)(\chi_J/\chi_E)t^2$. This is a finite-time ballistic
second-moment bound. The correlation profile need not be a positive
probability distribution, and no infinite-time Drude weight follows.

\subsection{An explicit finite-volume resonant conversion}
In the original physical basis take the uniform background
$|u\rangle=(|0\rangle+|1\rangle)/\sqrt2$ and the two excitation
flavors $|v\rangle=(-|0\rangle+|1\rangle)/\sqrt2$,
$|w\rangle=|2\rangle$. Equivalently, rotate the density by
\[
U=\begin{pmatrix}1/\sqrt2&1/\sqrt2&0\\
-1/\sqrt2&1/\sqrt2&0\\0&0&1\end{pmatrix}\in SU(3).
\]
Its Fermat tensor has $f^{011}=1/\sqrt2$ and
$f^{022}=f^{012}=0$. Equation~\eqref{st:eq:magnoncompression} then
gives the exact compression
\begin{equation}
P_{\rm mag}V_{f,L}P_{\rm mag}=\sqrt2\sin k\,\sigma_y,
\qquad \langle k,w|V_{f,L}|k,v\rangle=i\sqrt2\sin k.
\label{st:eq:conversion}
\end{equation}
By the variational principle, the complete exchange-energy cluster
containing these states has a first-order splitting width at least
$2\sqrt2|\lambda\sin k|-O_L(\lambda^2)$. The two eigenvalues of the
compression are not asserted to be eigenvalues of the full perturbation.

For fixed $L$, write $V_{\rm res}=\sum_E\Pi_EV_{f,L}\Pi_E$ using
the complete energy projections of $H_{0,L}$. Finite-dimensional
averaging gives, for fixed bounded $s$ and $\lambda\to0^+$,
\begin{equation}
e^{iH_{0,L}s/\lambda}e^{-iH_{\lambda,L}s/\lambda}
=e^{-isV_{\rm res}}+O_L(\lambda).
\label{st:eq:averaging}
\end{equation}
Indeed a first-order block change of basis removes off-diagonal energy
blocks with denominators $E-E'\ne0$; its $O_L(\lambda^2)$ Hamiltonian
remainder accumulates only $O_L(\lambda)$ error on this time scale.
Equations~\eqref{st:eq:conversion}--\eqref{st:eq:averaging} imply, for
small fixed $s$,
\begin{equation}
\big|\langle k,w|e^{-iH_{\lambda,L}s/\lambda}|k,v\rangle\big|^2
=2\sin^2 k\,s^2+O_L(s^3)+O_L(\lambda).\label{st:eq:probability}
\end{equation}
On the same time scale the current bound, for every $L\ge5$, reads
\[
0\le1-C_J(s/\lambda)/\chi_J
\le243s^2\lambda^4+O(s^2\lambda^6).
\]
Thus the same Hamiltonian has a controlled resonant conversion and
strongly suppressed current drift. The protocols use different states
and ensembles. The error constants in \eqref{st:eq:averaging} and
\eqref{st:eq:probability} depend on $L$ and its spectral gaps, unlike the
current bound. No thermodynamic $\lambda^{-1}$ flavor-relaxation law
or isolated magnon band has been established.

\section{Controls, recoverability, and remaining questions}\label{main:app:controls}
A matched exact YBE control is
$h_\theta=(U_\theta\otimes U_\theta^\dagger)P$ with
$U_\theta=\operatorname{diag}(e^{i\theta},1,e^{-i\theta})$.
It is a Hermitian involutive braid generator, so
$\check R(u)=\I+u h_\theta$ satisfies additive braid YBE.
It preserves diagonal charges and has
$\kappa_J=0$, $\chi_J=16/9$, $\chi_E=8/9$.
The charge-preserving control $P+\lambda Z\otimes Z$, with
$Z=\operatorname{diag}(1,-1,0)$, instead has
\[
\kappa_J=\frac{16}{9}\lambda^2(\lambda^2+2\lambda+2),\quad
\chi_J=\frac{16}{9}(\lambda^2+\lambda+1),\quad
\chi_E=\frac49(\lambda^2+\lambda+2).
\]
Its normalized current curvature begins at $2\lambda^2$, despite two
exact diagonal charges. The corresponding bound from
\eqref{st:eq:timebound} guarantees only an $O(|\lambda|^{-1})$ window.
This comparison concerns guaranteed windows, not actual relaxation times.

\subsection{Trace derivation for the comparison controls}

The exact control acts as
\(h_\theta|ab\rangle=e^{i(\theta_b-\theta_a)}|ba\rangle\).
Its square is the identity, and both three-site braid products multiply
\(|abc\rangle\) by \(e^{2i(\theta_c-\theta_a)}\) and reverse the
three letters. Hence \(F(h_\theta)=2\Delta h_\theta\).
The common three-site change of basis
\(\I\otimes U_\theta^\dagger\otimes U_\theta^{-2}\)
maps the two exchange bonds to the two \(h_\theta\) bonds. It preserves
the current norm and all partial traces. This proves its stated
susceptibilities and zero curvature.

For the \(Z\otimes Z\) control, put \(X=P_{12}\), \(Y=P_{23}\),
and use the diagonal three-site operators
\[
A=(Z_1-Z_2)Z_3,\qquad B=Z_1(Z_2-Z_3).
\]
Moving diagonal factors through the two permutations gives
\[
[h_{12},h_{23}]=C_0+\lambda C_1,\qquad
C_0=XY-YX,\quad C_1=AX+BY.
\]
The required traces are consequences of the elementary permutation rule
\begin{equation}
\Tr\!\left[P_\pi(A_1\otimes\cdots\otimes A_n)\right]
=\prod_{\text{cycles }(i_1\ldots i_k)\text{ of }\pi}
\tr(A_{i_k}\cdots A_{i_1}),
\label{st:eq:permutationtrace}
\end{equation}
with the cycle order fixed by the convention for \(P_\pi\).
For diagonal \(A_i\) their order is immaterial. Here
\(\tr Z=\tr Z^3=0\), \(\tr Z^2=\tr Z^4=2\), and
\(\tr\I=3\). Thus, for \(j_r=iC_r\),
\[
\left(\Tr(j_rj_s)\right)_{r,s=0,1}
=\begin{pmatrix}48&24\\24&48\end{pmatrix}.
\]
More explicitly, the three entries are
\(2(3^3-3)\), \(4(3\tr Z^2-(\tr Z)^2)\), and
\(4\cdot3(\tr Z^2)^2\).
For the last one, \(\Tr A^2=\Tr B^2=24\); the mixed
three-cycle trace vanishes because it sets all three eigenvalues equal.
The current has zero one-site partial traces. Dividing its norm by 27
gives \(\chi_J\). Similarly
\(\Tr h=3\), \(\Tr h^2=9+4\lambda+4\lambda^2\), and the
two partial traces of \(h-\I/3\) vanish. This gives \(\chi_E\).

To compute the quotient residual, define
\[
K=(Z_1-Z_2)^2P,\qquad W=Z^2-\tfrac23\I,\qquad
O=2BXY+2AYX.
\]
Both partial traces of \(K\) vanish, since a partial trace of a
weighted swap sets its two eigenvalues equal. Expanding one more
commutator gives
\begin{align*}
F(h)&=2\Delta P+\lambda\left(2\Delta(Z\otimes Z)+O\right)\\
&\quad+\lambda^2\left(-K_{12}Z_3^2+Z_1^2K_{23}\right).
\end{align*}
For example, \([Z_1Z_2+Z_2Z_3,AX]=-A^2X\) and the
corresponding \(BY\) term is \(B^2Y\), giving the last line.
The unnormalized endpoint traces of \(O\) are \(-2K\) and \(2K\).
Therefore its two quotient coefficients are
\[
r_1=O-\tfrac23\Delta K,\qquad
r_2=-K_{12}W_3+W_1K_{23}.
\]
They have only span-three terms. Their Gram matrix is
\begin{equation}
\left(\Tr(r_rr_s)\right)_{r,s=1,2}
=\begin{pmatrix}96&48\\48&48\end{pmatrix}.
\label{st:eq:controlresidualgram}
\end{equation}
Here are scalar contractions establishing all three entries.
For \(z_a\in\{1,-1,0\}\),
\[
\norm K_\HS^2=\sum_{a,b}(z_a-z_b)^4=36,\qquad
\tr W^2=\tfrac23,\qquad \norm O_\HS^2=8\Tr B^2=192.
\]
The orthogonal boundary part has norm
\(\norm{(2/3)\Delta K}_\HS^2=96\), leaving
\(\norm{r_1}_\HS^2=96\).
The two summands of \(r_2\) are orthogonal by
\eqref{st:eq:permutationtrace}, so
\(\norm{r_2}_\HS^2=2\cdot36\cdot(2/3)=48\).
Finally their cross contraction is
\[
\Tr(r_1r_2)=\Tr(Or_2)
=-8\sum_{a,b}z_a(z_b-z_a)^3(z_a^2-\tfrac23)=48.
\]
The sum is \(-6\): the rows with \(z_a=1,-1,0\) contribute
\(-3,-3,0\), respectively.
Division of the quadratic form in
\eqref{st:eq:controlresidualgram} by 27 gives
\(\kappa_J=16\lambda^2(\lambda^2+2\lambda+2)/9\).
This derives all control coefficients from the permutation trace rule.

\subsection{Recovery and calibration}

Charge recovery itself is an eight-dimensional Gram eigenproblem. If a
measured density $\widetilde h$ obeys
$\norm{h-\widetilde h}_\HS\le\eta$, a normalized smallest-eigenvalue
vector $\widetilde q$ gives the explicit certificate
\[
\norm{[h,\widetilde q_1+\widetilde q_2]}_\HS
\le\sqrt{\lambda_{\min}(G_{\widetilde h})}+4\eta.
\]
This is standard linear algebra. The new structural estimate guarantees
a small eigenvalue from a small residual, but does not yet provide
numerical uniform values of $C,r_0$. Subspace recovery must be judged
against spectral gaps, especially near the eightfold degeneracy at $P$.

The explicit current window also admits a calibration certificate.
For $\widetilde h=h_\lambda+E$ with $\norm E_\HS\le\eta$, take
$M\ge\max(\norm{h_\lambda}_{\rm op},\norm{\widetilde h}_{\rm op})$.
Expanding the two commutators in $F$ and the local current gives
\[
\widetilde\varepsilon\le\varepsilon_\lambda+24\sqrt3M^2\eta,
\qquad
\sqrt{\widetilde\chi_J}\ge\sqrt{\chi_{J,\lambda}}
-4M\eta/\sqrt3.
\]
For the second bound, Lemma~\ref{st:lem:norm} applied to the traceless
current difference bounds its periodic variance by its local squared
norm divided by nine. If the numerator is positive,
\begin{equation}
|t|\le t_{\rm cert}:=
\frac{3\sqrt{\chi_{J,\lambda}}-4\sqrt3M\eta}
{\varepsilon_\lambda+24\sqrt3M^2\eta}
\quad\Longrightarrow\quad
\widetilde C_J(t)\ge\widetilde\chi_J/2.\label{st:eq:noise}
\end{equation}
An analogous two-site translation bound gives
$\widetilde\Gamma_{\rm one}\ge
(\delta(h_\lambda)-4\sqrt2\eta)_+^2/3$, even when $E$ is not
balanced. Thus $\eta=O(|\lambda|^3)$ is a sufficient calibration
precision to retain a cubic current window and exclude exact uniform
one-site charges for small nonzero $\lambda$. These conservative
bounds do not determine an optimal noise tolerance or effective
constants for Theorem~\ref{st:thm:sharp}.

Weak integrability-breaking constructions already yield high-order
quasi-conserved quantities and rigorous persistence bounds
\cite{surace}; their finite-periodic realization has also been studied
\cite{vanovac}. Recent exact results derive commuting local hierarchies
or Yang--Baxter realizations from the Reshetikhin conservation law
under their stated hypotheses \cite{hokkyo,ssi}. These exact criteria
concern vanishing residual; the estimate here controls the least
one-site charge defect also when the residual is nonzero.
The structural contribution is the sharp
qutrit-neighborhood constraint, its explicitly saturating nearest-neighbor
response, and the resonant obstruction to repairing the broken one-site
charges. The general implication from quasi-conservation to slow change
is established background. Proposition~\ref{st:prop:scalarrepair} identifies the
first-order repair of the invariant hierarchy; its higher-order repair
and effective uniform stability constants remain separate questions.

The structural claims do not require a thermodynamic relaxation theorem.
The operator-norm bound alone does not establish macroscopic thermal
weight: for
$F_q=i[V_f,Q_L(q)]$, it only gives
$\norm{\mathcal P_L(F_q)}_\HS^2/(L3^L)\ge3/(5L3^L)$ in the unit
Fermat case. Nor does it select a thermodynamic time scale. A future
dynamical result could establish finite memory change at some
$t_q=o(|\lambda|^{-3})$ without presupposing $t_q\sim|\lambda|^{-1}$
or irreversible relaxation. The next two appendixes prove thermal
protection and a fixed-strength complete-memory contraction. They narrow
this question without establishing a fixed fractional loss as
$\lambda\to0$.

\section{Thermal protection by a finite-range bilocal correction}
\label{main:app:thermal}

\subsection{Local identities and tracial norms}

Write $H_{\lambda,L}=H_{0,L}+\lambda V_L+\lambda^2W_L$, with
$H_{0,L}=\sum_jP_{j,j+1}$, $V_L=\sum_jT_{j,j+1}$, and
$W_L=\sum_jN_{j,j+1}$. Throughout this appendix, $T=T_f$ is
Hermitian intrinsic, $N$ is fixed Hermitian with scalar partial traces,
and $q$ is nonzero traceless Hermitian. The raw path has $N=0$.
Let
\begin{equation}
Q_L=\sum_jq_j,\quad \chi_q=\frac{\tr q^2}{3},\quad
S=i[T,q_1+q_2],\quad r=-\frac{i}{2}SP.
\label{main:eq:thermalr}
\end{equation}
The intrinsic tensor formula and permutation trace identities in
Appendix~\ref{main:app:resonance} imply
\begin{equation}
S=S^\dagger,\quad r=r^\dagger,\quad PSP=-S,\quad PrP=-r,
\quad i[r,P]=S,\quad
\Tr_1S=\Tr_2S=\Tr_1r=\Tr_2r=0.
\label{main:eq:thermalidentities}
\end{equation}
Indeed, the intrinsic subspace and its zero-partial-trace conditions,
also after multiplication by $P$, are invariant under simultaneous
one-site conjugation. Differentiating that action gives the same
conditions for $S$. Exchange oddness makes $SP$ anti-Hermitian and
gives the commutator identity for $r$.

Use $\langle A,B\rangle_0=3^{-L}\Tr(A^\dagger B)$ and
$\norm A_0^2=\langle A,A\rangle_0$ on a ring, with the analogous
normalized trace on a local support. Define the nonnegative constants
\begin{equation}
\sigma^2=\frac{\tr S^2}{9},\quad s_q=\frac{\sigma}{\sqrt{\chi_q}},
\quad a_T=\norm T_{\rm op},\quad a_N=\norm N_{\rm op},\quad
n_q=\frac{\norm{i[N,q_1+q_2]}_\HS}{3\sqrt{\chi_q}}.
\label{main:eq:thermalconstants}
\end{equation}
The autocorrelation is real and even in time by its energy-basis
spectral representation. Tracial unitarity gives
\begin{equation}
1-c_{q,L}(t)=\frac{\norm{Q_L(t)-Q_L}_0^2}{2L\chi_q}.
\label{main:eq:chargeunitarity}
\end{equation}

\subsection{An exact truncated identity}

For an integer $R\ge1$ and $L\ge2R+3$, use oriented pairs
$r_{ba}=-r_{ab}$ and set
\begin{equation}
w_d=1-\frac{d}{R+1},\quad
X_{R,L}=\sum_j\sum_{d=1}^Rw_dr_{j,j+d},\quad
s_R=\sum_{d=1}^Rw_d^2=\frac{R(2R+1)}{6(R+1)}\le\frac R3.
\label{main:eq:cutoffweights}
\end{equation}
No unoriented pair repeats under the length hypothesis. Distinct pair
supports are orthogonal because $r$ has zero partial traces. Since
$\norm r_{0,2}^2=\sigma^2/4$,
\begin{equation}
\frac{\norm{X_{R,L}}_0^2}{L}=\frac{\sigma^2s_R}{4}.
\label{main:eq:Xvariance}
\end{equation}
For each nearest-neighbor bond introduce
\begin{equation}
D_{j,k}=i(r_{j,k}-r_{j+1,k})P_{j,j+1},\qquad
K_{j,R}=\{j-R,\ldots,j-1,j+2,\ldots,j+R+1\}.
\end{equation}
Commutation with each permutation gives the exact finite-ring identity
\begin{align}
F_{S,L}:=\sum_jS_{j,j+1}&=i[X_{R,L},H_{0,L}]+B_{R,L},
\label{main:eq:cutoffidentity}\\
B_{R,L}&=\frac{1}{R+1}\left(\sum_jS_{j,j+1}
             +\sum_j\sum_{k\in K_{j,R}}D_{j,k}\right).
\label{main:eq:cutoffB}
\end{align}
To check the coefficients, the pair on the bond contributes $w_1S$.
For a site outside that bond, the two pairs joining it to the bond
contribute a difference of weights. Extend $w_d$ by zero for $d>R$.
Each of the $2R$ nonzero differences is $-1/(R+1)$, including
the two cutoff endpoints. Finally $1-w_1=1/(R+1)$.
These observations prove \eqref{main:eq:cutoffidentity} without taking
a limit or discarding an endpoint term.

Each $D_{j,k}$ has zero partial trace over all three sites. Tracing
the spectator $k$ uses the zero marginals of $r$. Tracing either bond
site cancels the two terms: in a decomposition
$r=\sum_{a,b}c_{ab}t^a\otimes t^b$ with traceless $t^a$, the
identities $\Tr_1[(t^a\otimes\I)P]=t^a$ and
$\Tr_1[(\I\otimes t^a)P]=t^a$ give identical contractions.
Also,
\begin{equation}
\norm{D_{j,k}}_{0,3}^2
=\norm{r_{j,k}-r_{j+1,k}}_{0,3}^2=\frac{\sigma^2}{2}.
\end{equation}
Operators with different exact three-site supports are orthogonal and
are orthogonal to the two-site $S$ terms. A three-site set contains
at most two ring bonds for $L\ge5$, so its multiplicity in the
double sum is at most two. Cauchy--Schwarz within each repeated support,
and $2R$ terms per bond, therefore give
\begin{equation}
\frac{\norm{B_{R,L}}_0^2}{L}
\le\sigma^2\frac{2R+1}{(R+1)^2}.
\label{main:eq:Bvariance}
\end{equation}

\subsection{Commutators with the perturbations}

For any Hermitian two-site $y$ with scalar partial traces and
$Y_L=\sum_jy_{j,j+1}$, we claim
\begin{equation}
\frac{\norm{i[Y_L,X_{R,L}]}_0^2}{L}
\le\norm y_{\rm op}^2\sigma^2(3+16s_R).
\label{main:eq:YXvariance}
\end{equation}
There are two support types. Coincident pairs have the traceless
density $iw_1[y,r]$. Only its own and its two neighboring translates
can overlap. Their total variance per site is at most
$3w_1^2\norm{[y,r]}_{0,2}^2\le3\norm y_{\rm op}^2\sigma^2$.

Terms sharing exactly one site have zero partial trace over each of
their three sites. For a nonshared site, the marginal of $y$ is
scalar or that of $r$ is zero. For the shared site, expanding in
one-site tensors leaves the trace of a commutator. Consequently these
terms are orthogonal to the coincident-pair terms. A given three-site
set has at most two choices of the $y$ bond and at most two choices
of the $r$ pair for each bond, hence multiplicity at most four.
Each unweighted term has squared norm at most
$\norm y_{\rm op}^2\sigma^2$, by
$\norm{[y,r]}_0\le2\norm y_{\rm op}\norm r_0$.
The weighted count per site is
$2w_1^2+4\sum_{d=2}^Rw_d^2\le4s_R$.
Cauchy--Schwarz within each support gives a variance at most
$16\norm y_{\rm op}^2\sigma^2s_R$.
Adding the orthogonal support types proves \eqref{main:eq:YXvariance}.

\subsection{Full dynamics and the order of limits}

Set $\widetilde Q_L=Q_L+\lambda X_{R,L}$. Equation
\eqref{main:eq:cutoffidentity} yields
\begin{equation}
i[H_{\lambda,L},\widetilde Q_L]
=\lambda B_{R,L}+\lambda^2i[V_L,X_{R,L}]
 +\lambda^2i[W_L,Q_L]+\lambda^3i[W_L,X_{R,L}].
\label{main:eq:fullthermalforce}
\end{equation}
Tracial unitarity and time integration under $H_{\lambda,L}$ imply
\begin{equation}
\norm{Q_L(t)-Q_L}_0\le2|\lambda|\norm{X_{R,L}}_0
       +|t|\norm{i[H_{\lambda,L},\widetilde Q_L]}_0.
\end{equation}
The density $i[N,q_1+q_2]$ has zero marginals, so its translates
are orthogonal and
$\norm{i[W_L,Q_L]}_0^2/L=\chi_q n_q^2$.
Equations \eqref{main:eq:Xvariance}, \eqref{main:eq:Bvariance},
and \eqref{main:eq:YXvariance} now prove
$\norm{Q_L(q,t)-Q_L(q)}_0/\sqrt{L\chi_q}\le E_R(\lambda,t)$ and $1-c_{q,L}(t)\le E_R(\lambda,t)^2/2$ at every real time with
\begin{align}
E_R(\lambda,t)={}&|\lambda|s_q\sqrt{s_R}\notag\\
&+|t|\left[
|\lambda|s_q\frac{\sqrt{2R+1}}{R+1}
+\lambda^2s_q(a_T+|\lambda|a_N)\sqrt{3+16s_R}
+\lambda^2 n_q\right].
\label{main:eq:ER}
\end{align}
This estimates the full Hamiltonian, not a truncated Dyson series.

At fixed nonzero $\lambda$, fixed $R$, and fixed time, locality
gives the thermodynamic limits of the normalized extensive correlations.
One can express each as a spatial sum of two local correlations;
Lieb--Robinson approximation by a finite-support observable and
factorization of the infinite-temperature state give exponential tails
outside the light cone~\cite{bravyi}. The finite-volume bound is
uniform for $L\ge2R+3$ and therefore passes to this limit.
Only afterwards let $\lambda\to0$, choosing
$R=\lceil|\lambda|^{-1}\rceil$. For $0<|\lambda|\le1$,
the scaling $E_R=O(|\lambda|^{1/2})+O(|t|\,|\lambda|^{3/2})$ follows. The constants $s_q,n_q$ are bounded
uniformly on $\tr q^2=1$, a compact finite-dimensional sphere.
Hence $1-c_{q,\infty}\to0$ uniformly when
$|\lambda|^{3/2}|t_\lambda|\to0$.

For arbitrary real unit vectors $v,w$, let $q_v=\sum_av_aq_a$.
Cauchy--Schwarz gives
\begin{equation}
|v^{\mathsf T}(M_L-I_8)w|
\le\frac{\norm{Q_L(q_v,t)-Q_L(q_v)}_0}{\sqrt{L/3}}
\le\sup_{\tr q^2=1}E_R.
\end{equation}
This proves the full matrix claim in Theorem~\ref{main:thm:thermal}
for every intrinsic $T$, without a discrete symmetry assumption.
For Fermat, the real diagonal form \eqref{main:eq:memoryblocks} improves
the right-hand bound on $\norm{M_L-I_8}_{\rm op}$ to
$\tfrac12\sup_q E_R^2$. At $t=s_0/|\lambda|$ it is $O(|\lambda|)$.
The correction range diverges, the first-order remainder has not been
removed, and no smallness in operator norm is asserted. Thus the proof
does not solve the forbidden finite-ring repair equation.

\subsection{Low-frequency mass and leading response}

When $\sigma>0$, let $\nu_{S,L}$ be the spectral measure of
$F_{S,L}$ for the self-adjoint Liouvillian $[H_{0,L},\,\cdot\,]$,
normalized by $\norm{F_{S,L}}_0^2=L\sigma^2$. It is a positive
measure of total mass one. Projecting \eqref{main:eq:cutoffidentity}
onto $|\omega|\le\epsilon$ gives
\begin{equation}
\nu_{S,L}([-\epsilon,\epsilon])
\le\left(\frac{\epsilon}{2}\sqrt{s_R}
       +\frac{\sqrt{2R+1}}{R+1}\right)^2.
\label{main:eq:finitefrequency}
\end{equation}
The commutator contribution is bounded by
$\epsilon\norm{X_{R,L}}_0$ on this spectral subspace.
Fixed-time convergence of local correlations and continuity at zero
give a limiting probability measure $\nu_{S,\infty}$.
To pass the inequality to the closed interval, use a slightly wider
open interval with the same fixed $R$, apply weak convergence, and
then decrease its width. Choose $R=\lceil1/\epsilon\rceil$.
For $0<\epsilon\le1$, $s_R\le2/(3\epsilon)$ and
$(2R+1)/(R+1)^2\le2\epsilon$ prove
\eqref{main:eq:lowfrequency}. In particular,
$\nu_{S,\infty}(\{0\})=0$. This is a bound on the total nearby
mass and does not establish absolute continuity or a positive density.

The same identity bounds the leading fixed-time charge response. Define
\begin{equation}
A_{q,L}(t)=\frac{1}{2L\chi_q}
\left\lVert\int_0^t F_{S,L}(s)\,ds\right\rVert_0^2,
\qquad F_{S,L}(s)=e^{isH_{0,L}}F_{S,L}e^{-isH_{0,L}}.
\label{main:eq:leadingthermal}
\end{equation}
It is the coefficient of $\lambda^2$ in $1-c_{q,L}(t)$ at fixed
$L,t$. Integrating \eqref{main:eq:cutoffidentity} under $H_0$ gives
\begin{equation}
A_{q,\infty}(t)\le\frac{s_q^2}{2}
\left(\sqrt{s_R}+|t|\frac{\sqrt{2R+1}}{R+1}\right)^2.
\end{equation}
Taking $R=\lceil|t|\rceil$ for $|t|\ge1$ yields
\begin{equation}
A_{q,\infty}(t)\le
\left(\frac43+\frac2{\sqrt3}\right)s_q^2|t|.
\label{main:eq:lineargrowth}
\end{equation}
This excludes persistent superlinear asymptotic growth of the leading
coefficient. It does not decide whether $A_{q,\infty}(t)/t$ has a
positive limit, and it is not a uniform expansion at
coupling-dependent times. The full-dynamics conclusion instead uses
\eqref{main:eq:ER}.

\section{Complete Fermat charge memory and exact spectral moments}
\label{main:app:memory}

\subsection{An exact reduction of all eight directions}

Let $X,Z$ be the Weyl matrices in Sec.~\ref{main:sec:memory}. For
$W=Z,X,XZ,XZ^2$, use
\begin{equation}
q_{W,+}=\frac{W+W^\dagger}{\sqrt6},\qquad
q_{W,-}=\frac{W-W^\dagger}{i\sqrt6}.
\end{equation}
They form a real orthonormal basis with $\tr q^2=1$.
Both Fermat Hamiltonians commute with global $X$ and global $Z$;
these identities hold separately for $P,T,N$. On the complex traceless
Weyl basis, the eight characters of these adjoint actions are distinct.
Invariant correlations between different characters vanish. Passing to
Hermitian pairs initially allows a scalar rotation block for each
conjugate pair.

Let $C\ket a=\ket{-a\bmod3}$ and let $\mathcal R$ reflect the
ring. Each of $P,T,N$ is invariant under the combined transformation
$\mathcal R C^{\otimes L}$. The uniform operator built from $W$
maps to a phase times its adjoint. The phase cancels in its normalized
autocorrelation, which therefore equals its complex conjugate. The
antisymmetric part of each real block vanishes, leaving $c_WI_2$.
Finally the Hamiltonian is real in the color basis. Complex conjugation
maps $XZ$ to $XZ^2$ and reverses time; these autocorrelations are even,
so the two scalars agree. This proves \eqref{main:eq:memoryblocks}
at all times and all ring lengths.

\subsection{Local evaluation of the spectral moments}

For either Hermitian member $q$ of a channel, put
\begin{equation}
a=[h,q_1+q_2],\qquad
b=[h_{12},a_{12}]+[h_{12},a_{23}]+[h_{23},a_{12}].
\label{main:eq:momentdensities}
\end{equation}
The periodic sums are $[H,Q]$ and $[H,[H,Q]]$, respectively.
The scalar marginals of $h$ imply that both partial traces of $a$
vanish. For $\chi_q=\tr q^2/3$,
\begin{equation}
m_2=\frac{\tr(a^\dagger a)}{9\chi_q},\qquad
m_4=\frac{3\norm{b_1}_\HS^2+2\norm{b_2}_\HS^2+
                    \norm{b_3}_\HS^2}{27\chi_q}.
\label{main:eq:localmoments}
\end{equation}
Here $b_s$ are the span-$s$ components of the orthogonal representative
of $b$ modulo boundary differences. The operator-word argument of
Appendix~\ref{main:app:norm} proves these identities for every $L\ge5$.
It also shows that the moments are independent of volume in this range.

For clarity, the complete finite contractions can be specified by small
coefficient Gram matrices. Write
$b=\sum_{k=1}^4\lambda^k b^{(k)}$ and let $b_s^{(k)}$ denote
its projected span components. Define
\begin{equation}
G_{jk}=\frac{\sum_{s=1}^3(4-s)
  \tr[(b_s^{(j)})^\dagger b_s^{(k)}]}{27\chi_q}.
\label{main:eq:momentgram}
\end{equation}
Then $m_4=(\lambda,\lambda^2,\lambda^3,\lambda^4)
G(\lambda,\lambda^2,\lambda^3,\lambda^4)^{\mathsf T}$.
The matrix is real symmetric; its only possibly nonzero entries are
listed below. The input operators are the dyadic formulas for $P,T,N$
in Appendix~\ref{main:app:fermat}. Thus these contractions can be
evaluated using $E_{ab}E_{cd}=\delta_{bc}E_{ad}$,
$\tr E_{ab}=\delta_{ab}$, and $1+\omega+\omega^2=0$,
followed by the padding averages of Appendix~\ref{main:app:norm}:
\begin{center}
\setlength{\tabcolsep}{6pt}
\renewcommand{\arraystretch}{1.1}
\begin{tabular}{llrrrrrr}
\toprule
Path & Channel & $G_{11}$ & $G_{22}$ & $G_{33}$ & $G_{44}$ & $G_{13}$ & $G_{24}$\\
\midrule
Raw & $Z$ & 96 & 264 & 0 & 0 & 0 & 0\\
Raw & $X,XZ$ & 32 & 56 & 0 & 0 & 0 & 0\\
Compensated & $Z$ & 96 & 328 & 444 & 900 & $-8$ & 168\\
Compensated & $X$ & 32 & 312 & 708 & 4464 & 72 & 336\\
Compensated & $XZ$ & 32 & 248 & 636 & 3132 & 72 & 300\\
\bottomrule
\end{tabular}
\end{center}
The second moment follows by the same two-site multiplication, or from
the corresponding Weyl charge-Gram eigenvalue divided by three.
Combining the coefficients gives Table~\ref{main:tab:moments}.

\begin{table}[htbp]
\centering
\caption{Exact second and fourth spectral moments of complete charge
memory, with $u=\lambda^2$, for every periodic $L\ge5$.}
\label{main:tab:moments}
\setlength{\tabcolsep}{6pt}
\renewcommand{\arraystretch}{1.1}
\begin{tabular}{llll}
\toprule
Path & Channel & $m_2$ & $m_4$\\
\midrule
Raw & $Z$ & $12u$ & $24u(4+11u)$\\
Raw & $X,XZ$ & $4u$ & $8u(4+7u)$\\
Compensated & $Z$ & $12u(1+u)$ & $12u(8+26u+65u^2+75u^3)$\\
Compensated & $X$ & $4u(1+12u)$ & $4u(8+114u+345u^2+1116u^3)$\\
Compensated & $XZ$ & $4u(1+9u)$ & $4u(8+98u+309u^2+783u^3)$\\
\bottomrule
\end{tabular}
\end{table}

\subsection{Proof of simultaneous contraction}

For all real $x$,
$x^2/2-x^4/24\le1-\cos x\le x^2/2$. Integrating this
pointwise inequality against each positive spectral measure proves
\eqref{main:eq:momentbounds}, with no discarded Taylor remainder.
At the time $t_*$ defined in Proposition~\ref{main:prop:memoryloss}, the lower bound for the $XZ$
loss is exactly $d_*(u)$. Subtracting it from the corresponding $X$
and $Z$ lower bounds gives positive factors times, respectively,
\begin{align}
&32+408u+1197u^2+1701u^3,\\
&16+172u+354u^2+1008u^3-4374u^4.
\end{align}
Both are positive for $0\le u\le1/25$; for the second,
$4374u^4\le4374/25^4<16$. Thus all three channel losses
are at least $d_*(u)$. In this interval $t_*^2\le3/4$ and the
largest second moment is $12u(1+u)\le312/625$.
The upper bound in \eqref{main:eq:momentbounds} keeps all three
correlations positive. Their absolute values therefore equal the
correlations, proving the matrix-norm inequality in
Proposition~\ref{main:prop:memoryloss}. The current part follows from
\eqref{st:eq:timebound}; locality passes both bounds to the
thermodynamic limit at fixed strength.

\subsection{A common leading response and its scope}

At fixed $L,t$, spatial reflection changes $\lambda$ to $-\lambda$
for both Fermat paths and leaves every uniform one-site charge unchanged.
Their memories are consequently even in $\lambda$. Since $[H_0,Q]=0$,
\eqref{main:eq:chargeunitarity} gives
\begin{equation}
1-c_{W,L}(t)=\lambda^2 A_{W,L}(t)+O_{L,t}(\lambda^4),
\end{equation}
where $A_{W,L}$ is \eqref{main:eq:leadingthermal}. It depends on $T$
alone and is common to the raw and compensated paths. The Hermitian
force densities lie in the real intrinsic representation underlying
$\Sym^3(\C^3)$. This complex ten-dimensional representation is not
equivalent to its conjugate; its realification is irreducible.
The quadratic form \eqref{main:eq:leadingthermal} is real symmetric
and $SU(3)$ invariant, so it is proportional to the local norm squared.
The leading force norms in Table~\ref{main:tab:moments} give
\begin{equation}
A_{Z,L}=3A_{X,L}=3A_{XZ,L}.
\end{equation}
This normalized leading-response kernel is common to Hermitian intrinsic
directions. Equation~\eqref{main:eq:lineargrowth} controls its
thermodynamic growth. It neither identifies a full relaxation law nor
extends \eqref{main:eq:lowfrequency} to the compensated current force,
which is a different local operator.
\bibliography{references}
\end{document}